\documentclass[a4paper, 11pt]{article}
\usepackage{graphicx}
\usepackage{my_style}

\newcommand{\concat}{}
\newcommand{\AssocPOVM}[1]{\mathrm{E}^{#1}}
\newcommand{\Outk}{\mathsf{K}}   
\newcommand{\Outl}{\mathsf{L}}  
\newcommand{\Outs}{\mathsf{S}}  

\makeatletter
\newcommand{\circle@arrow}[2]{%
  \m@th
  \ooalign{%
    \hidewidth$#1\circ\mkern1mu$\hidewidth\cr
    $#1#2$\cr
  }%
}

\DeclareRobustCommand{\rcircarrow}{%
  \mathrel{\vphantom{\rightarrow}\mathpalette\circle@arrow\longrightarrow}%
}

\DeclareRobustCommand{\lcircarrow}{%
  \mathrel{\vphantom{\leftarrow}\mathpalette\circle@arrow\longleftarrow}%
}
\makeatother

\usepackage[
  backend=biber,        
  style=numeric-comp,   
  sorting=none,         
  giveninits=true,      
  doi=false,             
  url=false,             
  eprint=true,         
  maxbibnames=99,
  date=year
]{biblatex}
\DeclareFieldFormat*{title}{\normalfont#1}
\DeclareFieldFormat*{subtitle}{\normalfont#1}   
\DeclareFieldFormat*{booktitle}{\normalfont#1}
\DeclareFieldFormat*{maintitle}{\normalfont#1}
\DeclareFieldFormat*{citetitle}{\normalfont#1}
\DeclareFieldFormat*{journaltitle}{\normalfont#1}
\DeclareFieldFormat{pages}{\mkfirstpage{#1}}
\DeclareFieldFormat[article]{volume}{\mkbibbold{#1}}
\renewbibmacro{in:}{}
\DeclareBibliographyDriver{article}{%
  \usebibmacro{bibindex}%
  \usebibmacro{begentry}%
  \usebibmacro{author/translator+others}%
  \setunit{\labelnamepunct}\space
  \printfield[title]{title}\addcomma\space
  \href{\thefield{url}}{%
    \printfield{journaltitle}\space
    \printfield{volume},\space
    \printfield[pages]{pages}\space
    (\printfield{year})%
  }%
  \usebibmacro{finentry}%
}
\DeclareBibliographyDriver{misc}{%
  \usebibmacro{bibindex}%
  \usebibmacro{begentry}%
  \usebibmacro{author/translator+others}%
  \setunit{\labelnamepunct}\space
  \printfield[title]{title}\addcomma\space
  \href{\thefield{url}}{arXiv:\thefield{eprint}\space(\printfield{year})}%
  \usebibmacro{finentry}%
}

\DeclareBibliographyDriver{book}{%
  \usebibmacro{bibindex}%
  \usebibmacro{begentry}%
  \usebibmacro{author/editor+others}%
  \setunit{\labelnamepunct}\space
  \href{\thefield{url}}{\printfield[title]{title}}%
  \addcomma\space
  \printlist{publisher}%
  \setunit{\space}%
  \printtext{(\printfield{year})}%
  \usebibmacro{finentry}%
}

\DeclareBibliographyDriver{inproceedings}{%
  \usebibmacro{bibindex}%
  \usebibmacro{begentry}%
  \usebibmacro{author/translator+others}%
  \setunit{\labelnamepunct}\space
  \printfield[title]{title}\addcomma\space
  \iffieldundef{url}
    {%
      \printtext{Proc.\ }%
      \printfield{booktitle}%
      \setunit{\addcomma\space}%
      \printfield[pages]{pages}
      \setunit{\space}%
      \printtext{(\printfield{year})}%
    }{%
      \href{\thefield{url}}{%
        Proc.\ \printfield{booktitle}\space
        \printfield[pages]{pages}\space
        (\printfield{year})%
      }%
    }%
  \usebibmacro{finentry}%
}

\begin{document}

\title{
  Characterizing Space-Constrained Implementability of Quantum Instruments via Signaling Conditions
}

\author{
  Kosuke Matsui
  \affiliation{1}
   \affiliation{2}
  \email{kosuke.matsui@phys.s.u-tokyo.ac.jp}
  \and
  Jun-Yi Wu
  \affiliation{2}
   \affiliation{3}
  \email{junyiwuphysics@gmail.com}
  \and
  Hayata Yamasaki
   \affiliation{4}
  \email{hayata.yamasaki@gmail.com}
  Min-Hsiu Hsieh
   \affiliation{2}
  \email{min-hsiu.hsieh@foxconn.com}
  Mio Murao
   \affiliation{1}
  \email{murao@phys.s.u-tokyo.ac.jp}
  }

\address{1}{
  Department of Physics, Graduate School of Science, The University of Tokyo, Hongo 7-3-1, Bunkyo-ku, Tokyo, Japan
}
\address{2}{
  Hon Hai (Foxconn) Research Institute, Taipei, Taiwan
}
\address{3}{
  Department of Physics, Tamkang University, New Taipei 251301, Taiwan, ROC
}
\address{4}{
  Department of Computer Science, Graduate School of Information Science and Technology, The University of Tokyo, Hongo 7-3-1, Bunkyo-ku, Tokyo, Japan
}

\abstract{
  Scaling up the number of qubits available on quantum processors remains technically demanding even in the long term; it is therefore crucial to clarify the number of qubits required to implement a given quantum operation.
  For the most general class of quantum operations, known as quantum instruments, the qubit requirements are not well understood, especially when mid-circuit measurements and delayed input preparation are permitted.
  In this work, we characterize lower and upper bounds on the number of qubits required to implement a given quantum instrument in terms of the causal structure of the instrument.
  We further apply our results to entanglement distillation protocols based on stabilizer codes and show that, in these cases, the lower and upper bounds coincide, so the optimal qubit requirement is determined.
  In particular, we compute that the optimal number of qubits is 3 for the $[[9,1,3]]$-code-based protocol and 4 for the $[[5,1,3]]$-code-based protocol.
  }

\maketitle

\tableofcontents

\section{Introduction}\label{sec:intro}

The number of qubits available on current quantum processors is still insufficient to execute practically useful quantum operations~\cite{Bravyi2024}. However, scaling up the number of qubits is expected to remain technically challenging because error rates begin to rise once the system exceeds a certain size~\cite{2023google, mohseni2025buildquantumsupercomputerscaling, zhou2025}.
Consequently, to execute quantum operations with a limited number of qubits, it is crucial to clarify the number of qubits required to implement a given quantum operation.
In what follows, we use the term \emph{space} to denote the number of qubits that are simultaneously required to execute a quantum operation.

When the operation is unitary, analyzing the required space is straightforward, since the input system size (equivalently, the output system size) is necessary and sufficient for its implementation.
Indeed, the space must be large enough to hold the entire input or output state; conversely, a system of that size suffices to implement the unitary operation by decomposing it into elementary gates~\cite{nielsen2000}.
Broadening our scope, we next consider quantum operations that employ auxiliary systems and measurements, formalized as quantum instruments.
Since a quantum instrument admits a Stinespring dilation~\cite{Stinespring1955,nielsen2000}, i.e., a realization using a unitary operation together with an auxiliary system followed by a final projective measurement, the required space is upper-bounded by the size of the systems involved in the dilation.
Here, with a slight abuse of terminology, we use the term Stinespring dilation for instruments, originally a term for channels.
However, reasoning at the level of circuit compilation, the above upper bound may not be optimal.
Specifically, by performing mid-circuit measurements and reusing the measured qubits in subsequent operations, one can implement an instrument with less space than the upper bound suggested by its Stinespring dilation, as illustrated in \Cref{fig:example_delay}.
Here, by deferring the initialization of the auxiliary qubit until after the mid-circuit measurement, the number of qubits simultaneously used during the circuit execution is reduced.
Such mid-circuit measurements and qubit reuse are now feasible in multiple physical platforms~\cite{Iqbal2024, Ryan-Anderson2021realtimefaulttolerantqec, Gramham2023midcircuitmeasurements, Lis2023midcircuitoperations, PhysRevApplied.17.014014, Sivak_2023}.
As a further space-saving technique, we consider preparing only part of the input state at the beginning of a circuit and deferring the preparation of the remainder, as illustrated in \Cref{fig:example_no_delay}.
Here, the delayed inputs are loaded after a mid-circuit measurement, which further reduces the simultaneous qubit usage during the circuit execution.
The delayed-input technique is reasonable for algorithms whose input states are product states across qubits, such as entanglement distillation protocols~\cite{bennett1996purification, bennett1996mixedstateentanglement, matsumoto2003}.
It also applies when a quantum processor performs operations while communicating with other processors, for example, in distributed quantum computation~\cite{Eisert2000,Wu2023entanglement,Andres_Martinez_2024,Main_2025}. Even when the input to the local processor is entangled across qubits, inputs may be supplied sequentially by other processors, thereby enabling the delayed-input preparation.

\begin{figure}[h]
  \centering
  \begin{subfigure}{\textwidth}
    \centering
    \includegraphics[width=0.8\linewidth]{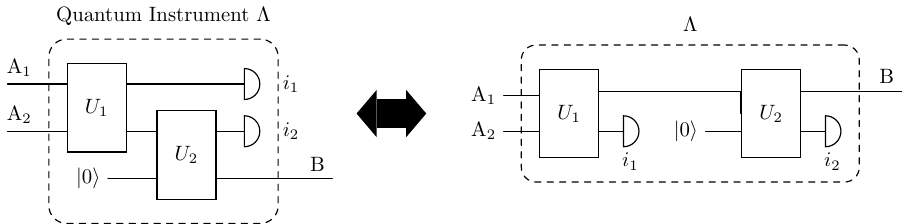}
    \caption{}\label{fig:example_delay}
  \end{subfigure}
  \vspace{2ex}
  \begin{subfigure}{\textwidth}
    \centering
    \includegraphics[width=0.8\linewidth]{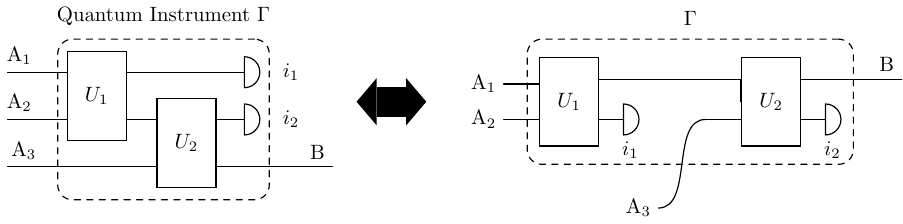}
    \caption{}\label{fig:example_no_delay}
  \end{subfigure}
  \caption{Two examples of reducing the number of qubits required to execute quantum instruments.
  Here, $\Lambda := \{\Lambda_{(i_1, i_2)}\}_{(i_1,i_2)}$ and $\Gamma := \{\Gamma_{(i_1, i_2)}\}_{(i_1,i_2)}$ are quantum instruments, and $U_1$ and $U_2$ represent two-qubit unitary operations.
  In each case, the quantum instrument shown on the left-hand side can be implemented with two qubits, as depicted on the right.
  (a) The quantum instrument $\Lambda$ has two input qubits $\mathrm{A}_1$ and $\mathrm{A}_2$, and employs one auxiliary qubit in its Stinespring dilation.
  The mid-circuit measurement is performed for outcome $i_1$, and the measured qubit is reused for the initialization of the auxiliary qubit.
  (b) The quantum instrument $\Gamma$ has three input qubits $\mathrm{A}_1$, $\mathrm{A}_2$, and $\mathrm{A}_3$. The input state of $\mathrm{A}_3$ is prepared after the circuit execution starts, and loaded after the mid-circuit measurement for outcome $i_1$.}\label{fig:example_space_reduction}
\end{figure}

Various studies have been conducted to investigate space requirements of quantum channels and quantum instruments.
For POVMs, it has been shown that a single auxiliary qubit suffices to implement them without delayed inputs~\cite{Lloyd2001engineeringquantumdynamics,Andersson2008binarysearchtrees,Ivashkov2024povmswithdynamiccircuits}.
Similar techniques have been applied to quantum channels, yielding analogous results~\cite{Shen2017quantumchannelconstruction}.
There is also a study that reduces the space requirements for POVMs not by allowing mid-circuit measurements, but instead by tolerating probabilistic success or depolarizing noise~\cite{kotowski2025prettygoodsimulationquantummeasurements}.
From a resource-theoretic framework, convertibility among families of instruments that does not require additional space has been analyzed~\cite{Ji2024}.
Several works focus on the space requirements of specific quantum algorithms, such as Shor's factoring algorithm~\cite{Griffiths1996, Parker2000, beauregard2003, zalka2006, haner2017, Gidney2021howtofactorbit, chevignard2025, gidney2025factor2048bitrsa}, entanglement distillation protocols~\cite{gidney2023entanglementpurification}, magic state distillation protocols~\cite{Litinski2019gameofsurfacecodes}, and reversible logic synthesis~\cite{bhattacharjee2019}.
These studies employ techniques tailored to each algorithm and demonstrate implementations that use less space than previous work.
In addition, a compilation method has been proposed that reduces space by exploiting circuit connectivity to identify opportunities for mid-circuit measurements~\cite{DeCross2023}.
However, these lines of work leave a fundamental open question:
when delayed inputs are allowed, how can we characterize, in general, the number of qubits required to implement a given quantum instrument?

In this work, we answer that question using signaling conditions, which characterize the causal relations between the input and output of a quantum instrument.
For quantum instruments composed of a unitary operation and projective measurements, we show that certain signaling conditions provide upper and lower bounds on the number of qubits required to implement the instrument with the aid of delayed inputs.
We also apply these results to entanglement distillation protocols based on stabilizer codes, showing that the upper and lower bounds coincide in these cases. For several well-known stabilizer codes, we compute the optimal number of qubits for implementing the instruments used in the corresponding entanglement distillation protocols.

The remainder of this paper is organized as follows.
\Cref{sec:notation} introduces the basic notation.
\Cref{sec:def_space} defines classes of quantum instruments implementable under space constraints.
\Cref{sec:tools} presents analytical tools for studying the space requirements of quantum instruments:
\Cref{sec:pp_inst} covers the composability of quantum instruments and
\Cref{sec:signaling} covers the outcome no-signaling condition.
\Cref{sec:qubit_reduction} presents our main results: lower and upper bounds on the number of qubits required to implement a given quantum instrument under space constraints.
Finally, \Cref{sec:ent_distill} applies our results to entanglement distillation protocols.

\section{Notation and Preliminaries}\label{sec:notation}

We use $\mathbb{N}=\{0,1,2,\ldots\}, \mathbb{Z}_{>0}=\{1,2,\ldots\}$, and $[n]=\{1,2,\ldots,n\}$ for $n\in\mathbb{Z}_{>0}$.
For binary strings $u,v\in\{0,1\}^\ast$, write $u\concat v$ for concatenation.
Hilbert spaces are denoted by $\mathcal{H}$. For a Hilbert space $\mathcal{H}$, let $\mathcal{L}(\mathcal{H})$ denote the space of linear operators on $\mathcal{H}$.

Quantum systems are described by Hilbert spaces; in this work we focus on qubit systems, i.e., $\mathcal{H} \cong (\mathbb{C}^2)^{\otimes n}$ for some $n \in \mathbb{N}$.
Quantum states are described by density operators $\rho \in \mathcal{L}(\mathcal{H})$ that are positive semidefinite with unit trace.
A quantum channel (a deterministic transformation of quantum states) is described by a completely positive, trace-preserving (CPTP) map $\mathcal{E}$.
A quantum instrument (a probabilistic transformation) is described by a set $\{\Lambda_k\}_{k \in \Outk}$ of quantum operations (completely positive, trace-nonincreasing maps) such that $\sum_{k \in \Outk} \Lambda_k$ is trace-preserving. Throughout, we take the outcome set to be $\Outk=\{0,1\}^T$ for some $T \in \mathbb{N}$.
A positive operator-valued measure (POVM) is a special case of a quantum instrument that has only classical outcomes, described by a set $\{E_k\}_{k \in \Outk}$ of positive semidefinite operators satisfying $\sum_{k \in \Outk} E_k = \mathbb{I}$.
A projective measurement is a POVM whose elements are projectors $\{P_k\}_k$ with $P_k^2 = P_k = P_k^\dagger$ and $\sum_k P_k = \mathbb{I}$.

\section{Definitions of Quantum Instruments Implementable under Space Constraints}\label{sec:def_space}
To determine whether a given quantum instrument is implementable under space constraints, we must first formalize the notion of space-constrained implementability, namely, by defining the class of instruments implementable under space constraints.
In this section, we provide two definitions of space-constrained implementable instruments: one for the setting that allows delayed inputs and one for the setting that does not, which are previewed in \Cref{sec:intro}.

To define space-constrained implementable instruments, it is not sufficient to consider the system size involved in its Stinespring dilation; one must examine the decomposability of a circuit-level compilation.
For instance, although the quantum instruments $\Lambda$ and $\Gamma$ in \Cref{fig:example_space_reduction} employ three qubits in their Stinespring dilations, they can be implemented with two qubits by decomposing into a sequence of operations, each of which uses only two qubits.
These observations suggest that space-constrained implementable instruments should be formalized in terms of whether the instrument can be executed as a sequential composition of building-block operations.
Accordingly, in what follows, we first specify the elementary operations admissible under a given space constraint, and then define space-constrained implementable quantum instruments as any instrument obtainable as their composition.

\subsection{Definition for the Setting without Delayed Inputs}\label{sec:def_space_wodi}
In this section, we examine the setting where the entire input state must be present at the beginning; the measured qubits are therefore always reinitialized to a fixed state such as $\ket{0}$.

In this study, we assume the following elementary operations can be performed under an $m$-qubit space constraint in the setting without delayed inputs:
\begin{assumption}[Elementary Operation Set (without Delayed Inputs)]\label{assm:ops_without_delayed}
  Fix the number of the available qubits $m \in \mathbb{N}$. We assume that the following operations can be performed under an $m$-qubit space constraint in the setting without delayed inputs:
  \begin{enumerate}[label=(\alph*)]
    \item Unitary operation on the $m$ qubits that depends on the classical value available at that time.
    \item Computational basis measurements on a subset of the $m$ qubits. The choice of measured qubits can depend on the classical value available at that time.
    \item  Reset a subset of the $m$ qubits to $\ket{0}$s. The choice of reset qubits can depend on the classical value available at that time.
    \item Classical processing on the classical value available at that time.
  \end{enumerate}
\end{assumption}

The operations in \Cref{assm:ops_without_delayed} are formally expressed as follows.
Let $\mathcal H_{\mathrm m} \cong (\mathbb C^2)^{\otimes m}$ denote the system of the $m$ available qubits.
The system before and after each operation is expressed by a set $\{\rho_{k}\}_{k \in \Outk}$, where $\rho_{k} \in \mathcal{L}(\mathcal{H}_{\mathrm{m}})$ is the unnormalized state when the classical value $k$ is obtained,
and the set $\{\rho_{k}\}_{k \in \Outk}$ is then updated as follows, according to the rule for each operation:

\smallskip
\begin{minipage}[t]{0.65\columnwidth}
  \vspace{0pt}
  \noindent\textbf{Unitary operation:} Let $U_{k} \in \mathcal{L}(\mathcal H_{\mathrm{m}})$ be the unitary operator applied when the classical value is $k$. Then,
  \begin{align}
  \{ \rho_{k} \}_{k \in \Outk} \mapsto \{ U_{k} \,\rho_{k}\, U_{k}^\dagger \}_{k \in \Outk}.
  \end{align}
  \vfill
\end{minipage}
\hfill
\begin{minipage}[t]{0.30\columnwidth}
  \vspace{3pt}
  \centering
  \includegraphics[width=0.8\linewidth]{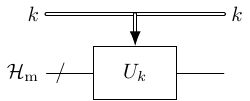}
  \captionof{figure}{Unitary operation.}
  \label{fig:elementary_u}
\end{minipage}

\vspace{0.5cm}
  \begin{minipage}[t]{0.65\columnwidth}
    \noindent\textbf{Computational basis measurement:} Let $S : = S(k) $ be a subset of the $m$ available qubits measured in the computational basis. Then,
    \begin{align}
    \{ \rho_{k} \}_{k \in \Outk} \mapsto \{ M_x \,\rho_{k}\, M_x^\dagger \}_{k x \in \Outk \times \{0,1\}^{|S|}},
    \end{align}
    where $M_x := \ketbra{0}{x}_S \otimes \mathbb I$ is the measurement operator corresponding to the outcome $x \in \{0,1\}^{|S|}$, and $k x$ denotes the concatenation of the binary strings $k$ and $x$, and $\Outk \times \{0,1\}^{|S|} := \{ k x : k \in \Outk, x \in \{0,1\}^{|S|} \}$.
  \end{minipage}
  \hfill
  \begin{minipage}[t]{0.30\columnwidth}
    \vspace{10pt}
    \centering
    \includegraphics[width=0.95\linewidth]{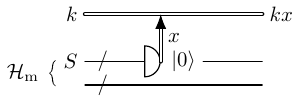}
    \captionof{figure}{Computational basis measurement.}
    \label{fig:elementary_m}
  \end{minipage}

\vspace{0.5cm}
\begin{minipage}[t]{0.65\columnwidth}
  \vspace{0pt}
  \noindent\textbf{State Reset:} Let $S : = S(k) $ be a subset of the $m$ available qubits to be reset to $\ket{0}$. Then,
  \begin{align}
  \{ \rho_{k} \}_{k \in \Outk} \mapsto \qty{ \sum_{x \in \{0,1\}^{|S|}} M_x \rho_{k} M_x^\dagger }_{k \in \Outk}.
  \end{align}
  Since the state-reset operation can be written as a computational basis measurement followed by a classical processing that forgets the measurement outcome, we may omit the state-reset operation from the set of elementary operations without loss of generality in what follows.
\end{minipage}
\hfill
\begin{minipage}[t]{0.30\columnwidth}
  \vspace{20pt}
  \centering
  \includegraphics[width=0.95\linewidth]{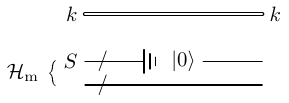}
  \captionof{figure}{State reset.}
  \label{fig:elementary_r}
\end{minipage}

\vspace{0.5cm}
\begin{minipage}[t]{0.65\columnwidth}
  \vspace{0pt}
  \noindent\textbf{Classical processing:}
  Let $f: \Outk \to \Outk'$ be the function used to update the classical value.
  Then,
  \begin{align}
  \{ \rho_{k} \}_{k \in \Outk} \mapsto \qty{ \sum_{k \in f^{-1}(k')} \,\rho_{k} }_{k' \in \Outk'},
  \end{align}
  where $f^{-1}(k') := \{ k \in \Outk : f(k) = k' \}$ is the preimage of $k'$ by $f$, so $f$ may not be injective.
  Here, the updated state is given by the sum of original $\rho_k$ over all $k$ that could have been mapped to $k'$. This means that the classical value $k$ is updated to $k' = f(k)$, and the subsequent operations have access only to $k'$ and do not know which $k$ was mapped to $k'$.
\end{minipage}
\hfill
\begin{minipage}[t]{0.30\columnwidth}
  \vspace{10pt}
  \centering
  \includegraphics[width=0.95\linewidth]{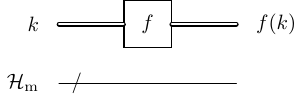}
  \captionof{figure}{Classical processing.}
  \label{fig:elementary_c}
\end{minipage}

\bigskip
These transformations can be uniformly described using quantum instrument formalism:
Let $\{ \rho_{k} \}_{k \in \Outk}, \{ \sigma_{k'} \}_{k' \in \Outk'}$ be the sets of unnormalized states before and after the transformations, respectively, and define a quantum instrument $\{ \Phi_{k'\mid k} \}_{k' \in \Outk'}$ for each $k \in \Outk$ by
\begin{align}
\Phi_{k' \mid k}(\rho)
:=
\begin{cases}
\delta_{k', k} U_{k} \,\rho\, U_{k}^\dagger
& \text{(Unitary operation),} \\
\sum_{x \in \{0,1\}^{|S|}} \delta_{k', k \concat x} \,M_x \,\rho\, M_x^\dagger
& \text{(Computational basis measurement),} \\
\delta_{k', f(k)} \, \rho
& \text{(Classical processing)}.
\end{cases}
\label{eq:elementary_instr_def}
\end{align}
Then, the updated states $\sigma_{k'}$ can be written as
\begin{align}
\sigma_{k'} = \sum_{k \in \Outk} \Phi_{k' \mid k}(\rho_{k})
\qquad \forall \, k' \in \Outk'.
\label{eq:elementary_instr_update}
\end{align}
The equality can be verified by substituting \Cref{eq:elementary_instr_def} into \Cref{eq:elementary_instr_update}, and comparing it with the description of each operation above.

Employing the notation above, we now give a precise definition of the space-constrained implementable quantum instruments as compositions of the operations specified in \Cref{assm:ops_without_delayed}.
\begin{definition}[$m$-qubit Implementable Instruments (without Delayed Inputs)]\label{def:space_instr_wo_delayed}
  Fix the number of the available qubits $m \in \mathbb{N}$, and let $\mathcal H_{\mathrm m} \cong (\mathbb C^2)^{\otimes m}$ denote the system of the available qubits.
  Let
  $
  \{\Lambda_k : \mathcal{L} (\mathcal H_{\mathrm{in}}) \to \mathcal{L}(\mathcal H_{\mathrm{out}}) \}_{k \in \Outk}
  $
  be a quantum instrument where $\mathcal H_{\mathrm{in}} \cong (\mathbb C^2)^{\otimes n_\mathrm{in}}, \mathcal H_{\mathrm{out}} \cong (\mathbb C^2)^{\otimes n_\mathrm{out}}$ for $n_{\mathrm{in}},n_{\mathrm{out}} \in \{0,1,\ldots,m\}$.

  The instrument $\{\Lambda_k \}_{k \in \Outk}$ is $m$-qubit implementable (without delayed inputs)
  if each $\Lambda_k$ can be written as follows, which is also illustrated in \Cref{fig:def_m_space_instr_wodi}:
  \begin{multline}
  \Lambda_{k}(\rho)
  =
  \Tr_{(\mathbb{C}^2)^{\otimes (m-n_\mathrm{out})}}
  \qty[
  \qty(
    \displaystyle\sum_{ k_{1} \in \Outk_{1}\; \cdots \; k_{T-1} \in \Outk_{T-1} }
    \Phi^{(T)}_{k\mid k_{T-1}} \circ \cdots \circ \Phi^{(1)}_{k_1\mid k_0}
  )
  \qty(\rho \otimes \ketbra{0}^{\otimes (m-n_{\mathrm{in}})})]\\
  \quad \forall \rho \in \mathcal{L}(\mathcal H_{\mathrm{in}}), \label{eq:def_space_instr_wo_delayed}
  \end{multline}
  where, for each round $t\in \{1,2, \cdots, T\}$ and each previously obtained classical value $k_{t-1} \in \Outk_{t-1}$, the quantum instrument $\{ \Phi^{(t)}_{k_t\mid k_{t-1}} : \mathcal{L} (\mathcal H_{\mathrm{m}}) \to \mathcal{L} (\mathcal H_{\mathrm{m}})\}_{k_t \in \Outk_t}$
  \footnote{At the first round we identify $\mathcal{H}_{\mathrm{m}}$ with $\mathcal{H}_{\mathrm{in}} \otimes (\mathbb{C}^2)^{\otimes (m-n_{\mathrm{in}})}$, and at the last round we identify $\mathcal{H}_{\mathrm{m}}$ with $\mathcal{H}_{\mathrm{out}} \otimes (\mathbb{C}^2)^{\otimes (m-n_{\mathrm{out}})}$} is given as one of the following:

    \smallskip
   \noindent\textbf{Unitary operation:} For all $ \rho \in \mathcal{L}(\mathcal{H}_{\mathrm{m}})$,
   \begin{align}
    \Phi^{(t)}_{k_t \mid k_{t-1}}(\rho)
    = \delta_{k_t, k_{t-1}} U_{k_{t-1}} \,\rho\, U_{k_{t-1}}^\dagger
    \qquad k_t \in \Outk_t := \Outk_{t-1}, \label{eq:unitary_op}
   \end{align}
   where $U_{k_{t-1}}$ is a unitary operator on $\mathcal H_{\mathrm m}$.

    \smallskip
   \noindent\textbf{Computational basis measurement:} For all $ \rho \in \mathcal{L}(\mathcal{H}_{\mathrm{m}})$,
    \begin{align}
    \Phi^{(t)}_{k_t \mid k_{t-1}}(\rho)
    = \sum_{x \in \{0,1\}^{|S|}} \delta_{k_t, k_{t-1} \concat x} \,
    M_x \,\rho\, M_x^\dagger
    \qquad k_t \in \Outk_t := \Outk_{t-1} \times \{0,1\}^{|S|}, \label{eq:measurement_op}
    \end{align}
    where $S : = S(k) $ is a subset of the $m$ available qubits and $M_x := \ketbra{0}{x}_S \otimes \mathbb I$.

    \smallskip
    \noindent\textbf{Classical processing:} For all $ \rho \in \mathcal{L}(\mathcal{H}_{\mathrm{m}})$,
    \begin{align}
      \Phi^{(t)}_{k_t \mid k_{t-1}}&(\rho)
      = \delta_{k_t, f(k_{t-1})} \, \rho \label{eq:classical_op}
     \qquad k_t \in \Outk_t,
    \end{align}
    where $f: \Outk_{t-1} \to \Outk_{t}$ is a function on classical values.

    Here, $k_0 \in \Outk_0 = \{0\}$ is a fixed initial classical value.

    \begin{figure}[H]
      \centering
      \includegraphics[width=0.9\linewidth]{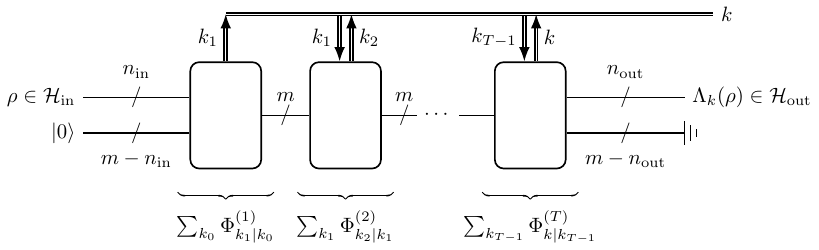}
      \caption{Definition of an $m$-qubit implementable instrument (without delayed inputs), as given in \Cref{eq:def_space_instr_wo_delayed}. A quantum instrument that admits this decomposition is called an $m$-qubit implementable instrument (without delayed inputs). Each instrument $\{\Phi^{(t)}_{k_t\mid k_{t-1}}\}_{k_t \in \Outk_t}$ is one of the following: (a) a unitary operation; (b) a computational basis measurement; (c) a classical processing.}
      \label{fig:def_m_space_instr_wodi}
    \end{figure}

\end{definition}

\paragraph{Notes.}
\begin{itemize}
  \item For notational convenience, \Cref{def:space_instr_wo_delayed} considers only the case where $n_{\mathrm{in}}, n_{\mathrm{out}} \leq m$. If necessary, we can additionally define that quantum instruments with either $n_{\mathrm{in}} > m$ or $n_{\mathrm{out}} > m$ are not $m$-qubit implementable (without delayed inputs), which is a reasonable definition since either the input state or the output state cannot be held in $m$ qubits in that case.
  \item In \Cref{def:space_instr_wo_delayed}, the final-round outcome set $\Outk_T$ may be larger than the original outcome set $\Outk$ in order to make $\{ \Phi^{(T)}_{k_T \mid k_{T-1}} \}_{k_T \in \Outk_T}$ trace-preserving on its input space. The probability of obtaining any additional outcome $k \in \Outk_T \setminus \Outk$ is required to be zero by the trace-preserving property of $\{\Lambda_k\}_{k \in \Outk}$.
\end{itemize}

\subsection{Definition for the Setting with Delayed Inputs}\label{sec:def_space_di}
In this section, we consider the setting where parts of the input state can be prepared in a delayed manner.
As in the previous section, we first define quantum instruments implementable under space constraints in this setting as compositions of elementary operations.
Here, we add the delayed-input loading operation to the elementary operations.
\begin{assumption}[Elementary Operation Set (with Delayed Inputs)]\label{assm:ops_with_delayed}
  Fix the number of the available qubits $m \in \mathbb{N}$. We assume that the following operations can be performed under an $m$-qubit space constraint in the setting with delayed inputs:
  \begin{enumerate}[label=(\alph*)]
    \item Unitary operation on the $m$ qubits that depends on the classical value available at that time. \label{item:unitary_op_wdi}
    \item Computational basis measurements on a subset of the $m$ qubits. The choice of measured qubits can depend on the classical value available at that time. \label{item:measurement_op_wdi}
    \item Reset a subset of the $m$ qubits to $\ket{0}$s. The choice of reset qubits can depend on the classical value available at that time. \label{item:reset_op_wdi}
    \item Classical processing on the classical value available at that time. \label{item:classical_proc_op_wdi}
    \item Input-loading operation: Measure in a subset of the $m$ qubits and loading part of the input state to the measured qubits. The choice of measured qubits can depend on the classical value available at that time. (See \Cref{fig:elementary_l} for an illustration.)\label{item:delayed_input_loading_op_wdi}
  \end{enumerate}
\end{assumption}

\begin{remark}\label{rem:indep_loading}
  We assume that each qubit's input may be prepared independently at any time after the circuit execution starts.
  For example, we do not consider restrictions in which input states in two particular subsystems must be prepared together, or in which input preparation must follow a specific order.
\end{remark}
\begin{remark}\label{rem:fixed_order_loading}
  We assume that the order in which input states are prepared is determined before the circuit execution starts instead of dynamically determined by the classical values during the execution. Allowing the dynamical ordering is left for future work.
\end{remark}

When giving formal descriptions of the operations in \Cref{assm:ops_with_delayed}, special consideration is needed for the input system.
Since the input state does not necessarily reside in the available qubit system $\mathcal{H}_\mathrm{m}$ when the circuit begins, we need to introduce a notional system $\mathcal{H}_\mathrm{in}$ to hold the input state, as illustrated in \Cref{fig:def_m_space_instr}.
This system $\mathcal{H}_\mathrm{in}$ is not counted toward the space cost.
Taking $\mathcal{H}_\mathrm{in}$ into account, the elementary operations in \Cref{assm:ops_with_delayed} can be described as a transformation on sets of unnormalized states $\{\rho_k\}_{k\in \Outk}$ in $\mathcal{H}_\mathrm{m} \otimes \mathcal{H}_\mathrm{in}$.
Operations~\ref{item:unitary_op_wdi} to~\ref{item:classical_proc_op_wdi} in \Cref{assm:ops_with_delayed} can be expressed similarly to those in the setting without delayed inputs, except for the identity operation on $\mathcal{H}_\mathrm{in}$.
The input-loading operation \ref{item:delayed_input_loading_op_wdi} is described as follows:

\begin{minipage}[t]{0.65\columnwidth}
  \vspace{0pt}
  Let $S : = S(k) $ be a subset of the $m$ available qubits, and  $J$ be a subset of the unloaded input qubits satisfying $\dim \mathcal{H}_J = \dim \mathcal{H}_S$. The set $\{\rho_k\}_{k \in \Outk}$ is then updated as
  \begin{align}
    \{ \rho_k \}_{k \in \Outk}
    \mapsto
    \qty{ \qty(\bra{x}_S \otimes \mathbb{I}) \,\rho_k\, \qty(\ket{x}_S \otimes \mathbb{I}) }_{kx \in \Outk \times \{0,1\}^{|S|}} \,\, ,
  \end{align}
  and, thereafter, the system labels are updated to
  \begin{align}
    \mathcal{H}_\mathrm{m} := \mathcal{H}_{S^{\mathrm{c}}} \otimes \mathcal{H}_{J}, \qquad
    \mathcal{H}_\mathrm{in} := \mathcal{H}_{J^{\mathrm{c}}}, \label{eq:label_update_wdi}
  \end{align}
  where $S^{\mathrm{c}}$ is the complement of $S$ in the available qubits, and $J^{\mathrm{c}}$ is the complement of $J$ in the unloaded input qubits.
  Note that $S$ can depend on the classical value $k$, whereas $J$ is fixed before the circuit execution starts, as mentioned in \Cref{rem:fixed_order_loading}.
\end{minipage}
\hfill
\begin{minipage}[t]{0.30\columnwidth}
  \vspace{6pt}
  \centering
  \includegraphics[width=0.95\linewidth]{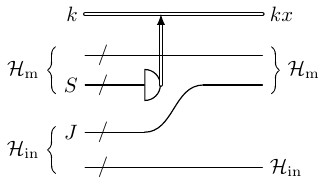}
  \captionof{figure}{Input-loading operation.}
  \label{fig:elementary_l}
\end{minipage}

\medskip
Omitting the state-reset operation \ref{item:reset_op_wdi} as in the setting without delayed inputs, we now give a definition of space-constrained implementable instruments in the setting with delayed inputs as a composition of the elementary operations in \Cref{assm:ops_with_delayed}.

\begin{definition}[$m$-Qubit Implementable Instruments (with Delayed Inputs)]\label{def:m_space_instr_wdi}
  Fix $m \in \mathbb{N}$, and let $\mathcal H_{\mathrm m} \cong (\mathbb C^2)^{\otimes m}$.
  Let
  $
  \{\Lambda_k : \mathcal{L} (\mathcal H_{\mathrm{in}}) \to \mathcal{L}(\mathcal H_{\mathrm{out}}) \}_{k \in \Outk}
  $
  be a quantum instrument where $\mathcal H_{\mathrm{in}} \cong (\mathbb C^2)^{\otimes n_{\mathrm{in}}}$ and
  $\mathcal H_{\mathrm{out}} \cong (\mathbb C^2)^{\otimes n_{\mathrm{out}}}$ for $n_{\mathrm{in}},n_{\mathrm{out}}\in \mathbb{Z}_{\geq 0}$.

  The quantum instrument $\{\Lambda_k \}_{k \in \Outk}$ is $m$-qubit implementable (with delayed inputs)
  if each $\Lambda_k$ can be written as follows, which is also illustrated in \Cref{fig:def_m_space_instr}:
  \begin{multline}
  \Lambda_{k}(\rho_{\mathrm{in}})
  =
  \Tr_{(\mathbb{C}^2)^{\otimes (m-n_\mathrm{out})}}
  \qty[
  \qty(
    \sum_{k_1 \in \Outk_1, \cdots, k_{T-1} \in \Outk_{T-1}}
    \Phi^{(T)}_{k\mid k_{T-1}}
    \circ\cdots\circ
    \Phi^{(1)}_{k_1\mid k_0}
  )
  \qty( \ketbra{0}^{\otimes m} \otimes \rho_{\mathrm{in}})
  ]\\
  \quad \forall \rho_{\mathrm{in}} \in \mathcal{L}(\mathcal H_{\mathrm{in}}),
  \label{eq:def_space_instr_delayed}
  \end{multline}
  where, for each round $t\in \{1,2, \cdots, T\}$ and each previously obtained classical value $k_{t-1} \in \Outk_{t-1}$, the quantum instrument $\{ \Phi^{(t)}_{k_t\mid k_{t-1}} : \mathcal{L} (\mathcal H_{\mathrm{m}} \otimes \mathcal H_{\mathrm{in}}) \to \mathcal{L} (\mathcal H_{\mathrm{m}} \otimes \mathcal H_{\mathrm{in}})\}_{k_t \in \Outk_{t}}$ is one of the following:

  \smallskip
  \noindent\textbf{Unitary operation.} For all $ \rho \in \mathcal{L}(\mathcal{H}_{\mathrm{m}} \otimes \mathcal{H}_{\mathrm{in}})$,
  \begin{align}
    \Phi^{(t)}_{k_t \mid k_{t-1}}(\rho)
    = \delta_{k_t, k_{t-1}} (U_{k_{t-1}} \otimes \mathbb{I}_\mathrm{in})\,\rho\, (U_{k_{t-1}} \otimes \mathbb{I}_\mathrm{in})^\dagger
    \qquad \forall \, k_t \in \Outk_t := \Outk_{t-1},
    \label{eq:def_elementary_u_wdi}
  \end{align}
  where $U_{k_{t-1}}$ is a unitary operator on $\mathcal H_{\mathrm m}$.

  \smallskip
  \noindent\textbf{Computational basis measurement.} For all $ \rho \in \mathcal{L}(\mathcal{H}_{\mathrm{m}} \otimes \mathcal{H}_{\mathrm{in}})$,
  \begin{align}
    \Phi^{(t)}_{k_t \mid k_{t-1}}(\rho)
    = \sum_{x \in \{0,1\}^{|S|}} \delta_{k_t, k_{t-1} \concat x} \,
    (M_x \otimes \mathbb{I}_{\mathrm{in}}) \,\rho\, (M_x \otimes \mathbb{I}_{\mathrm{in}})^\dagger
    \qquad \forall \, k_t \in \Outk_t := \Outk_{t-1} \times \{0,1\}^{|S|},
    \label{eq:def_elementary_m_wdi}
  \end{align}
  where $S : = S(k_{t-1}) $ is a subset of the $m$ available qubits and $M_x := \ketbra{0}{x}_S \otimes \mathbb I$.

  \smallskip
  \noindent\textbf{Classical processing.} For all $ \rho \in \mathcal{L}(\mathcal{H}_{\mathrm{m}} \otimes \mathcal{H}_{\mathrm{in}})$,
  \begin{align}
    \Phi^{(t)}_{k_t \mid k_{t-1}}&(\rho)
    = \delta_{k_t, f(k_{t-1})} \, \rho
    \qquad \forall \, k_t \in \Outk_t,
    \label{eq:def_elementary_c_wdi}
  \end{align}
  where $f: \Outk_{t-1} \to \Outk_t$ is a function on classical values.

  \smallskip
  \noindent\textbf{Input-loading operation.} Let $S : = S(k_{t-1}) $ be a subset of the $m$ available qubits, and $J$ be a subset of the unloaded input qubits satisfying $\dim \mathcal{H}_J = \dim \mathcal{H}_S$. For all $ \rho \in \mathcal{L}(\mathcal{H}_{\mathrm{m}} \otimes \mathcal{H}_{\mathrm{in}})$,
  \begin{align}
    \Phi^{(t)}_{k_t \mid k_{t-1}}(\rho)
    =
    \sum_{x \in \{0,1\}^{|S|}}
    \delta_{k_t, k_{t-1} x}
    \qty(\bra{x}_S \otimes \mathbb{I}) \,\rho\, \qty(\ket{x}_S \otimes \mathbb{I})
    \quad
    k_t \in \Outk_t := \Outk_{t-1} \times \{0,1\}^{|S|},
    \label{eq:def_elementary_l_wdi}
  \end{align}
  and redefine
  $
    \mathcal{H}_\mathrm{m} := \mathcal{H}_{S^{\mathrm{c}}} \otimes \mathcal{H}_{J},
  $ and
  $
    \mathcal{H}_\mathrm{in} := \mathcal{H}_{J^{\mathrm{c}}}
  $.

  Here, $k_0 \in \Outk_0 = \{0\}$ is a fixed initial classical value.
    \begin{figure}[H]
    \centering
    \includegraphics[width=0.8\linewidth]{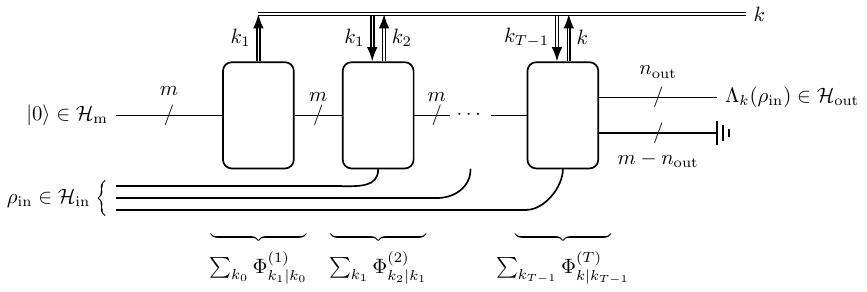}
    \caption{Definition of an $m$-qubit implementable instrument (with delayed inputs), as given in \Cref{eq:def_space_instr_delayed}. A quantum instrument that admits this decomposition is called an $m$-qubit implementable instrument (with delayed inputs). Here, $\mathcal{H}_\mathrm{m}$ is the system of the $m$ available qubits, and $\mathcal{H}_\mathrm{in}$ is the notional system that holds the input state. Each instrument $\{\Phi^{(t)}_{k_t\mid k_{t-1}}\}_{k_t \in \Outk_t}$ is one of the following: (a) a unitary operation; (b) a computational basis measurement; (c) a classical processing; (d) an input-loading operation.}
    \label{fig:def_m_space_instr}
    \end{figure}
\end{definition}

\begin{remark}\label{rem:stairs_wdi}
As an immediate consequence of \Cref{def:m_space_instr_wdi}, a quantum instrument $\Lambda := \{\Lambda_k\}_{k \in \Outk}$ is an $m$-qubit implementable instrument (with delayed inputs) if and only if $\Lambda$ can be expressed in the form illustrated in \Cref{fig:stairs_m_space_instr_wdi}, formally written as follows:
\begin{multline}
    \Lambda_k (\rho_\mathrm{in}) =
    \Tr_{(\mathbb{C}^2)^{\otimes (m-n_\mathrm{out})}}
    \qty[
      \qty(
      \sum_{k_0, \cdots, k_{{n_\mathrm{in}}} }
      \widetilde\Gamma^{({n_\mathrm{in} +1})}_{k\mid k_{{n_\mathrm{in}}}}
      \circ\cdots\circ
      \widetilde\Gamma^{(1)}_{k_1\mid k_0}
    )
    \qty(\ketbra{0}^{\otimes m} \otimes \rho_{\mathrm{in}} )
    ]\\
    \quad \forall \rho_{\mathrm{in}} \in \mathcal{L}(\mathcal H_{\mathrm{in}}), \label{eq:stairs_wdi_1}
  \end{multline}
  where each $\{\widetilde{\Gamma}^{(t)}_{k_t \mid k_{t-1}} \}_{k_t \in \Outk_{t}}$ is a quantum instrument that factors into the $m$-qubit implementable instrument (without delayed inputs) $\{ \Gamma^{(t)}_{k_t \mid k_{t-1}} \}_{k_t \in \Outk_{t}}$ tensored with the identity operation on $\text{A}_{t}, \cdots, \text{A}_{n_\mathrm{in}}$:
  \begin{align}
    \widetilde{\Gamma}^{(t)}_{k_t \mid k_{t-1}} :=
    \Gamma^{(t)}_{k_t \mid k_{t-1}} \otimes \operatorname{id}_{\mathrm{A}_{t}, \cdots, \mathrm{A}_{n_\mathrm{in}}}. \label{eq:stairs_wdi_2}
  \end{align}
  Note that, for $1 \leq t \leq n_\mathrm{in}$, the instrument $\{\Gamma^{(t)}_{k_t \mid k_{t-1}} \}_{k_t \in \Outk_{t}}$ has the input system $(\mathbb{C}^2)^{\otimes m}$ and the output system $(\mathbb{C}^2)^{\otimes (m-1)}$, whereas for $t = n_\mathrm{in}+1$, it has both input and output systems $(\mathbb{C}^2)^{\otimes m}$.

\begin{figure}[H]
\centering
\includegraphics[width=0.99\linewidth]{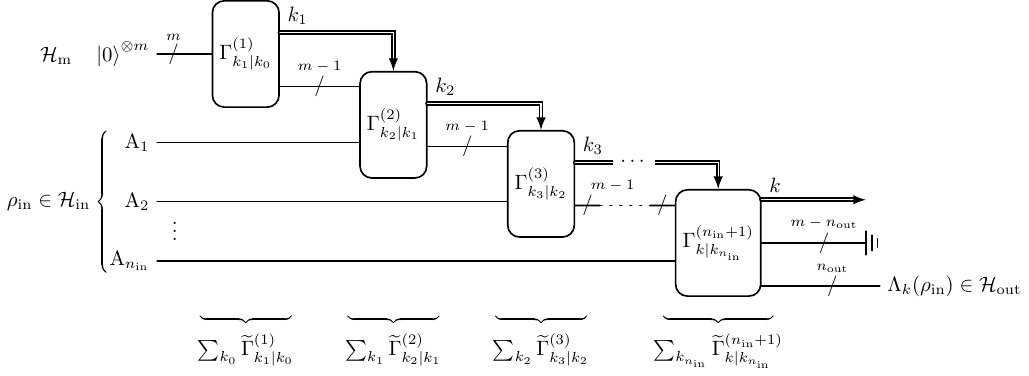}
\caption{An equivalent expression of an $m$-qubit implementable instrument (with delayed inputs). Each instrument $\{\Gamma^{(t)}_{k_t \mid k_{t-1}} \}_{k_t}$ is an $m$-qubit implementable instrument (without delayed inputs) from $\mathcal{L}((\mathbb{C}^2)^{\otimes m})$ to $\mathcal{L}((\mathbb{C}^2)^{\otimes (m-1)})$ for $1 \leq t \leq n_\mathrm{in}$, and from $\mathcal{L}((\mathbb{C}^2)^{\otimes m})$ to $\mathcal{L}((\mathbb{C}^2)^{\otimes m})$ for $t = n_\mathrm{in}+1$.}
\label{fig:stairs_m_space_instr_wdi}
\end{figure}
\end{remark}

\begin{proof}[The proof sketch of \Cref{rem:stairs_wdi}]
  The full proof is given in \Cref{app:pf_stairs_space_inst_wdi}. Here we provide only a sketch.

  Each input-loading operation can be decomposed into a sequence of input-loading operations, each of which loads one qubit, as illustrated in~\Cref{fig:elementary_l_decomp}. Thus, without loss of generality, we may assume that each input-loading operation in \Cref{def:m_space_instr_wdi} loads one qubit.

  \begin{figure}[H]
    \centering
    \includegraphics[width=0.7\linewidth]{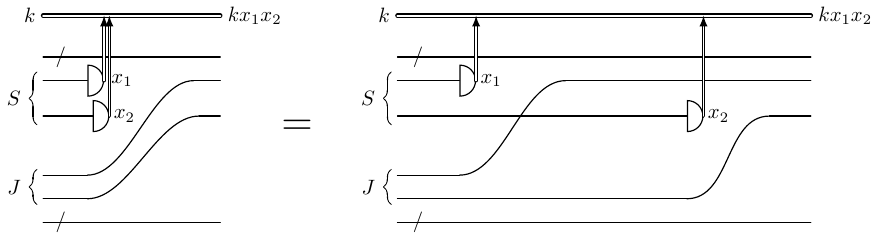}
    \caption{Decomposition of an input-loading operation into input-loading operations, each of which loads one qubit.}
    \label{fig:elementary_l_decomp}
  \end{figure}

  Assume that $\Lambda$ is an $m$-qubit implementable instrument (with delayed inputs) as defined in \Cref{def:m_space_instr_wdi}.
  Partition the sequence of elementary quantum instruments by each input-loading instrument.
  By grouping the elementary instruments between two consecutive input-loading instruments into a single $m$-qubit implementable instrument (without delayed inputs),
  we obtain the expression in \Cref{eq:stairs_wdi_1,eq:stairs_wdi_2}.

  The converse direction is straightforward: any $\Lambda$ of the form \Cref{eq:stairs_wdi_1,eq:stairs_wdi_2} can be implemented as a composition of $m$-qubit implementable instruments (without delayed inputs) and input-loading instruments.
\end{proof}

\begin{remark}
  \Cref{rem:stairs_wdi} holds only under the assumption that each input qubit can be prepared independently at any time after the circuit execution starts (\Cref{rem:indep_loading}).
  This is because the proof step that decomposes an input-loading operation into a sequence of single-qubit input-loading operations is valid only under that assumption.
\end{remark}

\section{Analytical Tools for Space Requirements of Quantum Instruments}\label{sec:tools}
In this section, we introduce analytical tools for studying the space requirements of quantum instruments.
The two main tools are the composability of instruments (\Cref{sec:pp_inst}) and the outcome no-signaling condition (\Cref{sec:signaling}).
Before turning to these tools, we first introduce two preliminary notions concerning quantum instruments, which are particularly convenient for analyzing space requirements and will recur in the lemmas and proofs below.

\vspace{2ex}
\noindent\textbf{POVM Associated with a Quantum Instrument.}

The POVM associated with a quantum instrument $\Lambda := \{ \Lambda_k: \mathcal{L}(\mathcal{H}_\mathrm{in}) \to \mathcal{L}(\mathcal{H}_\mathrm{out}) \}_k$ means the POVM that yields the same outcome probabilities as $\Lambda$ for any input state, formally defined as the POVM $\AssocPOVM{\Lambda} := \{\AssocPOVM{\Lambda}_k\}_{k}$ satisfying
\begin{align}
  \Tr[\AssocPOVM{\Lambda}_k \rho] \;=\; \Tr[\Lambda_k(\rho)]
  \qquad \forall\,\rho \in \mathcal{L}(\mathcal{H}_{\mathrm{in}}) .
\end{align}
In~\cite{Lepp2021}, this notion is called the induced POVM of the quantum instrument.
By definition, when $\Lambda$ admits a Kraus representation $\Lambda_k(\rho) = \sum_i K_{k,i}\,\rho\,K_{k,i}^\dagger$, its associated POVM can be written as $ \AssocPOVM{\Lambda}_k \;=\; \sum_i K_{k,i}^\dagger K_{k,i}\,.$

\vspace{2ex}
\noindent\textbf{Quantum Instruments with Kraus-Rank-1 CP maps.}

In the lemmas and proofs below, we often focus on instruments $\{\Lambda_k\}_k$ in which each quantum operation $\Lambda_k$ has Kraus rank~1. Such instruments have properties preferable for analysis of space requirements (e.g., \Cref{lem:output_dim_indecomp} below). In~\cite{Lepp2021}, such instruments are called indecomposable instruments.

\vspace{2ex}
For any fixed POVM $E=\{E_k\}_k$, the quantum instruments $\Lambda:=\{\Lambda_k\}_k$ whose associated POVM equals $E$ are not unique. Among these, we may further restrict attention to instruments for which each $\Lambda_k$ has Kraus rank~1; these are again non-unique. A canonical example is the L\"uders instrument, given by
\begin{align}
  \Lambda_k(\rho)\;=\;\sqrt{E_k}\,\rho\,\sqrt{E_k}\,.
\end{align}
Within these instruments, those with the smallest output dimension are especially useful for our analysis of space requirements, as shown in \Cref{lem:output_dim_indecomp} below.

\begin{lemma}\label{lem:output_dim_indecomp}
Let $E=\{E_k\}_{k \in \Outk}\subseteq \mathcal{L}(\mathcal{H}_{\mathrm{in}})$ be a POVM, and set
$r_\ast \coloneqq \max_{k \in \Outk}\operatorname{rank}(E_k)$.
\begin{enumerate}
  \item[(i)] (\emph{Existence}) For any integer $r\ge r_\ast$, there exists a quantum instrument
  $\Gamma=\{\Gamma_k:\mathcal{L}(\mathcal{H}_{\mathrm{in}})\to \mathcal{L}(\mathcal{H}_{\mathrm{out}})\}_{k \in \Outk}$
  with $\dim \mathcal{H}_{\mathrm{out}}=r$ such that each $\Gamma_k$ has Kraus rank~$1$ and the
  associated POVM of $\Gamma$ equals $E$.
  \item[(ii)] (\emph{Optimality}) For any quantum instrument $\Gamma$ whose associated POVM is $E$ and for which each
  $\Gamma_k$ has Kraus rank~$1$, one must have $\dim \mathcal{H}_{\mathrm{out}} \ge r_\ast$.
\end{enumerate}
\end{lemma}
In particular, $r_\ast=\max_{k}\operatorname{rank}(E_k)$ is the smallest achievable output dimension among all quantum instruments whose associated POVM is $E$ and for which each quantum operation has Kraus rank~$1$.

\begin{proof}
  See \Cref{app:pf_lem_output_dim_indecomp}.
\end{proof}

\subsection{Composability of Quantum Instruments} \label{sec:pp_inst}
In \Cref{def:space_instr_wo_delayed,def:m_space_instr_wdi}, the space-constrained quantum instruments are defined via decomposability into the elementary instruments. In this respect, the analysis of space requirements of instruments can be viewed as a special case of decomposability analysis: whether a given quantum instrument can be expressed as a decomposition of other quantum instruments.
Accordingly, in this section, we introduce the notion of \emph{composability of quantum instruments} and collect its properties that are useful for our analysis of space requirements.

Concretely, the composability of quantum instruments formalizes when a given quantum instrument can be decomposed using another given quantum instrument, and is defined as follows:
\begin{definition}[Composability of Quantum Instruments (\cite{Lepp2021})]\label{def:comp_inst}
    Let $\{\Lambda_k : \mathcal{L}(\mathcal{H}_{\mathrm{in}}) \to \mathcal{L}(\mathcal{H}_{\mathrm{out}}) \}_{k \in \Outk}$ and $\{\Gamma_l : \mathcal{L}(\mathcal{H}_{\mathrm{in}})\to \mathcal{L}(\mathcal{H}_{\mathrm{mid}}) \}_{l \in \Outl}$ be quantum instruments.
    We say that $\{\Lambda_k\}_k$ is composable from $\{\Gamma_l\}_l$ and write
    $\{\Gamma_l\}_l \rcircarrow \{\Lambda_k\}_k$,
    if there exists, for each $l \in \Outl$, a quantum instrument $\{\Theta_{k \mid l}:\mathcal{L}(\mathcal{H}_{\mathrm{mid}})\to\mathcal{L}(\mathcal{H}_{\mathrm{out}})\}_{k \in \widetilde{\Outk}}$ such that,
    \begin{align}
    \Lambda_k(\rho)
    =\sum_{l \in \Outl}\Theta_{k \mid l} \circ \Gamma_l(\rho)
    \qquad \forall\rho \in \mathcal{L}(\mathcal{H}_{\mathrm{in}}) \text{ and } k \in \Outk.
    \end{align}
\end{definition}
\begin{figure}[h]
  \centering
  \includegraphics[width=0.5\textwidth]{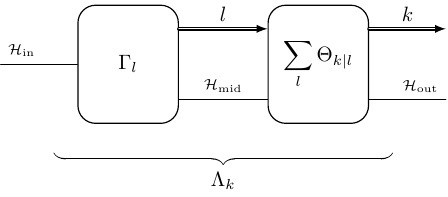}
  \caption{Composability of quantum instruments defined in \Cref{def:comp_inst}. We say that $\{\Lambda_k\}_{k \in \Outk}$ is composable from $\{\Gamma_l\}_{l \in \Outl}$ if there exists a quantum instrument $\{\Theta_{k \mid l}\}_{k \in \widetilde{\Outk}}$ for each $l \in \Outl$ such that $\Lambda_k = \sum_{l \in \Outl}\Theta_{k \mid l} \circ \Gamma_l$ for all $k \in \Outk$. This means that, after performing $\{\Gamma_l\}_{l \in \Outl}$, we can implement $\{\Lambda_k\}_{k \in \Outk}$ by applying an additional quantum instrument $\{\Theta_{k \mid l}\}_{k \in \widetilde{\Outk}}$ depending on the outcome $l$.
  }
  \label{fig:def_postproc}
\end{figure}

The composability of quantum instruments is also referred to as the post-processing relation for quantum instruments in~\cite{Lepp2021}.
\Cref{fig:def_postproc} illustrates the definition of the composability of quantum instruments $\{\Lambda_k\}_{k \in \Outk}$ and $\{\Gamma_l\}_{l \in \Outl}$. In words, $\{\Lambda_k\}_{k \in \Outk}$ is composable from $\{\Gamma_l\}_{l \in \Outl}$ if, after performing $\{\Gamma_l\}_{l \in \Outl}$, we can implement $\{ \Lambda_k \}_{k \in \Outk}$ by applying an additional quantum instrument $\{\Theta_{k \mid l}\}_{k \in \widetilde{\Outk}}$ depending on the outcome $l$.

\begin{remark}\label{rem:enlarge_outcome_set}
  In \Cref{def:comp_inst}, the outcome set of $\{\Theta_{k\mid l}\}_{k \in \widetilde{\Outk}}$ may be larger than that of $\{\Lambda_k\}_{k \in \Outk}$. The extra outcomes $k\in \widetilde{\Outk} \setminus \Outk$ satisfy
  \begin{align}
  \Theta_{k\mid l}\circ \Gamma_l \;=\; 0 \qquad \forall k \in \widetilde{\Outk} \setminus \Outk, \label{eq:zero_outcomes}
  \end{align}
  for every $l \in \Outl$,
  so they never occur when performing after $\Gamma_l$; they are included only to ensure that $\sum_{k \in \widetilde{\Outk}}\Theta_{k\mid l}$ is trace-preserving on $\mathcal H_{\mathrm{mid}}$ for each fixed $l$.
  Indeed, \Cref{eq:zero_outcomes} is obtained as follows:
  \begin{align}
  \Tr[\rho]
  &= \Tr[\sum_{k \in \widetilde{\Outk}} \sum_{l \in \Outl} \; \Theta_{k\mid l}\circ\Gamma_{l}(\rho)] \qquad (\because \text{trace-preserving for } \{\Theta_{k\mid l}\}_{k \in \widetilde{\Outk}}, \{\Gamma_l\}_{l \in \Outl})\\
  &= \Tr[\sum_{k \in \Outk}\Lambda_k(\rho)]
     + \sum_{k \in \widetilde{\Outk} \setminus \Outk,\, l \in \Outl} \Tr[\Theta_{k\mid l} \circ \Gamma_{l}(\rho)] \\
  &= \Tr[\rho]
     + \sum_{k \in \widetilde{\Outk} \setminus \Outk,\, l \in \Outl}\Tr[\Theta_{k\mid l} \circ \Gamma_{l}(\rho)]
  \qquad (\because \text{trace-preserving for } \{\Lambda_k\}_{k \in \Outk}) ,
  \end{align}
  so each nonnegative summand must vanish:
  $\Tr[\Theta_{k\mid l} \circ \Gamma_l(\rho)]=0$ for every $k \in \widetilde{\Outk} \setminus \Outk$ and $l \in \Outl$.
\end{remark}

Furthermore, as a special case of \Cref{def:comp_inst}, the composability of POVMs can be simplified as follows.
Let $\Lambda_k (\rho) = \Tr[ E_k \rho]$ and $\Gamma_l(\rho) = \Tr[F_l \rho]$ for all $\rho \in \mathcal{L}(\mathcal{H})$, where $\{E_k\}_{k \in \Outk}$ and $\{F_l\}_{l \in \Outl}$ are POVMs on $\mathcal{H}$.
In this case, the input and output space of $\{\Theta_{k \mid l}\}_{k \in \widetilde{\Outk}}$ in \Cref{def:comp_inst} are both $\mathbb{C}$, so we may write $\Theta_{k \mid l}: x \mapsto \nu_{k,l} x$ where $\nu_{k,l}$ is a non-negative scalar $\nu_{k,l} \geq 0$ satisfying $\sum_{k \in \Outk} \nu_{k,l} = 1$ for each $l \in \Outl$.
Consequently, if $\{E_k\}_{k \in \Outk}$ is composable from $\{F_l\}_{l \in \Outl}$, then there exists a column-stochastic matrix $\nu = (\nu_{k,l})_{k \in \Outk, l \in \Outl}$ such that
\begin{align}
E_k = \sum_{l \in \Outl} \nu_{k,l} F_l \quad \forall k \in \Outk.
\end{align}
Here, a column-stochastic matrix is a matrix $\nu = (\nu_{k,l})_{k \in \Outk, l \in \Outl}$ that satisfies $\nu_{k,l} \geq 0$ for all $k \in \Outk, l \in \Outl$ and $\sum_{k \in \Outk} \nu_{k,l} = 1$ for all $l \in \Outl$.

When certain conditions are met, the composability of quantum instruments can be characterized via the composability of their associated POVMs.
First, prior work~\cite{Lepp2021} shows that the composability of two quantum instruments implies the composability of their associated POVMs in the opposite direction.

\begin{lemma}[Prop.~8 in \cite{Lepp2021}]\label{lem:nec_cond_pp_indecomp_instr}
Let $\Lambda := \{\Lambda_k\}_{k \in \Outk}$ and $\Gamma := \{\Gamma_l\}_{l \in \Outl}$ be quantum instruments where each $\Lambda_k$ has Kraus rank 1 for every $k \in \Outk$. Then,
\begin{align}
\Gamma \rcircarrow \Lambda
\quad\Longrightarrow\quad
 \AssocPOVM{\Gamma} \lcircarrow \AssocPOVM{\Lambda},
\end{align}
where $\AssocPOVM{\Lambda}$ and $\AssocPOVM{\Gamma}$ are the POVMs associated with $\Lambda$ and $\Gamma$, respectively.
\end{lemma}
\begin{proof}
See Prop.~8 in \cite{Lepp2021}.
\end{proof}

We further show that, under additional Kraus-rank constraints and conditions on the associated POVMs, the above POVM-level composability is not only necessary but also sufficient.
This result is useful when we decompose instruments into other instruments based on composability relations of their associated POVMs.

\begin{lemma}\label{lem:pp_cond_indecomp_instr}
Let $\Lambda := \{\Lambda_k\}_{k \in \Outk}$ and $\Gamma := \{\Gamma_l\}_{l \in \Outl}$ be quantum instruments where each $\Lambda_k$ and each $\Gamma_l$ has Kraus rank 1 for every $k \in \Outk$ and $l \in \Outl$. Suppose further that the associated POVM $\AssocPOVM{\Lambda} $ is composable from some projective measurement. Then
\begin{align}
\Gamma \rcircarrow \Lambda
\quad\Longleftrightarrow\quad
 \AssocPOVM{\Gamma} \lcircarrow \AssocPOVM{\Lambda},
\end{align}
where $\AssocPOVM{\Lambda}$ and $\AssocPOVM{\Gamma}$ are the POVMs associated with $\Lambda$ and $\Gamma$, respectively.
\end{lemma}

\begin{remark}
By construction in the proof in \Cref{app:pf_pp_cond_indecomp_instr}, the quantum instrument $\{\Theta_{k\mid l}\}_{k \in \widetilde{\Outk}}$ such that $\sum_l \Theta_{k\mid l} \circ \Gamma_l = \Lambda_k$ can be chosen so that each $\Theta_{k\mid l}$ has Kraus rank 1 for every $k \in \widetilde{\Outk}$ and $l \in \Outl$.
\end{remark}

\begin{proof}
  See \Cref{app:pf_pp_cond_indecomp_instr}.
\end{proof}

\subsection{No-Signaling Condition for Quantum Instruments}
\label{sec:signaling}

The no-signaling condition for quantum channels is known as a criterion that formalizes the causal relations between a channel's input and output subsystems~\cite{bruss2000, beckman2001, schumacher2004, Oreshkov2012, Lorenz2021, Apadula2024nosignalling}, and it has been used in various tasks such as distributed implementation of bipartite quantum channels~\cite{beckman2001, Eggeling2002} and channel discrimination~\cite{chiribella2012}. Formally, for a quantum channel $\mathcal{E} : \mathcal{L}(\mathcal{H}_\text{A} \otimes \mathcal{H}_\text{B}) \to \mathcal{L}(\mathcal{H}_\text{C} \otimes \mathcal{H}_\mathrm{D})$, the no-signaling condition from $\mathrm{B}$ to $\mathrm{C}$ is defined as the existence of a quantum channel $\mathcal{E}' : \mathcal{L}(\mathcal{H}_\text{A}) \to \mathcal{L}(\mathcal{H}_\text{C})$ such that
\begin{align}
\Tr_{\mathrm{D}}[\mathcal{E}(\rho)]
= \mathcal{E}'(\Tr_{\mathrm{B}}[\rho])
\qquad \forall \rho \in \mathcal{L}(\mathcal{H}_\text{A} \otimes \mathcal{H}_\text{B}).
\label{eq:def_ns}
\end{align}
In \Cref{eq:def_ns}, the left-hand side is the marginal output state on $\mathrm{C}$, while the right-hand side is computed only from the marginal input state on $\mathrm{A}$. Thus, the above no-signaling condition states that the input state on B does not affect the output state on C.

In our analysis of space requirements of quantum instruments, we introduce a no-signaling condition for quantum instruments that characterizes the causal relations between the quantum input and classical outcome of an instrument.
Such a no-signaling condition arises naturally in our analysis: In a space-constrained implementation, when a mid-circuit measurement yields outcome $k$ and a delayed input A is loaded thereafter, the classical outcome $k$ must not depend on the input state in A.
We call this no-signaling condition the \emph{outcome no-signaling condition}, and define it as follows:

\begin{definition}[Outcome No-Signaling Condition]\label{def:ons}
Let $\Lambda := \qty{\Lambda_k : \mathcal{L}(\mathcal{H}_\text{A} \otimes \mathcal{H}_\text{B}) \to \mathcal{L}(\mathcal{H}_\text{C})}_{k \in \Outk}$ be a quantum instrument. We say that $\Lambda$ satisfies the outcome no-signaling condition from $\text{B}$ and write $\text{B} \nrightarrow \mathrm{cl}$, if there exists a POVM $\{F_k \in \mathcal{L}(\mathcal{H}_\text{A})\}_{k \in \Outk}$ such that for all $k \in \Outk$,
\begin{align}
\Tr[\Lambda_k(\rho)]
= \Tr[F_k\,\Tr_\text{B}(\rho)]
\quad \forall \rho \in \mathcal{L}(\mathcal{H}_\text{A} \otimes \mathcal{H}_\text{B}).
\label{eq:def_ons}
\end{align}
\end{definition}

\begin{figure}[h]
  \centering
  \begin{subfigure}[t]{0.50\textwidth}
    \centering
    \includegraphics[width=\linewidth]{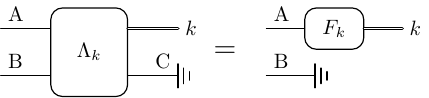}\vspace{0pt}
    \caption{}
    \label{fig:ons_def}
  \end{subfigure}
  \hspace{0.15\textwidth}
  \begin{subfigure}[t]{0.2\textwidth}
    \centering
    \includegraphics[width=\linewidth]{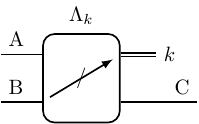}\vspace{0pt}
    \caption{}
    \label{fig:ons_meaning}
  \end{subfigure}
  \caption{ Overview of the outcome no-signaling condition $\text{B} \nrightarrow \mathrm{cl}$ defined in \Cref{def:ons}. (a) The defining equality (\Cref{eq:def_ons}). The left-hand side is the probability of obtaining outcome $k$ from a quantum instrument $\Lambda$, while the right-hand side is computed only from the marginal input state on $\text{A}$. (b) Interpretation of $\text{B} \nrightarrow \mathrm{cl}$. The classical outcome of the quantum instrument $\Lambda$ does not depend on the input state of subsystem B.}
  \label{fig:suff_qubit_reduction}
\end{figure}

With a slight abuse of notation, we say that a POVM $\{E_k\}_{k \in \Outk}$ satisfies the outcome no-signaling condition $\text{B} \nrightarrow \mathrm{cl}$ if there exists a POVM $\{F_k\}_{k \in \Outk}$ on $\mathcal{H}_\text{A}$ such that $\Tr[ E_k \rho] = \Tr[ F_k \,\Tr_\text{B}(\rho)]$ for all $k \in \Outk$ and $\rho \in \mathcal{L}(\mathcal{H}_\text{A} \otimes \mathcal{H}_\text{B})$.

\Cref{fig:ons_def} illustrates the definition of the outcome no-signaling condition. The left-hand side of \Cref{eq:def_ons} gives the probability of obtaining each outcome $k$ from $\Lambda$, whereas the right-hand side depends only on the marginal input state in subsystem A. Hence, the classical outcome of the instrument $\Lambda$ does not depend on the input state of subsystem B, as shown in \Cref{fig:ons_meaning}.

\begin{remark}\label{rem:ons_povm}
The outcome no-signaling condition for $\Lambda$ is equivalent to the outcome no-signaling condition for the POVM associated with $\Lambda$. Formally, for $\Lambda := \{\Lambda_k : \mathcal{L}(\mathcal{H}_\text{A} \otimes \mathcal{H}_\text{B}) \to \mathcal{L}(\mathcal{H}_\text{C})\}_{k \in \Outk}$, we have
\begin{align}
\Lambda \text{ satisfies B} \nrightarrow \mathrm{cl}
\quad\Longleftrightarrow\quad
\AssocPOVM{\Lambda} \text{ satisfies B} \nrightarrow \mathrm{cl}.
\label{eq:ons_povm}
\end{align}
Indeed, by the definition of the associated POVM, the defining equality for outcome no-signaling (\Cref{eq:def_ons}) can be written as
\begin{align}
\operatorname{Tr}\!\big[\AssocPOVM{\Lambda}_k\, \rho\big]
= \operatorname{Tr}\!\big[\Lambda_k(\rho)\big]
= \operatorname{Tr}\!\big[F_k\, \operatorname{Tr}_{\mathrm{B}}(\rho)\big]
\qquad \forall\, \rho \in \mathcal{L}(\mathcal{H}_{\mathrm{A}} \otimes \mathcal{H}_{\mathrm{B}}),\; k \in \Outk.
\end{align}
\end{remark}

Furthermore, the outcome no-signaling condition can be characterized as a decomposition that makes the independence of the input subsystem manifest:

\begin{theorem}\label{thm:ons_decomp}
  Let $\Lambda := \{\Lambda_k : \mathcal{L}(\mathcal{H}_\text{A} \otimes \mathcal{H}_\text{B}) \to \mathcal{L}(\mathcal{H}_\text{C})\}_{k \in \Outk}$ be a quantum instrument. The following two conditions are equivalent:
  \begin{enumerate}[label=(\alph*)]
    \item The quantum instrument $\Lambda$ satisfies the outcome no-signaling condition $\text{B} \nrightarrow \mathrm{cl}$.
    \item There exists a quantum instrument $\Gamma := \{\Gamma_k : \mathcal{L}(\mathcal{H}_{\text{A}}) \to \mathcal{L}(\mathcal{H}_\text{X})\}_{k \in \Outk}$ and a quantum channel $\mathcal{E}^{(k)} : \mathcal{L} (\mathcal{H}_\text{X} \otimes \mathcal{H}_{\text{B}}) \to \mathcal{L} (\mathcal{H}_\text{C})$ for each $k \in \Outk$, such that
    \begin{align}
      \Lambda_k(\rho) =  \qty( \mathcal{E}^{(k)} \circ \qty( \Gamma_k \otimes \operatorname{id}_{\text{B}} )) (\rho)
      \qquad\forall \rho \in \mathcal{L}(\mathcal{H}_{\text{A}} \otimes \mathcal{H}_{\text{B}}).
      \label{eq:ons_decomp_general}
    \end{align}
  \end{enumerate}
\end{theorem}

\begin{figure}[h]
  \centering
  \includegraphics[width=0.8\linewidth]{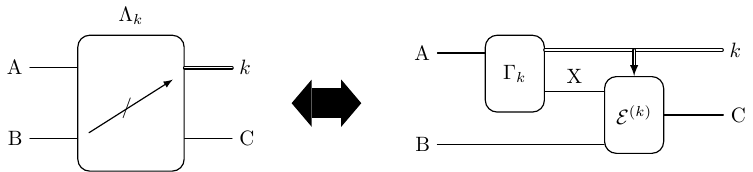}
  \caption{Equivalence between the outcome no-signaling condition $\text{B} \nrightarrow \mathrm{cl}$ and a decomposition in \Cref{eq:ons_decomp_general}. The outcome no-signaling condition $\text{B} \nrightarrow \mathrm{cl}$ on a quantum instrument $\Lambda$ holds if and only if $\Lambda$ can be decomposed into a quantum instrument $\Gamma$ acting only on subsystem A, followed by a quantum channel $\mathcal{E}^{(k)}$ that may depend on the classical outcome $k$.}
  \label{fig:ons_decomp}
\end{figure}

\begin{proof}[Proof Sketch]
  The full proof is given in \Cref{app:pf_thm_ons_decomp}. Here, we provide a proof sketch.

  The implication (b) $\Rightarrow$ (a) is straightforward by taking the trace on both sides of the equality in (b).

  For the converse direction (a) $\Rightarrow$ (b), note first that the outcome no-signaling condition $\mathrm{B} \nrightarrow \mathrm{cl}$ ensures that the POVM associated with $\Lambda$ factorizes as
  \begin{equation}\label{eq:ons_povm_recall}
    \AssocPOVM{\Lambda}_k = F_k \otimes \mathbb{I}_{\mathrm B}
    \qquad \forall\,k \in \Outk,
  \end{equation}
  for some POVM $\{F_k\}_{k\in\Outk}$ on subsystem A. We then take $\Gamma$ to be the L\"uders instrument for $\{F_k\}_{k\in\Outk}$ on A and explicitly construct, for each outcome $k$, a quantum channel that reproduces $\Lambda_k$. This yields the desired decomposition and establishes (a) $\Rightarrow$ (b).
\end{proof}

In our analysis of space requirements of quantum instruments, the spaces needed to implement $\Gamma$ and $\mathcal{E}^{(k)}$ in \Cref{thm:ons_decomp} are crucial. In this respect, we show below that, when restricting to qubit systems and imposing a Kraus-rank constraint on $\Lambda$, each $\mathcal{E}^{(k)}$ can be chosen to be unitary.

\begin{lemma}\label{lem:ons_decomp_mindim}
  Let $\mathcal{H}_\text{A} \cong (\mathbb{C}^{2})^{\otimes n_\text{A}}, \mathcal{H}_\text{B} \cong (\mathbb{C}^{2})^{\otimes n_\text{B}}, \mathcal{H}_\text{C} \cong (\mathbb{C}^{2})^{\otimes n_\text{C}}$ for some $n_\text{A}, n_\text{B}, n_\text{C} \in \mathbb{Z}_{\geq 0}$ and $\Lambda := \{\Lambda_k : \mathcal{L}(\mathcal{H}_\text{A} \otimes \mathcal{H}_\text{B}) \to \mathcal{L}(\mathcal{H}_\text{C})\}_{k \in \Outk}$ be a quantum instrument where each $\Lambda_k$ has Kraus rank 1 for every $k \in \Outk$. The following two conditions are equivalent:
  \begin{enumerate}[label=(\alph*)]
    \item The quantum instrument $\Lambda$ satisfies the outcome no-signaling condition $\text{B} \nrightarrow \mathrm{cl}$.
    \item There exists a quantum instrument $\Gamma := \{\Gamma_k : \mathcal{L}(\mathcal{H}_{\text{A}}) \to \mathcal{L}(\mathcal{H}_\text{X})\}_{k \in \Outk}$ with each $\Gamma_k$ having Kraus rank 1 and a unitary operator $ U_k : \mathcal{H}_\text{X} \otimes \mathcal{H}_{\text{B}} \to \mathcal{H}_\text{C}$ for each $k \in \Outk$, such that
    \begin{align}
      \Lambda_k(\rho) =  U_k \qty( \Gamma_k \otimes \operatorname{id}_{\text{B}} )(\rho) U_k^\dagger
      \qquad\forall \rho \in \mathcal{L}(\mathcal{H}_{\text{A}} \otimes \mathcal{H}_{\text{B}}).
    \end{align}
  \end{enumerate}
\end{lemma}

\begin{proof}
See \Cref{app:pf_lem_ons_decomp}.
\end{proof}

\section{Characterization of Space Requirements of Quantum Instruments}
\label{sec:qubit_reduction}

In this section, we characterize the number of qubits required to implement a given quantum instrument.
Specifically, we present necessary conditions and sufficient conditions in which a given quantum instrument is implementable under space constraints.
These conditions imply upper and lower bounds on the space requirements of quantum instruments.
In what follows, we state our results separately for the settings without and with delayed inputs.

\subsection{Space Requirements for the Setting without Delayed Inputs}
In the setting without delayed inputs, the space required to implement a quantum instrument must be at least as large as its input or output system size (whichever is larger) for holding the entire input state and the entire output state.
However, this lower bound may be insufficient because the Stinespring dilation generally requires auxiliary qubits beyond the input and output systems.
Accordingly, the analysis of space requirements in this setting can be rephrased as clarifying how much additional space is required beyond the input and output systems to implement a given quantum instrument.

The following result on space requirements of POVMs is known from prior work \cite{Lloyd2001engineeringquantumdynamics,Andersson2008binarysearchtrees}.

\begin{lemma}[\cite{Lloyd2001engineeringquantumdynamics, Andersson2008binarysearchtrees}]\label{lem:povm_space}
 Fix $m \in \mathbb{N}$. Every POVM on $(m-1)$ qubits is $m$-qubit implementable (without delayed inputs).
\end{lemma}

\begin{proof}[Proof Sketch]
  The full proof is given in \Cref{app:pf_povm_space} or in prior work \cite{Lloyd2001engineeringquantumdynamics, Andersson2008binarysearchtrees}; here, we provide a proof sketch in our notation.

  Let $E:=\{E_k\}_{k\in K}$ be a POVM on $\mathcal{H}\cong (\mathbb{C}^2)^{\otimes (m-1)}$. If necessary, enlarge the outcome set by adding $E_k=0$ so that $|K|=2^{T}$ for some $T\in\mathbb{N}$. Write each outcome label in binary as
  $k = k_{(1)} \concat k_{(2)} \concat \cdots \concat k_{(T)} \in \{0,1\}^{T}$, and for $t\ge 1$ write the prefix $k_{(< t)}:=k_{(1)}\concat \cdots \concat k_{(t-1)}$.

  We prove the claim by explicitly realizing $E$ as illustrated in \Cref{fig:povm_m_space}, namely,
  \begin{align}
    \Tr[E_k \rho]
    \;=\;
    \Tr\!\Big[
      \Gamma^{(T)}_{k_{(T)} \mid k_{(< T)}} \circ \cdots \circ
      \Gamma^{(2)}_{k_{(2)} \mid k_{(< 2)}} \circ
      \Gamma^{(1)}_{k_{(1)} \mid \emptyset}
      \,(\rho)
      \Big]
      \qquad \forall\, \rho \in \mathcal{L}(\mathcal{H}),\; k \in K.
      \label{eq:seq_const_povm}
    \end{align}
    At round $t\in\{1,\dots,T\}$, the instrument $\{\Gamma^{(t)}_{k_{(t)}\mid k_{(< t)}}:\mathcal L(\mathcal H)\to\mathcal L(\mathcal H)\}_{k_{(t)}\in\{0,1\}}$ reveals the $t$-th bit $k_{(t)}$ of the final outcome $k$.

    \begin{figure}[H]
    \centering
    \includegraphics[width=0.7\linewidth]{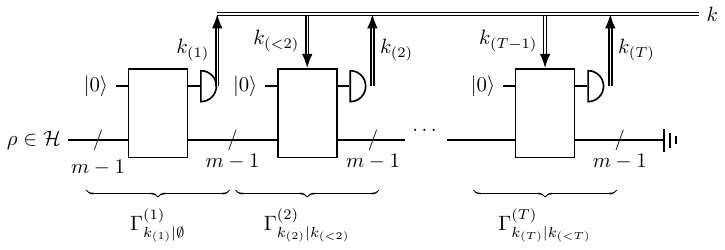}
    \caption{Realizing the POVM $E$ as a composition of the elementary operations in \Cref{def:space_instr_wo_delayed}. At round $t$, the instrument $\{\Gamma^{(t)}_{k_{(t)}\mid k_{(< t)}}\}_{k_{(t)}\in\{0,1\}}$ reveals $k_{(t)}$ (the $t$-th bit of outcome $k$). Each box in the figure represents an $m$-qubit unitary operation.}
    \label{fig:povm_m_space}
    \end{figure}

    As shown in the full proof, these instruments can be chosen recursively so that the cumulative instrument up to round $t$ coincides with the Lüders instrument for the coarse-grained POVM:
    \begin{align}
      \Big(
        \Gamma^{(t)}_{l_{(t)} \mid l_{(< t)}} \circ \cdots \circ
        \Gamma^{(1)}_{l_{(1)} \mid \emptyset}
      \Big) (\rho)
      \;=\;
      \sqrt{R_l}\,\rho\,\sqrt{R_l}
      \qquad \forall\, \rho \in \mathcal{L}(\mathcal{H}),
    \end{align}
    where, for any binary string $l$ of length $t$,
    \begin{align}
      \mathcal K_l \;:=\; \{\,k\in\{0,1\}^{T}:\ \text{the first $t$ bits of $k$ equal $l$}\,\}
      \qquad
      R_l \;:=\; \sum_{k\in \mathcal K_l} E_k .
    \end{align}

    By construction, each instrument $\{\Gamma^{(t)}_{k_{(t)}\mid k_{(< t)}}\}_{k_{(t)}\in\{0,1\}}$ has both input and output systems $\mathcal{H} \cong (\mathbb{C}^2)^{\otimes (m-1)}$ and has two outcomes. Hence, each instrument admits a Stinespring realization that (i) appends one ancilla qubit initialized to $\ket{0}$, (ii) applies a unitary operation on $\mathbb{C}^2 \otimes \mathcal{H} \cong (\mathbb{C}^2)^{\otimes m}$, and (iii) measures the ancilla in the computational basis, as illustrated in \Cref{fig:povm_m_space}. Therefore, the sequential construction in \Cref{eq:seq_const_povm} is $m$-qubit implementable (without delayed inputs).
\end{proof}

Based on \Cref{lem:povm_space}, we derive sufficient conditions for a quantum instrument to be $m$-qubit implementable (without delayed inputs) in terms of the associated POVM.

\begin{theorem} \label{thm:range_m_space_inst_wodi}
  Let $m \in \mathbb{N}$.
  Let $\Lambda := \qty{\Lambda_k : \mathcal{L}( \mathcal{H}_{\mathrm{in}}) \to \mathcal{L}( \mathcal{H}_{\mathrm{out}}) }_{k \in \Outk}$ be a quantum instrument
  where $\mathcal H_{\mathrm{in}} \cong (\mathbb C^2)^{\otimes n_\mathrm{in}}$ and $\mathcal H_{\mathrm{out}} \cong (\mathbb C^2)^{\otimes n_\mathrm{out}}$ with $n_\mathrm{in}, n_\mathrm{out} \in \{0, 1, \cdots m\}$.
  If there exist disjoint sets $\Outk_0, \Outk_1$ such that $\Outk = \Outk_0 \cup \Outk_1$ and, for each $b \in \{0,1\}$,
  \begin{align}
    \sum_{k \in \Outk_b} \AssocPOVM{\Lambda}_k &= P_b,\\
    \operatorname{rank}(P_b) &\leq 2^{m-1},
  \end{align}
  where $\qty{\AssocPOVM{\Lambda}_k}_{k \in \Outk}$ is the POVM associated with $\Lambda$, and $\{P_b\}_{b \in \{0,1\}}$ is a projective measurement on $\mathcal{H}_\mathrm{in}$, then $\Lambda$ is $m$-qubit implementable (without delayed inputs).
\end{theorem}

\begin{proof}[Proof Sketch]
  The full proof is given in \Cref{app:pf_range_m_space_inst_wodi}; here, we provide a proof sketch.

Write a Kraus representation $\Lambda_k(\rho)=\sum_{\alpha_k} A_{k,\alpha_k}\rho A_{k,\alpha_k}^\dagger$.
Since $\Lambda$ can be implemented by the refined instrument $\{\widetilde\Lambda_{(k,\alpha_k)}\}_{(k,\alpha_k)}$ with
$\widetilde\Lambda_{(k,\alpha_k)}(\rho):=A_{k,\alpha_k}\rho A_{k,\alpha_k}^\dagger$,
followed by the classical postprocessing $(k,\alpha_k)\mapsto k$, it suffices to treat the case in which each $\Lambda_k$ has Kraus rank~$1$, i.e., $\Lambda_k(\rho)=A_k\rho A_k^\dagger$.

We prove the claim by explicitly implementing $\Lambda$ as depicted in \Cref{fig:half_cut}, namely
\begin{align}
\Lambda_k(\rho) \;=\; \big(\mathcal{W}_k \circ \Theta_{k\mid b} \circ \Gamma_b\big)(\rho),
\end{align}
with components defined as follows.

\begin{figure}[H]
\centering
\includegraphics[width=0.8\linewidth]{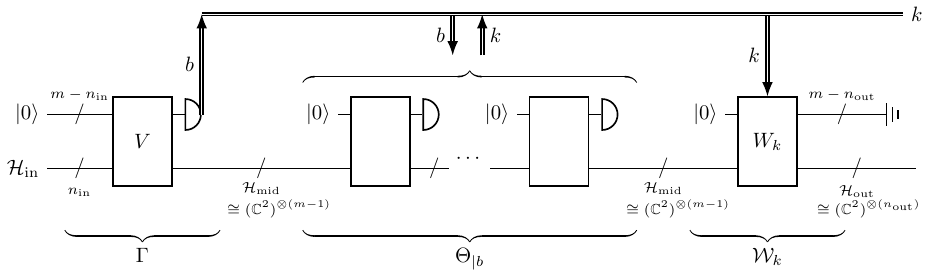}
\caption{Decomposition of the target instrument $\Lambda=\{\Lambda_k\}_{k\in\Outk}$ into three parts: the first instrument $\Gamma=\{\Gamma_b\}_{b\in\{0,1\}}$, the intermediate instrument $\Theta_{\mid b}=\{\Theta_{k\mid b}\}_{k\in\Outk_b}$, and the final channel $\mathcal{W}_{k}$. Each box in the figure represents an $m$-qubit unitary operation.}
\label{fig:half_cut}
\end{figure}

\noindent\textbf{First step ($\Gamma$).}
Perform the instrument $\Gamma=\{\Gamma_b:\mathcal{L}(\mathcal{H}_{\mathrm{in}})\to \mathcal{L}(\mathcal{H}_{\mathrm{mid}})\}_{b\in\{0,1\}}$
where its associated POVM is $\{P_b\}_{b\in\{0,1\}}$ and each $\Gamma_b$ has Kraus rank~$1$.
By \Cref{lem:output_dim_indecomp} and $\operatorname{rank}(P_b)\le 2^{m-1}$, we can choose $\dim\mathcal{H}_{\mathrm{mid}}=2^{m-1}$.
As a two-outcome instrument, $\Gamma$ admits a Stinespring realization that (i) appends $(m-n_{\mathrm{in}})$ ancilla qubits initialized to $\ket{0}$, (ii) applies an $m$-qubit unitary $V$, and (iii) measures the ancillas in the computational basis, as illustrated in \Cref{fig:half_cut}. Write a Kraus representation
\begin{align}
\Gamma_b(\rho):=K_b\rho K_b^\dagger.
\end{align}

\noindent\textbf{Intermediate step ($\Theta_{\mid b}$).}
For each $b\in\{0,1\}$, consider a set $\{N_{k \mid b} := K_b \AssocPOVM{\Lambda}_k K_b^\dagger\}_{k\in\Outk_b}$, which is a POVM on $\mathcal{H}_{\mathrm{mid}}$ with a trivial adjustment.
As in the proof of Lemma~\ref{lem:povm_space}, the Lüders instrument
\begin{align}
\Theta_{k\mid b}(\rho):=\sqrt{N_{k\mid b}}\;\rho\;\sqrt{N_{k\mid b}},
\end{align}
is implementable by repeating the operations “append one ancilla in $\ket{0}$, apply an $m$-qubit unitary, measure the ancilla,” as illustrated in \Cref{fig:half_cut}.

\noindent\textbf{Final step ($\mathcal{W}_k$).}
At this point we have implemented $\{\Theta_{k\mid b}\circ\Gamma_b\}_{k\in\Outk}$, where we may index outcomes only by $k \in \Outk$ since $b$ can be uniquely identified from $k \in \Outk_b$.
Its associated POVM agrees with that of $\Lambda$: $\AssocPOVM{\Theta_{k\mid b}\circ \Gamma_b}_k=\AssocPOVM{\Lambda}_k$.
By \Cref{lem:isometry_equivalence_of_kraus_op}, there exists a quantum channel
$\mathcal{W}_k:\mathcal{L}(\mathcal{H}_{\mathrm{mid}})\to\mathcal{L}(\mathcal{H}_{\mathrm{out}})$ of the form
\begin{align}
\mathcal{W}_k(\rho):=\Tr_{(\mathbb{C}^2)^{\otimes (m-n_{\mathrm{out}})}}\!\big[\,W_k\,(\ketbra{0}\otimes \rho)\,W_k^\dagger\,\big],
\end{align}
for some $m$-qubit unitary $W_k$, such that $\Lambda_k(\rho)=(\mathcal{W}_k\circ \Theta_{k\mid b}\circ \Gamma_b)(\rho)$.

As indicated in \Cref{fig:half_cut}, this composition uses only the elementary unitary operations and computational basis measurements specified in \Cref{def:space_instr_wo_delayed}. Hence $\Lambda$ is $m$-qubit implementable (without delayed inputs).

\end{proof}

From \Cref{thm:range_m_space_inst_wodi}, we obtain the following corollary.

\begin{corollary}\label{cor:range_m_space_inst_wodi_trivial}
  Let $m \in \mathbb{N}$, and let $\Lambda := \qty{\Lambda_k : \mathcal{L}( (\mathbb{C}^2)^{\otimes n_{\mathrm{in}}}) \to \mathcal{L}( (\mathbb{C}^2)^{\otimes n_{\mathrm{out}}}) }_{k \in \Outk}$ be a quantum instrument with
  $n_{\mathrm{in}} \le m-1$ and $n_{\mathrm{out}} \le m$.
  Then $\Lambda$ is $m$-qubit implementable (without delayed inputs).
\end{corollary}

\begin{proof}
  Apply \Cref{thm:range_m_space_inst_wodi} with the disjoint sets $\Outk, \emptyset$, which trivially satisfy $\Outk = \Outk \cup \emptyset$.
  Then $\sum_{k \in \Outk} \AssocPOVM{\Lambda}_k = \mathbb{I}_{\mathrm{in}}$ and
  $\operatorname{rank}(\mathbb{I}_{\mathrm{in}}) = \dim \mathcal{H}_{\mathrm{in}} = 2^{n_{\mathrm{in}}} \le 2^{m-1}$,
  which satisfies the rank condition.
\end{proof}

\subsection{Space Requirements for the Setting with Delayed Inputs}
In the setting with delayed inputs, the space required to implement a quantum instrument must be at least as large as the output system size for holding the entire output state, but it can be smaller than the input system size because the input state does not need to be held in the space at the same time.

In \Cref{thm:necc_qubit_reduction,thm:suff_qubit_reduction} below, we focus on space requirements of quantum instruments $\Lambda := \{\Lambda_k\}_k$ that can be written as
\begin{align}
  \Lambda_k(\rho) := \Tr_\mathrm{R} \qty[(\ketbra{k}_\mathrm{R} \otimes \mathbb{I}_\mathrm{out} ) U \rho \, U^\dagger] \qquad \forall \rho \in \mathcal{L}(\mathcal{H}_\mathrm{in}), \label{eq:target_instr}
\end{align}
where $U : \mathcal{H}_\mathrm{in} \to \mathcal{H}_\mathrm{R} \otimes \mathcal{H}_\mathrm{out}$ is a unitary operator and $\mathcal{H}_\mathrm{in} \cong (\mathbb{C}^2)^{\otimes n_\mathrm{in}},\mathcal{H}_\mathrm{out} \cong (\mathbb{C}^2)^{\otimes n_\mathrm{out}}$ for some $n_\mathrm{in}, n_\mathrm{out} \in \mathbb{N}$, which is illustrated in \Cref{fig:target_instr}. Quantum instruments executed in our primary application, entanglement distillation protocols based on stabilizer codes, admit the expression in \Cref{eq:target_instr}. Since $\Lambda$ is trivially $n_\mathrm{in}$-qubit implementable, analyzing its space requirements is rephrased as determining the largest $T \in \{0,1,\ldots,n_\mathrm{in}\}$ such that $\Lambda$ is $(n_\mathrm{in}-T)$-qubit implementable. Accordingly, we present necessary conditions and sufficient conditions for $\Lambda$ to be $(n_\mathrm{in}-T)$-qubit implementable for each $T \in \{0,1,\ldots,n_\mathrm{in}\}$ below.

\begin{figure}[h]
  \centering
  \centering
  \includegraphics[width=0.6\linewidth]{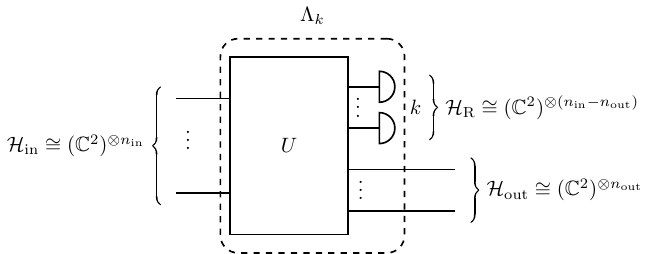}\vspace{0pt}
  \caption{Quantum instrument $\Lambda := \{\Lambda_k\}_k$ whose space requirements are analyzed in \Cref{thm:suff_qubit_reduction,thm:necc_qubit_reduction}. It is composed of an $n_\mathrm{in}$-qubit unitary operator $U: \mathcal{H}_\mathrm{in} \to \mathcal{H}_\mathrm{R} \otimes \mathcal{H}_\mathrm{out}$ followed by a computational basis measurement on $\mathcal{H}_\mathrm{R}$. Since $\Lambda$ is trivially $n_\mathrm{in}$-qubit implementable, \Cref{thm:suff_qubit_reduction,thm:necc_qubit_reduction} analyze the conditions under which $\Lambda$ is $(n_\mathrm{in}-T)$-qubit implementable for each $T \in \{0,1,\ldots,n_\mathrm{in}\}$.}
  \label{fig:target_instr}
\end{figure}

\begin{remark}\label{rem:equiv_def_target_instr}
  Up to this point, we have defined $\Lambda$ by giving an explicit form in \Cref{eq:target_instr}. Alternatively, $\Lambda$ can be characterized by the following equivalent conditions.

  Let $\Lambda := \{\Lambda_k : \mathcal{L}(\mathcal{H}_\mathrm{in}) \to \mathcal{L}(\mathcal{H}_\mathrm{out})\}_{k \in \Outk}$ be a quantum instrument with $\mathcal{H}_\mathrm{in} \cong (\mathbb{C}^2)^{\otimes n_\mathrm{in}}$ and $\mathcal{H}_\mathrm{out} \cong (\mathbb{C}^2)^{\otimes n_\mathrm{out}}$ for some $n_\mathrm{in}, n_\mathrm{out} \in \mathbb{N}$. The following are equivalent:
  \begin{enumerate}[label=(\alph*)]
    \item There exists a unitary operator $U :\mathcal{H}_\mathrm{in} \to \mathcal{H}_\mathrm{R} \otimes \mathcal{H}_\mathrm{out}$ such that $\Lambda_k(\rho) = \Tr_\mathrm{R} \qty[(\ketbra{k}_\mathrm{R} \otimes \mathbb{I}_\mathrm{out} ) U \rho \, U^\dagger]$ for all $\rho \in \mathcal{L}(\mathcal{H}_\mathrm{in})$.
    \item Each $\Lambda_k$ has Kraus rank~$1$, and the associated POVM $\AssocPOVM{\Lambda}$ is a projective measurement with $\operatorname{rank} \AssocPOVM{\Lambda}_k=2^{n_\mathrm{out}}$ for all $k \in \Outk$.
  \end{enumerate}

  The equivalence is proved as follows. The direction (a) $\Rightarrow$ (b) is straightforward by computing the Kraus ranks and the associated POVM. For (b) $\Rightarrow$ (a), since the associated POVM is a projective measurement and each element has rank $2^{n_\mathrm{out}}$, the number of outcomes is
  $
    |\Outk|=\dim \mathcal{H}_\mathrm{in}/2^{n_\mathrm{out}} = 2^{\,n_\mathrm{in}-n_\mathrm{out}}.
  $
  Hence there exists an isometry $U : \mathcal{H}_\mathrm{in} \to \mathcal{H}_\mathrm{R} \otimes \mathcal{H}_\mathrm{out}$ with
  $\mathcal{H}_\mathrm{R} \cong (\mathbb{C}^2)^{\otimes (n_\mathrm{in}-n_\mathrm{out})}$ such that
  $
    \Lambda_k(\rho) = \Tr_\mathrm{R}\!\big[(\ketbra{k}_\mathrm{R} \otimes \mathbb{I}_\mathrm{out})\, U \rho\, U^\dagger\big].
  $
  Because the input and output dimensions coincide, this isometry $U$ is in fact unitary.
\end{remark}

We now present sufficient conditions under which $\Lambda$ is $(n_\mathrm{in}-T)$-qubit implementable in terms of the POVM associated with $\Lambda$ and outcome no-signaling conditions. Any value of $(n_\mathrm{in}-T)$ satisfying these conditions, in particular the smallest such value, yields an upper bound on the number of qubits required to implement $\Lambda$.

\begin{theorem}\label{thm:suff_qubit_reduction}
  Let $\Lambda := \qty{\Lambda_k : \mathcal{L}(\mathcal{H}_\mathrm{in}) \to \mathcal{L}(\mathcal{H}_\mathrm{out})}_{k \in \Outk}$ be the quantum instrument defined by
  $\Lambda_k(\rho) := \Tr_\mathrm{R} \qty[(\ketbra{k}_\mathrm{R} \otimes \mathbb{I}_\mathrm{out} ) U \rho \, U^\dagger]$ for all $\rho \in \mathcal{L}(\mathcal{H}_\mathrm{in})$,
  where $U : \mathcal{H}_\mathrm{in} \to \mathcal{H}_\mathrm{R} \otimes \mathcal{H}_\mathrm{out}$ is a unitary operator and $\mathcal{H}_\mathrm{in} \cong (\mathbb{C}^2)^{\otimes n_\mathrm{in}},\mathcal{H}_\mathrm{out} \cong (\mathbb{C}^2)^{\otimes n_\mathrm{out}}$ for some $n_\mathrm{in}, n_\mathrm{out} \in \mathbb{N}$.

  Suppose there exist projective measurements $E^{(1)}, E^{(2)}, \cdots, E^{(T)}$ on $\mathcal{H}_\mathrm{in}$ such that
  \begin{itemize}
    \item The composability conditions hold: $E^{(1)} \lcircarrow E^{(2)} \lcircarrow \cdots \lcircarrow E^{(T)} \lcircarrow \AssocPOVM{\Lambda}$.
    \item Each $E^{(t)}$ satisfies the outcome no-signaling condition $\mathrm{A}_t \nrightarrow \mathrm{cl}$, where $\{\mathrm{A}_1, \mathrm{A}_2, \ldots, \mathrm{A}_T\}$ is an ordered subset of the input qubits.
    \item Each $E^{(t)} := \{ E^{(t)}_{k_t} \}_{k_t \in \Outk_t} $ is a projective measurement satisfying $\rank E^{(t)}_{k_t} = 2^{n_\mathrm{in} - t}$ for all $k_t \in \Outk_t$.
  \end{itemize}
  Then the quantum instrument $\Lambda$ is $(n_\mathrm{in}-T)$-qubit implementable (with delayed inputs).
\end{theorem}

\begin{proof}[Proof Sketch]
  The full proof is given in \Cref{app:pf_thm_suff_qubit_reduction}; here, we provide a proof sketch.

  For convenience, set $m := n_\mathrm{in} - T$. We prove the claim by induction on $T$.

\emph{Base case $T=0$.}
Here $m = n_\mathrm{in}$. By definition, $\Lambda$ is $n_\mathrm{in}$-qubit implementable (with delayed inputs), so the base case holds.

\emph{Induction step.}
Assume the theorem holds for $T-1$ (induction hypothesis). In the induction step, we obtain the decomposition shown in \Cref{fig:induction_step_decomp}:
\begin{align}
  \Lambda_k \;=\; \sum_{k_T}\Psi_{k\mid k_T}\circ\bigl(G_{k_T}\otimes \mathrm{id}_{A_T}\bigr),
  \label{eq:lambda_decomp_induction}
\end{align}
with the components described below.

\begin{figure}[H]
  \centering
  \includegraphics[width=0.7\linewidth]{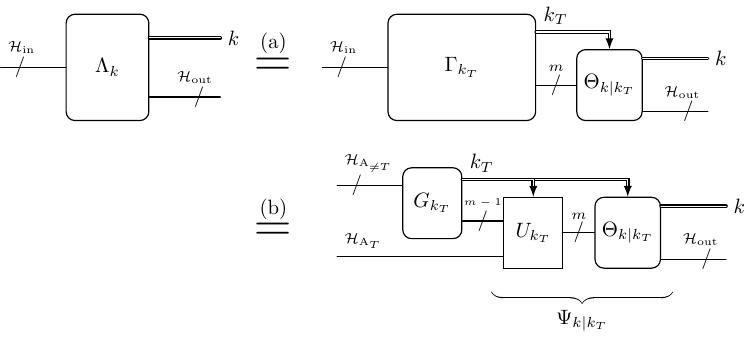}
  \caption{Decomposition of the instrument $\{\Lambda_k\}_k$ used in the induction step.
  The first equality comes from composability of $\{\Lambda_k\}_k$ from $\{\Gamma_{k_T}\}_{k_T}$, as explained in part (a).
  The second uses the outcome no-signaling condition $\mathrm A_T\nrightarrow\mathrm{cl}$ for $\{\Gamma_{k_T}\}_{k_T}$, as explained in part (b).
  Once $\{G_{k_T}\}_{k_T}$ and $\{\Psi_{k\mid k_T}\}_k$ are shown to be $m$-qubit implementable (with and without delayed inputs, respectively), it follows that $\{\Lambda_k\}_k$ is $m$-qubit implementable (with delayed inputs).}
  \label{fig:induction_step_decomp}
\end{figure}

First, define an instrument $\Gamma := \{\Gamma_{k_T}\}_{k_T\in\Outk_T}$ such that each $\Gamma_{k_T}$ has Kraus rank~1 and $\AssocPOVM{\Gamma}=E^{(T)}$. Because the associated POVM coincides, $\Gamma$ satisfies $\mathrm A_T\nrightarrow\mathrm{cl}$ by \Cref{rem:ons_povm}. By the Kraus rank condition and Lemma~\ref{lem:output_dim_indecomp}, the output system of $\Gamma$ can be taken to have $m$ qubits.

(a) Composability:
From the composability for the associated POVMs $\AssocPOVM{\Gamma}\lcircarrow \AssocPOVM{\Lambda}$ and \Cref{lem:pp_cond_indecomp_instr}, there exists an instrument $\Theta_{\mid k_T}=\{\Theta_{k\mid k_T}\}_k$ such that
\begin{align}
  \Lambda_k=\sum_{k_T}\Theta_{k\mid k_T}\circ \Gamma_{k_T}.
\end{align}

(b) Outcome no-signaling:
By \Cref{lem:ons_decomp_mindim}, the outcome no-signaling condition $\mathrm A_T\nrightarrow\mathrm{cl}$ yields
\begin{align}
  \Gamma_{k_T} \;=\; W_{k_T}\bigl(G_{k_T}\otimes \mathrm{id}_{A_T}\bigr)W_{k_T}^\dagger,
\end{align}
for some instrument $G=\{G_{k_T}\}_{k_T}$ and unitaries $W_{k_T}$. Define $\Psi_{k\mid k_T}(\sigma):=\Theta_{k\mid k_T}\!\bigl(W_{k_T}\sigma W_{k_T}^\dagger\bigr)$; then \Cref{eq:lambda_decomp_induction} holds.

From the original POVMs $E^{(1)},\ldots,E^{(T-1)}$, we can construct $(T-1)$ POVMs that satisfy the hypotheses of the theorem for $G$. Therefore, applying the induction hypothesis, $G$ is $m$-qubit implementable (with delayed inputs).
Also, from the rank conditions for $\AssocPOVM{\Gamma}$ and $\AssocPOVM{\Lambda}$ and by \Cref{thm:range_m_space_inst_wodi}, each $\Psi_{\mid k_T}$ is $m$-qubit implementable (without delayed inputs).

Combining these facts, \Cref{eq:lambda_decomp_induction} is the form \Cref{eq:stairs_wdi_1}; therefore $\Lambda$ is $m$-qubit implementable (with delayed inputs).
\end{proof}

We next present necessary conditions that any $(n_\mathrm{in}-T)$-qubit implementable instrument of $\Lambda$ must satisfy, stated in terms of the POVM associated with $\Lambda$ and outcome no-signaling conditions.
Any value of $(n_\mathrm{in}-T)$ failing to satisfy these conditions, in particular the largest such value, yields a lower bound on the number of qubits required to implement $\Lambda$.
The main difference from \Cref{thm:suff_qubit_reduction} is that $E^{(t)}$ need not be projective measurements; general POVMs are allowed, and the rank equalities are replaced with inequalities.

\begin{theorem}\label{thm:necc_qubit_reduction}
  Let $\Lambda := \qty{\Lambda_k : \mathcal{L}(\mathcal{H}_\mathrm{in}) \to \mathcal{L}(\mathcal{H}_\mathrm{out})}_{k \in \Outk}$ be the quantum instrument defined by
  $\Lambda_k(\rho) := \Tr_\mathrm{R} \qty[(\ketbra{k}_\mathrm{R} \otimes \mathbb{I}_\mathrm{out} ) U \rho \, U^\dagger]$ for all $\rho \in \mathcal{L}(\mathcal{H}_\mathrm{in})$,
  where $U : \mathcal{H}_\mathrm{in} \to \mathcal{H}_\mathrm{R} \otimes \mathcal{H}_\mathrm{out}$ be a unitary operator and $\mathcal{H}_\mathrm{in} \cong (\mathbb{C}^2)^{\otimes n_\mathrm{in}},\mathcal{H}_\mathrm{out} \cong (\mathbb{C}^2)^{\otimes n_\mathrm{out}}$ for some $n_\mathrm{in}, n_\mathrm{out} \in \mathbb{N}$.

  If the quantum instrument $\Lambda$ is an $(n_\mathrm{in}-T)$-qubit implementable instrument (with delayed inputs), then there exist POVMs $E^{(1)}, E^{(2)}, \cdots, E^{(T)}$ on $\mathcal{H}_\mathrm{in}$ such that
  \begin{itemize}
    \item The composability conditions hold: $E^{(1)} \lcircarrow E^{(2)} \lcircarrow \cdots \lcircarrow E^{(T)} \lcircarrow \AssocPOVM{\Lambda}$.
    \item Each $E^{(t)}$ satisfies the outcome no-signaling condition $\mathrm{A}_t \nrightarrow \mathrm{cl}$, where $\{\mathrm{A}_1, \mathrm{A}_2, \ldots, \mathrm{A}_T\}$ is an ordered subset of the input qubits.
    \item Each $E^{(t)} := \{ E^{(t)}_{k_t} \}_{k_t \in \Outk_t} $ is a POVM satisfying $\rank E^{(t)}_{k_t} \leq 2^{n_\mathrm{in} - t}$ for all $k_t \in \Outk_t$.
  \end{itemize}
\end{theorem}

\begin{proof}[Proof Sketch]
The full proof is given in \Cref{app:pf_thm_necc_qubit_reduction}; here we provide a sketch.

Since classical processing specified by a function $f$ can be viewed as either a relabeling (if $f$ is injective) or a grouping (if $f$ is not injective) of outcomes, any $m$-qubit implementable instrument (with delayed inputs) can, without loss of generality, be written so that the elementary classical-processing operation appears only once, at the final round of the sequence of elementary operations.

Assume $\Lambda$ is $(n_\mathrm{in}-T)$-qubit implementable (with delayed inputs). By \Cref{rem:stairs_wdi} and the above argument, $\Lambda$ admits the expression depicted in \Cref{fig:necc_stairs_wdi}.
Here, for each $t$, $\Gamma^{(t)}_{\mid k_{t-1}}:=\{\Gamma^{(t)}_{k_t\mid k_{t-1}}\}_{k_t\in\Outk_t}$ is $(n_\mathrm{in}-T)$-qubit implementable (without delayed inputs) and contains no elementary classical-processing operation.

\begin{figure}[H]
\centering
\includegraphics[width=0.99\linewidth]{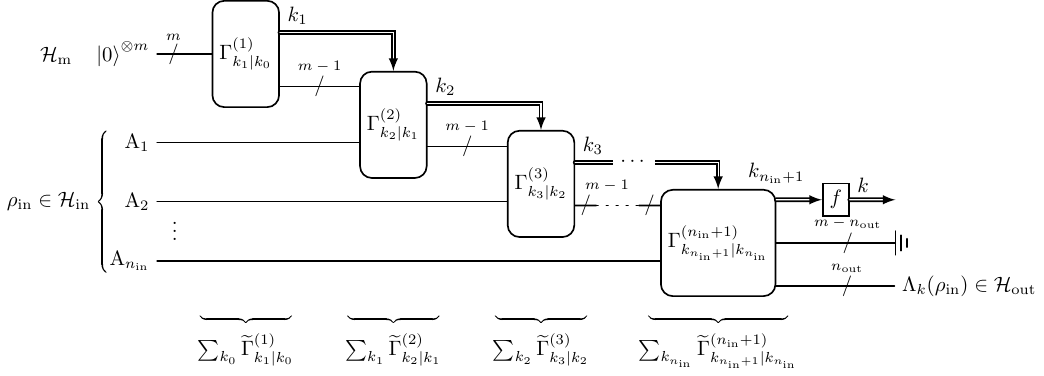}
\caption{An implementation of $\{\Lambda_k\}_k$ as an $(n_\mathrm{in}-T)$-qubit implementable instrument (with delayed inputs). Each $\{\Gamma^{(t)}_{k_t \mid k_{t-1}}\}_{k_t}$ is $(n_\mathrm{in}-T)$-qubit implementable (without delayed inputs) and contains no classical-processing operation. For brevity, we may write $m := n_\mathrm{in}-T$ in this figure.}
\label{fig:necc_stairs_wdi}
\end{figure}

For each $t\in [n_\mathrm{in}]$, define the accumulated instrument $\Xi^{(t)} := \{\Xi^{(t)}_{k_t}\}_{k_t\in\Outk_t}$ by
\begin{align}
  \Xi^{(t)}_{k_t}
  \;:=\;
  \sum_{k_0,\ldots,k_{t-1}}
  \widetilde{\Gamma}^{(t)}_{k_t\mid k_{t-1}}\circ \cdots \circ \widetilde{\Gamma}^{(1)}_{k_1\mid k_0}
  \qquad k_t\in\Outk_t.
  \label{eq:def_Xi}
\end{align}
The associated POVMs $\AssocPOVM{\Xi^{(t)}}$ for $t=n_\mathrm{in}-T+1,\ldots,n_\mathrm{in}$ satisfy all the required conditions in the theorem statement.
Composability follows from the instrument-level composability conditions $\Xi^{(n_\mathrm{in}-T+1)} \rcircarrow \cdots \rcircarrow \Xi^{(n_\mathrm{in})} \rcircarrow \Lambda $ together with \Cref{lem:pp_cond_indecomp_instr}, and the outcome no-signaling conditions follow from the same conditions on $\Xi^{(t)}$ plus \Cref{rem:ons_povm}.
Moreover, since no $\Gamma^{(t)}_{\mid k_{t-1}}$ contains classical processing, each $\Xi^{(t)}_{k_t}$ has Kraus rank~$1$; hence, by \Cref{lem:output_dim_indecomp}, the stated rank bound holds.

\end{proof}

\section{Application to Entanglement Distillation Protocols}\label{sec:ent_distill}
The primary application of \Cref{thm:suff_qubit_reduction,thm:necc_qubit_reduction} is entanglement distillation protocols~\cite{bennett1996purification, bennett1996mixedstateentanglement, matsumoto2003}.
Entanglement distillation protocols are bipartite protocols that transform multiple noisy Bell pairs into a smaller number of less-noisy Bell pairs by LOCC.
They are utilized to create high-fidelity entanglement between distant quantum processors in distributed quantum computing settings.
The local operations performed at each party in entanglement distillation protocols are a pertinent application of \Cref{thm:suff_qubit_reduction,thm:necc_qubit_reduction} for the following two reasons:
(i) Because the input to local operations of entanglement distillation protocols is a tensor product across qubits, the assumption that the inputs can be prepared sequentially is readily satisfied.
(ii) When these protocols are employed in distributed quantum computing settings, analyzing space requirements is crucial, since each quantum processor typically has a limited number of qubits.

In particular, we focus on entanglement distillation protocols based on stabilizer codes~\cite{matsumoto2003}.
From any $[[n,k]]$ stabilizer code, one can construct an entanglement distillation protocol that takes $n$ noisy Bell pairs as input and distills $k$ less-noisy Bell pairs.
In this protocol, each bipartite party performs a quantum instrument $\Lambda^{\mathrm{dist}} := \{\Lambda^{\mathrm{dist}}_s : \mathcal{L}(\mathcal{H}_\mathrm{in}) \to \mathcal{L}(\mathcal{H}_\mathrm{out})\}_{s \in \mathbb{F}_2^{n-k}}$ defined as
\begin{align}
  \Lambda^{\mathrm{dist}}_s(\rho) := \Tr_\mathrm{R}\!\qty[( \ketbra{s}_\mathrm{R} \otimes \mathbb{I}_\mathrm{out} )\, U_{\mathrm{enc}}^\dagger \rho \, U_{\mathrm{enc}}]
  \qquad \forall \rho \in \mathcal{L} (\mathcal{H}_\mathrm{in}), \label{eq:target_instr_ent}
\end{align}
where $\mathcal{H}_\mathrm{in} \cong (\mathbb{C}^2)^{\otimes n}, \mathcal{H}_\mathrm{out} \cong (\mathbb{C}^2)^{\otimes k}$, and $\mathcal{H}_\mathrm{R} \cong (\mathbb{C}^2)^{\otimes (n-k)}$.
Here, $U_{\mathrm{enc}}: \mathcal{H}_\mathrm{in} \to \mathcal{H}_\mathrm{R} \otimes \mathcal{H}_\mathrm{out} $ is the encoding unitary of the underlying stabilizer code, which satisfies
\begin{align}
  U_{\mathrm{enc}} \big(  Z_i \otimes \mathbb{I}_\mathrm{out} \big) U_{\mathrm{enc}}^\dagger =  g_i  \qquad \forall i \in \{1, \cdots, n-k\},
\end{align}
where $Z_i$ is the Pauli-$Z$ operator acting on the $i$-th qubit of $\mathcal{H}_\mathrm{R}$.
The measurement outcome $s$ corresponds to the error syndrome of the stabilizer code. After applying the instrument $\Lambda^{\mathrm{dist}}$, the parties communicate their measurement outcomes and perform a recovery operation according to the combined error syndrome.
For an $[[n,k]]$ stabilizer code, because $U_{\mathrm{enc}}$ acts on an $n$-qubit system, an implementation of $\Lambda^{\mathrm{dist}}$ without any space-reduction techniques requires $n$ qubits.

As \Cref{thm:no_gap_ent_distill} below states, for the quantum instrument $\Lambda^{\mathrm{dist}}$ defined above, the necessary conditions in \Cref{thm:necc_qubit_reduction} are also sufficient, and hence the smallest $(n-T)$ satisfying the conditions in \Cref{thm:necc_qubit_reduction} gives the optimal number of qubits for implementing $\Lambda^{\mathrm{dist}}$.
\begin{theorem} \label{thm:no_gap_ent_distill}
  Let $\mathcal{C}$ be an $[[n,k]]$ stabilizer code with stabilizer generators $\{g_1, \ldots, g_{n-k}\}$ and let $U_{\mathrm{enc}}: \mathcal{H}_\mathrm{in} \to \mathcal{H}_\mathrm{R} \otimes \mathcal{H}_\mathrm{out}$ be the encoding unitary of $\mathcal{C}$ satisfying $U_{\mathrm{enc}} \big(  Z_i \otimes \mathbb{I}_\mathrm{out} \big) U_{\mathrm{enc}}^\dagger = g_i$ for all $i \in \{1, \cdots, n-k\}$, where $\mathcal{H}_\mathrm{in} \cong (\mathbb{C}^2)^{\otimes n}$, $\mathcal{H}_\mathrm{out} \cong (\mathbb{C}^2)^{\otimes k}$, and $\mathcal{H}_\mathrm{R} \cong (\mathbb{C}^2)^{\otimes (n-k)}$.
  Let $\Lambda^{\mathrm{dist}} := \{\Lambda^{\mathrm{dist}}_s : \mathcal{L}(\mathcal{H}_\mathrm{in}) \to \mathcal{L}(\mathcal{H}_\mathrm{out})\}_{s \in \mathbb{F}_2^{n-k}}$ be the quantum instrument defined by
  \begin{align}
   \Lambda_s^{\mathrm{dist}}(\rho) := \Tr_\mathrm{R} \qty[( \ketbra{s}_\mathrm{R} \otimes \mathbb{I}_\mathrm{out} )\, U_{\mathrm{enc}}^\dagger \rho \, U_{\mathrm{enc}}]
   \qquad \forall \rho \in \mathcal{L}(\mathcal{H}_\mathrm{in}), \quad s \in \mathbb{F}_2^{n-k}.
  \end{align}

  For any fixed $T \in \{1, 2, \ldots, n\}$, if there exist POVMs $E^{(t)}$ for $t = 1, 2, \cdots, T$ that satisfy the conditions for $\Lambda^{\mathrm{dist}}$ stated in \Cref{thm:necc_qubit_reduction}, then there exist projective measurements $P^{(t)}$ for $t = 1, 2, \cdots, T$ that satisfy the conditions for $\Lambda^{\mathrm{dist}}$ stated in \Cref{thm:suff_qubit_reduction}.
\end{theorem}

\begin{remark}
The definition of space-constrained implementability (\Cref{def:m_space_instr_wdi}) requires that, for \emph{every} input state $\rho \in \mathcal{L}(\mathcal{H}_\mathrm{in})$, the outputs $\Lambda_k(\rho)$ be realizable within the given space constraint. In entanglement distillation protocols, however, the input states are restricted to tensor products of $n$ noisy Bell pairs, so we only need to implement $\Lambda^{\mathrm{dist}}_{s}(\rho)$ on this restricted set of states. This restriction may permit smaller space requirements than lower bounds inferred from \Cref{thm:necc_qubit_reduction}.
\end{remark}

\begin{proof}[Proof Sketch]
  The full proof is given in \Cref{app:pf_thm_no_gap_ent_distill}; here, we provide a proof sketch.

Write $\Lambda := \Lambda^{\mathrm{dist}}$ for brevity.
Assume there exist POVMs $E^{(t)} := \{E^{(t)}_{s_t}\}_{s_t \in \Outs_t}$ for $t=1,\ldots,T$ satisfying \Cref{thm:necc_qubit_reduction}.
By composability, there are column-stochastic matrices $\nu^{(t)} := (\nu^{(t)}_{s_t,s})_{s_t,s}$ with
\begin{align}
  E^{(t)}_{s_t} \;=\; \sum_{s} \nu^{(t)}_{s_t, s}\, \AssocPOVM{\Lambda}_s .
  \label{eq:def_column_ent}
\end{align}
Since $\AssocPOVM{\Lambda}$ is a projective measurement and each element has rank $2^k$, the rank bound for $E^{(t)}$ is equivalent to
\begin{align}
  \text{each row of $\nu^{(t)}$ has at most $2^{\,n-k-t}$ nonzero entries.}
  \label{eq:nonzero_count}
\end{align}

From the property of the encoding unitary, the associated POVM with $\Lambda$ can be expanded in terms of the generators:
\begin{equation}\label{eq:Ps-Fourier}
  \AssocPOVM{\Lambda}_s \;=\; \frac{1}{2^{n-k}}\sum_{r\in\mathbb F_2^{n-k}}(-1)^{s\cdot r}\, g^r,
  \qquad g^r:=\prod_{i=1}^{n-k} g_i^{\,r_{(i)}}.
\end{equation}
Applying \Cref{lem:fourier_support}, the outcome no-signaling constraints for $E^{(t)}$ are equivalent to the following
\emph{coset-constancy} of $\nu^{(t)}$: for each $s_t\in\Outs_t$,
\begin{align}
  \nu^{(t)}_{s_t,(s+\ell)} \;=\; \nu^{(t)}_{s_t,s}
  \quad \forall\, s\in\mathbb F_2^{n-k},\ \forall\, \ell \in \Outl_{[t,T]}^\perp,
  \label{eq:coset_constancy}
\end{align}
where
\begin{align}
  \Outl_{[t,T]}^\perp \;=\; \operatorname{span}\{x_\tau,z_\tau:\ \tau=t,\ldots,T\},
\end{align}
and $x_\tau,z_\tau\in\mathbb F_2^{n-k}$ are the $\tau$-th binary columns of the check matrix.
In words, each row of $\nu^{(t)}$ has the same entries in every coset of $\Outl_{[t,T]}^\perp$.

From \Cref{eq:nonzero_count} and \Cref{eq:coset_constancy} we deduce
\begin{align}
  \dim \Outl_{[t,T]}^\perp \;\le\; n-k-t \qquad \text{for each } t\in[T],
\end{align}
because, under coset constancy, the size of a coset must be no greater than the upper bound on the number of nonzero entries.

Once this dimension condition holds, we can construct new $0/1$ column-stochastic matrices $\mu^{(t)}$ for $t = 1,\cdots,T$ such that each row satisfies the coset-constancy condition, satisfying the same conditions as \Cref{eq:coset_constancy} and has exactly $2^{n-k-t}$ ones. These matrices yield projective measurements by $P^{(t)}_{s'_t} = \sum_s \mu^{(t)}_{s'_t, s} \AssocPOVM{\Lambda}_s$ required in \Cref{thm:suff_qubit_reduction}.
\end{proof}

For several well-known stabilizer codes, we compute the optimal space requirements for the corresponding instruments $\Lambda^{\mathrm{dist}}$ as the smallest $n-T$ such that $\Lambda^{\mathrm{dist}}$ is $(n-T)$-qubit implementable (with delayed inputs).
The results are summarized in \Cref{tab:ent_distill}.

\begin{table}[ht]
\centering
\caption{Optimal number of qubits for implementing the entanglement distillation instrument $\Lambda^{\mathrm{dist}}$ for several stabilizer codes.}
\label{tab:ent_distill}
\renewcommand{\arraystretch}{1.2}
\begin{tabular}{@{}c@{\hspace{10mm}}c@{}}
\hline\hline
\makecell{Underlying stabilizer code} &
\makecell{Optimal qubit requirements} \\
\hline
$[[5,1,3]]$ code~\cite{Laflamme1996perfectqec,nielsen2000}       & 4 \\
$[[7,1,3]]$ Steane code~\cite{Steane1996errorcorrectingcodes,nielsen2000} & 4 \\
$[[9,1,3]]$ Shor code~\cite{Shor1995decoherence,nielsen2000}   & 3 \\
\hline\hline
\end{tabular}
\end{table}

\newpage
\printbibliography

\clearpage
\addtocontents{toc}{\protect\setcounter{tocdepth}{1}}
\appendix
\section{Technical lemmas}

\subsection{\Cref{lem:isometry_equivalence_of_kraus_op}}

\begin{lemma} \label{lem:isometry_equivalence_of_kraus_op}
  Let $\mathcal{H}_1,\mathcal{H}_2,\mathcal{H}_3$ be finite-dimensional Hilbert spaces.
  If operators $A : \mathcal{H}_1 \to \mathcal{H}_2$ and $B : \mathcal{H}_1 \to \mathcal{H}_3$ satisfy $A^\dagger A = B^\dagger B$, then there exists a unitary operator $V : \operatorname{Ran}(A) \to \operatorname{Ran}(B)$ such that $B = V A$.
  Especially, if $ \dim \mathcal{H}_2 = \dim \mathcal{H}_3$, there exists a unitary operator $U: \mathcal{H}_2 \to \mathcal{H}_3$ such that $B = U A$.
\end{lemma}

\begin{proof}
Define $V : \operatorname{Ran}(A) \to \operatorname{Ran}(B)$ by
\begin{align}
  V(Ax)\;:=\;Bx
  \qquad \forall \, x\in\mathcal{H}_1.
\end{align}
This is well-defined because the value of $V$ is uniquely determined for each input: If $Ax=Ay$, then $
  0=\|A(x-y)\|^2=\langle x-y,\,A^\dagger A(x-y)\rangle
  =\langle x-y,\,B^\dagger B(x-y)\rangle=\|B(x-y)\|^2,$
so $Bx=By$.

Moreover, for any $x,y\in\mathcal{H}_1$ we have
\begin{align}
  \langle V(Ax),V(Ay)\rangle
  =\langle Bx,By\rangle
  =\langle x,B^\dagger B\,y\rangle
  =\langle x,A^\dagger A\,y\rangle
  =\langle Ax,Ay\rangle.
\end{align}
Thus $V$ is an isometry from $\operatorname{Ran}(A)$ to $\operatorname{Ran}(B)$.

Since $A^\dagger A=B^\dagger B$, we have $\|Ax\|=\|Bx\|$ for all $x\in\mathcal H_1$, hence $\ker A=\ker B$ and therefore
\begin{align}
\dim\operatorname{Ran}(A)=\dim\operatorname{Ran}(B) < \infty.
\end{align}
Therefore, $V$ is a surjective isometry, that is, a unitary operator.

\noindent\textbf{Finite-dimensional equal-dimension case:}

Write the orthogonal decompositions
$
\mathcal H_2=\operatorname{Ran}(A)\oplus \operatorname{Ran}(A)^\perp, \,
\mathcal H_3=\operatorname{Ran}(B)\oplus \operatorname{Ran}(B)^\perp .
$
From the equalities above and $\dim\mathcal H_2=\dim\mathcal H_3$ we get
\begin{align}
\dim\operatorname{Ran}(A)^\perp=\dim\operatorname{Ran}(B)^\perp,
\end{align}
which implies the existence of a unitary operator
$
W:\operatorname{Ran}(A)^\perp \longrightarrow \operatorname{Ran}(B)^\perp.
$
Define $U:\mathcal{H}_2 \to \mathcal{H}_3$ by
\begin{align}
U:=V\oplus W.
\end{align}
Then $U$ is unitary and, since $A(\mathcal H_1)\subseteq\operatorname{Ran}(A)$,
\begin{align}
UAx=V(Ax)=Bx
\qquad \forall \, x\in\mathcal H_1,
\end{align}
so $B=UA$ with $U$ unitary.

\end{proof}

\subsection{\Cref{lem:fourier_support}}

\begin{lemma}\label{lem:fourier_support}
Let $f:\mathbb{F}_2^m\to\mathbb{R}$ and define its unnormalized Fourier transform by
\begin{align}
\widetilde f(j)\;=\;\sum_{i\in\mathbb{F}_2^m} (-1)^{\,i\cdot j}\,f(i)
\qquad \forall \, j\in\mathbb{F}_2^m,
\end{align}
where $i\cdot j=\sum_{k=1}^m i_{(k)} j_{(k)}\in\mathbb{F}_2$ is the standard inner product over $\mathbb F_2$. Let $L$ be a linear subspace of $\mathbb{F}_2^m$ and $L^\perp=\{v\in\mathbb{F}_2^m:\ v\cdot \ell=0\ \forall\,\ell \in L\}$ its orthogonal complement. Then the following are equivalent:
\begin{enumerate}
\item[(i)] $\widetilde f(j)=0$ for all $j\notin L$ (i.e., the Fourier support of $f$ is contained in $L$).
\item[(ii)] $f(x+v)=f(x)$ for all $x\in\mathbb{F}_2^m$ and all $v\in L^\perp$ (equivalently, $f$ is constant on every coset of $L^\perp$).
\end{enumerate}
\end{lemma}

\begin{proof}
The inverse transform (with normalization $2^{-m}$) is
\begin{align}
f(x)=2^{-m}\sum_{j\in\mathbb{F}_2^m} (-1)^{\,x\cdot j}\,\widetilde f(j)
\qquad  x\in\mathbb{F}_2^m.
\end{align}

\smallskip
\noindent (i)$\Rightarrow$(ii): If $\widetilde f(j)=0$ for $j\notin L$, then for any $v\in L^\perp$,
\begin{align}
f(x+v)\;=\;2^{-m}\sum_{j\in L} (-1)^{(x+v)\cdot j}\,\widetilde f(j)
\;=\;2^{-m}\sum_{j\in L} (-1)^{x\cdot j}\underbrace{(-1)^{v\cdot j}}_{=1}\,\widetilde f(j)
\;=\;f(x),
\end{align}
since $v\cdot j=0$ for all $j\in L$.

\smallskip
\noindent (ii)$\Rightarrow$(i): Suppose $f$ is constant on each coset of $L^\perp$. Fix $j\notin L=(L^\perp)^\perp$. Then there exists $v_0\in L^\perp$ with $v_0\cdot j=1$. Partition $\mathbb{F}_2^m$ into cosets $C=x_0+L^\perp$. On any such $C$,
\begin{align}
\sum_{x\in C} (-1)^{x\cdot j} f(x)
= f(x_0)\sum_{v\in L^\perp} (-1)^{(x_0+v)\cdot j}
= f(x_0)(-1)^{x_0\cdot j}\sum_{v\in L^\perp} (-1)^{v\cdot j}.
\end{align}
Pairing $v$ with $v+v_0$ yields cancellation because $(-1)^{(v+v_0)\cdot j}=-(-1)^{v\cdot j}$. Hence $\sum_{v\in L^\perp} (-1)^{v\cdot j}=0$, so each coset contributes $0$, and therefore
\begin{align}
\widetilde f(j)=\sum_{x\in\mathbb{F}_2^m} (-1)^{x\cdot j} f(x)=0.
\end{align}
This holds for every $j\notin L$.
\end{proof}

\section{Additional proofs}
\label{app:pfs}
This appendix collects proofs and supplementary lemmas omitted from the main text.

\subsection{Proof of \Cref{rem:stairs_wdi}} \label{app:pf_stairs_space_inst_wdi}

Below, we show the equivalence between the expression in \Cref{eq:stairs_wdi_1} and the definition of $m$-qubit implementable instruments (with delayed inputs) (\Cref{def:m_space_instr_wdi}).
\begin{proof}
  Let $\{\Lambda_k : \mathcal{L} (\mathcal H_{\mathrm{in}}) \to \mathcal{L}(\mathcal H_{\mathrm{out}}) \}_{k \in \Outk}$
  be a quantum instrument where $\mathcal H_{\mathrm{in}} \cong (\mathbb C^2)^{\otimes n_{\mathrm{in}}}$ and
  $\mathcal H_{\mathrm{out}} \cong (\mathbb C^2)^{\otimes n_{\mathrm{out}}}$ for $n_{\mathrm{in}},n_{\mathrm{out}}\in \mathbb{Z}_{\geq 0}$.
  We will show that $\Lambda$ is an $m$-qubit implementable instrument (with delayed inputs) as in \Cref{def:m_space_instr_wdi} if and only if $\Lambda$ admits the expression \Cref{eq:stairs_wdi_1,eq:stairs_wdi_2}.

  First, we prove a property of the elementary input-loading instrument.

  \medskip
  \noindent\textbf{Decomposition into single-qubit loading instruments.}

  Consider an input-loading operation on $S = \{s_1, s_2, \ldots, s_{|S|}\}$ and $J = \{j_1, j_2, \ldots, j_{|S|}\}$.
  The transformation on a set of unnormalized states $\{ \rho_k \in \mathcal{H}_\mathrm{m} \otimes \mathcal{H}_\mathrm{in}\}_{k \in \Outk}$ by the input-loading operation on $S$ and $J$ can be decomposed as
  \begin{align}
    \{ \rho_k \}_{k \in \Outk}
    &\mapsto
    \qty{ \qty(\bra{x_1}_{s_1} \otimes \mathbb{I}) \,\rho_k\, \qty(\ket{x_1}_{s_1} \otimes \mathbb{I}) }_{k x_1 \in \Outk \times \{0,1\}} \\
    &\mapsto
    \qty{ \qty(\bra{x_1, x_2}_{s_1, s_2} \otimes \mathbb{I}) \,\rho_k\, \qty(\ket{x_1, x_2}_{s_1, s_2} \otimes \mathbb{I}) }_{k x_1 x_2 \in \Outk \times \{0,1\}^2} \\
    &\mapsto
    \cdots \\
    &\mapsto
    \qty{ \qty(\bra{x_1, x_2, \cdots, x_{|S|}}_{S} \otimes \mathbb{I}) \,\rho_k\, \qty(\ket{x_1, x_2, \cdots, x_{|S|}}_{S} \otimes \mathbb{I}) }_{k x \in \Outk \times \{0,1\}^{|S|}},
  \end{align}
  where $x_i \in \{0,1\}$ for all $i \in \{1, 2, \ldots, |S|\}$ and $x := (x_1, x_2, \ldots, x_{|S|})$.
  This decomposition is illustrated in \Cref{fig:pf_decomp_input_loading}.
  Each step loads one qubit, so we may assume, without loss of generality, that every input-loading operation in the definition of $m$-qubit implementable instruments (with delayed inputs) loads a single qubit.

  \begin{figure}[H]
    \centering
    \includegraphics[width=0.7\linewidth]{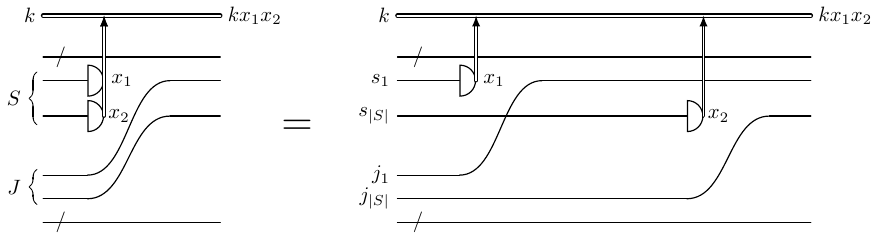}
    \caption{Decomposition of an input-loading operation into input-loading operations, each of which loads a single qubit. Thus, without loss of generality, we can assume that every input-loading operation in the definition of $m$-qubit implementable instruments (with delayed inputs) loads a single qubit.}
    \label{fig:pf_decomp_input_loading}
  \end{figure}

  \bigskip
  Now assume that $\Lambda$ is an $m$-qubit implementable instrument (with delayed inputs) as defined in \Cref{def:m_space_instr_wdi}, and thus
  \begin{align}
    \Lambda_k (\rho_{\mathrm{in}})
    =
    \Tr_{(\mathbb{C}^2)^{\otimes (m-n_\mathrm{out})}}
    \qty[
    \qty(
      \sum_{k_0, \cdots, k_{T-1}}
      \Phi^{(T)}_{k\mid k_{T-1}}
      \circ\cdots\circ
      \Phi^{(1)}_{k_1\mid k_0}
    )
    \qty( \ketbra{0}^{\otimes m} \otimes \rho_{\mathrm{in}} )
    ] \label{eq:stairs_wdi_pf_1}
    \quad \forall \rho_{\mathrm{in}} \in \mathcal{L}(\mathcal{H}_{\mathrm{in}}),
  \end{align}
  with the same notations as in \Cref{def:m_space_instr_wdi}. As shown above, we may assume that each input-loading operation loads one qubit.
  For $i \in \{1, 2, \ldots, n_\mathrm{in}\}$, let $t_i \in [T]$ be the round at which input $\text{A}_i$ is loaded, and set $t_0 := 0$ for notational convenience.

  For $i \in \{1, \ldots, n_\mathrm{in}\}$, set $\tau := t_{i-1}$ and $\tau' := t_{i}$ for brevity.
  Then the sequence of the elementary instruments from rounds $\tau + 1$ to $\tau'$ can be expressed by a single quantum instrument $\{\widetilde\Gamma^{(i)}_{k_{\tau'} \mid k_{\tau}}\}_{k_{\tau'}}$:
  \begin{align}
    \widetilde\Gamma^{(i)}_{k_{\tau'} \mid k_{\tau}}
    &:=
    \qty(
    \sum_{k_{\tau'-1}} \Phi^{(\tau')}_{k_{\tau'} \mid k_{\tau'-1}}
    )
    \circ
    \cdots
    \circ
    \qty(
    \sum_{k_{\tau}} \Phi^{(\tau+1)}_{k_{\tau+1} \mid k_{\tau}}).
  \end{align}
  By definition of the elementary instruments $\{\Phi^{(t)}_{k_t \mid k_{t-1}}\}_{k_t}$ in \Cref{def:m_space_instr_wdi}, the instrument $\{ \widetilde\Gamma^{(i)}_{k_{\tau'} \mid k_\tau}\}_{k_{\tau'}}$ acts trivially on $\mathcal{H}_{\mathrm{in}}$, which is $\mathcal{H}_{\mathrm{in}} = \mathcal{H}_{\text{A}_{i}}\otimes \ldots \otimes \mathcal{H}_{\text{A}_{n_\mathrm{in}}}$ between round $\tau +1$ and $\tau'$ by the update rule of $\mathcal{H}_\mathrm{in}$, and the part acting on $\mathcal{H}_\mathrm{m}$ is a composition of elementary instruments in the setting without delayed inputs (\Cref{def:space_instr_wo_delayed}), followed by a computational basis measurement on one qubit.
  Hence,
  \begin{align}
    \widetilde\Gamma^{(i)}_{k_{\tau'} \mid k_\tau} =
    \Gamma^{(i)}_{k_{\tau'} \mid k_\tau} \otimes \operatorname{id}_{\mathrm{A}_{i}, \cdots, \mathrm{A}_{n_\mathrm{in}}},
    \label{eq:pf_gammatilde_decomp}
  \end{align}
  where $\{\Gamma^{(i)}_{k_{\tau'} \mid k_\tau}\}_{k_{\tau'}}$ is an $m$-qubit implementable instrument (without delayed inputs) that has the input system $(\mathbb{C}^2)^{\otimes m}$ and the output system $(\mathbb{C}^2)^{\otimes (m-1)}$.
  For the sequence after round $ t_{n_\mathrm{in}}$, define
  \begin{align}
    \widetilde\Gamma^{(n_\mathrm{in}+1)}_{k  \mid k_{\tau}}
    &:=
    \sum_{k_{T-1}}
    \Phi^{(T)}_{k \mid k_{T-1}}
    \circ
    \cdots
    \circ
    \sum_{k_{\tau+1}} \Phi^{(\tau+1)}_{k_{\tau+1} \mid k_{\tau}}.
  \end{align}
  The same argument without the final computational basis measurement yields the decomposition in \Cref{eq:pf_gammatilde_decomp}, and here $\{ \widetilde\Gamma^{(n_\mathrm{in}+1)}_{k \mid k_\tau} \}_{k}$ is a quantum instrument that has both input and output systems $(\mathbb{C}^2)^{\otimes m}$.

  Substituting the above expressions into \Cref{eq:stairs_wdi_pf_1} and relabeling the outcome indices gives \Cref{eq:stairs_wdi_1,eq:stairs_wdi_2}.

  \bigskip
  Conversely, suppose $\Lambda$ is given by \Cref{eq:stairs_wdi_1,eq:stairs_wdi_2}.
  By the definition of $m$-qubit implementable instruments (without delayed inputs), each $\{\Gamma^{(t)}_{k_t \mid k_{t-1}} \}_{k_t \in \Outk_{t}}$ can be expressed as a composition of the elementary instruments as in \Cref{def:space_instr_wo_delayed}.
  Substituting these expressions into \Cref{eq:stairs_wdi_1,eq:stairs_wdi_2} yields a decomposition of $\Lambda$ as a composition of the elementary instruments in \Cref{def:m_space_instr_wdi}, and thus $\Lambda$ is an $m$-qubit implementable instrument (with delayed inputs).
\end{proof}

\subsection{Proof of \Cref{lem:output_dim_indecomp}} \label{app:pf_lem_output_dim_indecomp}

\begin{proof}
Let $E=\{E_k\}_{k \in \Outk}\subseteq \mathcal{L}(\mathcal{H}_{\mathrm{in}})$ be a POVM, and set
$r_\ast \coloneqq \max_{k \in \Outk}\operatorname{rank}(E_k)$.

\smallskip
\noindent\textbf{Existence.}
Fix a Hilbert space $\mathcal H_{\mathrm{out}}$ with $\dim\mathcal H_{\mathrm{out}}=r \geq r_\ast$.
For each $k$, $\dim\operatorname{Ran}(\sqrt{E_k})=\operatorname{rank}(E_k) \leq r_\ast \leq r$.
Thus, we can choose an isometry
\begin{align}
J_k:\operatorname{Ran}(\sqrt{E_k})\longrightarrow \mathcal H_{\mathrm{out}}.
\end{align}
Extend $J_k$ by $0$ on $(\operatorname{Ran}(\sqrt{E_k}))^\perp \subseteq \mathcal H_{\mathrm{in}}$, so $J_k$ is a partial isometry on $\mathcal H_{\mathrm{in}}$ with
$J_k^\dagger J_k=\mathbb{I}_{\operatorname{Ran}(\sqrt{E_k})}$.
Define a quantum instrument $\Gamma=\{\Gamma_k\}_{k \in \Outk}$ by
\begin{align}
L_k:=J_k\,\sqrt{E_k}\;:\;\mathcal H_{\mathrm{in}}\to \mathcal H_{\mathrm{out}},
\qquad
\Gamma_k(\rho):=L_k\,\rho\,L_k^\dagger.
\end{align}
Each $\Gamma_k$ has Kraus rank 1. Moreover,
\begin{align}
L_k^\dagger L_k
= \sqrt{E_k}\,J_k^\dagger J_k\,\sqrt{E_k}
= \sqrt{E_k}\,\mathbb{I}_{\operatorname{Ran}(\sqrt{E_k})}\,\sqrt{E_k}
= \sqrt{E_k}\,\sqrt{E_k}
= E_k,
\end{align}
so the associated POVM of $\Gamma$ is $\{E_k\}_k$. Finally,
$\sum_k L_k^\dagger L_k=\sum_k E_k=\mathbb I_{\mathcal H_{\mathrm{in}}}$ shows
$\sum_k\Gamma_k$ is trace preserving, i.e.\ $\Gamma$ is a valid quantum instrument with $\dim\mathcal H_{\mathrm{out}}=r$.

\smallskip
\noindent\textbf{Optimality.}
Let $\widehat\Gamma=\{\widehat\Gamma_k : \mathcal{L}(\mathcal H_{\mathrm{in}}) \to \mathcal{L}(\widehat{\mathcal H}_{\mathrm{out}}) \}_k$ be any quantum instrument whose associated POVM is $E$ and for which each $\widehat\Gamma_k$ has Kraus rank 1.
Then each $\widehat\Gamma_k$ has a single Kraus operator $\widehat L_k:\mathcal H_{\mathrm{in}}\to \widehat{\mathcal H}_{\mathrm{out}}$ and
\begin{align}
\widehat L_k^\dagger \widehat L_k = E_k.
\end{align}
Hence\footnote{For any linear map $L$, $\ker(L^\dagger L)=\ker(L)$ because for any $v \in \ker (L^\dagger L)$ we have $0 = \langle v, L^\dagger L v\rangle=\|Lv\|^2$. Since $L$ and $L^\dagger$ have the same domain, by the rank–nullity relationship, $\operatorname{rank}(L^\dagger L)=\operatorname{rank}(L)$; hence $\operatorname{rank}(E_k)=\operatorname{rank}(\widehat L_k)$.}
\begin{align}
\operatorname{rank}(E_k)=\operatorname{rank}(\widehat L_k)
\le \dim \widehat{\mathcal H}_{\mathrm{out}}\qquad \forall\; k  \in \Outk,
\end{align}
and taking the maximum over $k$ yields
$\dim \widehat{\mathcal H}_{\mathrm{out}}\ge r_\ast$.
\end{proof}

\subsection{Proof of \Cref{lem:pp_cond_indecomp_instr}} \label{app:pf_pp_cond_indecomp_instr}

\begin{proof}
Let $\Lambda := \{\Lambda_k\}_{k \in \Outk}$ and $\Gamma := \{\Gamma_l\}_{l \in \Outl}$ be quantum instruments such that each $\Lambda_k$ and each $\Gamma_l$ has Kraus rank~1 for every $k \in \Outk$ and $l \in \Outl$. Suppose further that there exists a projective measurement $\{P_m\}_m$ such that $\AssocPOVM{\Lambda} \lcircarrow \{P_m\}_m$.
Since the direction $\Gamma \rcircarrow \Lambda \Rightarrow \AssocPOVM{\Gamma} \lcircarrow \AssocPOVM{\Lambda}$ follows from \Cref{lem:nec_cond_pp_indecomp_instr}, we prove the converse.

Since each $\Gamma_l$ and each $\Lambda_k$ has Kraus rank~1, there exist Kraus operators $L_l,K_k$ such that
$\Gamma_l(\rho)=L_l\rho L_l^\dagger$, $\Lambda_k(\rho)=K_k\rho K_k^\dagger$.
Set the elements of the associated POVMs as $F_l:=L_l^\dagger L_l$ and $E_k:=K_k^\dagger K_k$.
By polar decomposition, there exist partial isometries $V_l$ and $W_k$ such that
\begin{align}
L_l=V_l\sqrt{F_l},\qquad K_k=W_k\sqrt{E_k},
\end{align}
and $V_l^\dagger V_l=\Pi_{F_l}$ and $W_k^\dagger W_k=\Pi_{E_k}$ where $\Pi_{F_l}$ and $\Pi_{E_k}$ are the projectors onto $\operatorname{Ran}(F_l)$ and $\operatorname{Ran}(E_k)$, respectively.

From $\qty{F_l}_l \lcircarrow \qty{E_k}_k \lcircarrow \qty{P_m}_m$, there exist column-stochastic matrices $\qty(\mu_{k,m})_{k,m}$ and $\qty(\nu_{l,k})_{l,k}$ such that
\begin{align}
  E_k = \sum_m \mu_{k,m}\,P_m,
\qquad
  F_l = \sum_k \nu_{l,k}\,E_k
      = \sum_m \qty(\sum_k \nu_{l,k}\,\mu_{k,m} ) P_m
      =: \sum_m \tau_{l,m}\,P_m .
\end{align}
Note that if $\tau_{l,m}=0$ then necessarily $\nu_{l,k} \mu_{k,m}=0$ for all $k$ since all entries are non–negative.

Define, for each $(k,l)$,
\begin{align}
X_{k\mid l}:=\sum_{m:\,\tau_{l,m}>0}\!\sqrt{\frac{\mu_{k,m}\nu_{l,k}}{\tau_{l,m}}}\,P_m,
\qquad
K_{k\mid l}:=W_k\,X_{k\mid l}\,V_l^\dagger.
\end{align}
Then a quantum instrument $\{\Theta_{k\mid l}(\sigma):=K_{k\mid l}\sigma K_{k\mid l}^\dagger\}_{k \in \Outk}$ for each $l \in \Outl$ satisfies
\begin{align}
\sum_{l}\Theta_{k\mid l} \circ \Gamma_l(\rho)
=\Lambda_k(\rho)
\qquad \forall \, k \in \Outk ,\,\rho \in \mathcal{L}(\mathcal{H}_{\mathrm{in}}).
\end{align}
Indeed,
\begin{align}
\sum_l \Theta_{k\mid l}\!\circ\!\Gamma_l(\rho)
&= \sum_l K_{k\mid l}\,L_l\,\rho\,L_l^\dagger \,K_{k\mid l}^\dagger \\
&= \sum_l W_k X_{k\mid l} V_l^\dagger V_l \sqrt{F_l}\,\rho\,\sqrt{F_l} V_l^\dagger V_l X_{k\mid l}^\dagger W_k^\dagger \\
&= W_k\qty(\sum_l X_{k\mid l}\sqrt{F_l}\,\rho\,\sqrt{F_l}X_{k\mid l}^\dagger)W_k^\dagger,
\end{align}
since $V_l^\dagger V_l=\Pi_{F_l}$ and $\Pi_{F_l}\sqrt{F_l}=\sqrt{F_l}$. Now compute the middle sum,
\begin{align}
X_{k\mid l}\sqrt{F_l}
&=\sum_{m:\,\tau_{l,m}>0}\sum_{m'}\sqrt{\frac{\mu_{k,m}\nu_{l,k}}{\tau_{l,m}}}\,P_m\,\sqrt{\tau_{l,m'}}\,P_{m'}
= \sum_m \sqrt{\mu_{k,m}\nu_{l,k}}\,P_m.
\end{align}
 Note that if $\tau_{l,m}=0$ then necessarily $\mu_{k,m}\nu_{l,k}=0$ for all $k$; hence we may freely extend sums over $m$ to all indices.
Summing over $l$ and using column–stochasticity $\sum_l \nu_{l,k}=1$,
\begin{align}
\sum_l X_{k\mid l}\sqrt{F_l}\,\rho\,\sqrt{F_l}X_{k\mid l}^\dagger
= \sum_{m,n}\sqrt{\mu_{k,m}\mu_{k,n}}\,P_m\,\rho\,P_n
= \Bigl(\sum_m \sqrt{\mu_{k,m}}\,P_m\Bigr)\rho \Bigl(\sum_n \sqrt{\mu_{k,n}}\,P_n\Bigr)
= \sqrt{E_k}\,\rho\,\sqrt{E_k}.
\end{align}
Therefore,
\begin{align}
\sum_l \Theta_{k\mid l}\!\circ\!\Gamma_l(\rho)
= W_k\,\sqrt{E_k}\,\rho\,\sqrt{E_k}\,W_k^\dagger
= K_k\,\rho\,K_k^\dagger
= \Lambda_k(\rho),
\end{align}
as claimed.

\medskip
\noindent\textbf{Instrument normalization via extra outcomes.}
Each instrument $\qty{\Theta_{k\mid l}}_{k \in \Outk}$ is not necessarily trace-preserving on all of $\mathcal H_{\mathrm{mid}}$, but it is on $\mathrm{Ran}\,L_l$. Indeed,\footnote{
  \Cref{eq:pi_mk_reduction} holds since $X_{k\mid l}=\sum_{m:\,\tau_{l,m}>0}\sqrt{\frac{\mu_{k,m}\nu_{l,k}}{\tau_{l,m}}}\,P_m$ and $E_k=\sum_m \mu_{k,m}P_m$,  $\operatorname{Ran} X_{k\mid l} \subseteq \operatorname{Ran}E_k$; hence $\Pi_{E_k}X_{k\mid l}=X_{k\mid l}$.
}
\begin{align}
\sum_{k \in \Outk} K_{k\mid l}^\dagger K_{k\mid l}
&=\sum_k V_l\,X_{k\mid l}^\dagger\,\underbrace{W_k^\dagger W_k}_{=\Pi_{E_k}}\,X_{k\mid l}\,V_l^\dagger\\
&= V_l\underbrace{\qty( \sum_k X_{k\mid l}^\dagger X_{k\mid l})}_{=\Pi_{F_l}} V_l^\dagger \label{eq:pi_mk_reduction}\\
&= V_l V_l^\dagger \\
&= \Pi_{L_l}.
\end{align}
The last equality uses that $V_l$ is the partial isometry in the polar decomposition of $L_l$, so $V_lV_l^\dagger$ is the projector onto $\mathrm{Ran}\,L_l$. Therefore $\{\Theta_{k\mid l}\}_{k \in \Outk}$ is trace-preserving on $\mathrm{Ran}\,L_l\subseteq\mathcal H_{\mathrm{mid}}$, which practically suffices since it is applied after $\Gamma_l(\rho)=L_l\rho L_l$.
If one insists on global trace preservation on $\mathcal H_{\mathrm{mid}}$, choose a finite index set $\Outk^0$ and operators $\{R_{k'\mid l}\}_{k'\in \Outk^0}$ with
$\sum_{k'\in \Outk^0} R_{k'\mid l}^\dagger R_{k'\mid l}=\mathbb{I}_{\mathcal H_{\mathrm{mid}}}-\Pi_{L_l}$, and add the extra outcomes
$\Theta_{k'\mid l}(\sigma):=R_{k'\mid l}\,\sigma\,R_{k'\mid l}^\dagger$. Replacing $\{\Theta_{k\mid l}\}_{k \in \Outk}$ by $\{\Theta_{j\mid l}\}_{j\in \Outk \cup \Outk^0}$ yields a trace-preserving instrument on $\mathcal H_{\mathrm{mid}}$ without changing the composed map $\sum_l \Theta_{k\mid l}\circ\Gamma_l$ (the added outcomes have zero-probability after composition).

\end{proof}

\subsection{Proof of \Cref{thm:ons_decomp}} \label{app:pf_thm_ons_decomp}

\begin{proof}
Let $\Lambda=\{\Lambda_k:\mathcal{L}(\mathcal{H}_\mathrm{A}\otimes\mathcal{H}_\mathrm{B}) \to \mathcal{L}(\mathcal{H}_\mathrm{C})\}_{k\in\Outk}$ be a quantum instrument.

\medskip
\noindent\textbf{(b) $\Rightarrow$ (a).}
Let $\{F_k\}_{k\in\Outk}$ be the POVM associated with $\Gamma$, i.e.\ $\Tr[\Gamma_k(\sigma)]=\Tr[F_k \sigma]$ for all $\sigma \in \mathcal{L}(\mathcal{H}_\mathrm{A})$. Since each $\mathcal{E}^{(k)}$ is trace-preserving,
\begin{align}
  \Tr[\Lambda_k(\rho)]
  \;=\; \Tr\!\bigl[ (\Gamma_k \otimes \operatorname{id}_\mathrm{B})(\rho) \bigr]
  \;=\; \Tr\!\bigl[ (F_k \otimes \mathbb{I}_\mathrm{B})\, \rho \bigr],
\end{align}
for all $\rho \in \mathcal{L}(\mathcal{H}_\mathrm{A}\otimes\mathcal{H}_\mathrm{B})$. Hence $\mathrm{B} \nrightarrow \mathrm{cl}$ holds.

\bigskip
\noindent\textbf{(a) $\Rightarrow$ (b).}
Assume $\mathrm{B} \nrightarrow \mathrm{cl}$. Then there exists a POVM $\{F_k\}_{k\in\Outk}$ on $\mathcal{H}_\mathrm{A}$ such that the POVM associated with $\Lambda$ factorizes as
\begin{equation}\label{eq:Ak-factor}
  A_k \;:=\; \AssocPOVM{\Lambda}_k \;=\; F_k \otimes \mathbb{I}_\mathrm{B} \qquad \forall\,k\in\Outk.
\end{equation}
Fix $k\in\Outk$ and choose a Kraus representation of $\Lambda_k$:
\begin{equation}\label{eq:LambdaKraus}
  \Lambda_k(\rho) \;=\; \sum_i K_{k,i}\, \rho \, K_{k,i}^\dagger,
  \qquad \sum_i K_{k,i}^\dagger K_{k,i} \;=\; A_k \;=\; F_k\otimes \mathbb{I}_\mathrm{B}.
\end{equation}

\medskip
\noindent\textbf{Construction of $\Gamma$ as the L\"uders instrument.}

Let $\Gamma$ be the L\"uders instrument for $\{F_k\}_k$ on $\mathcal{H}_\mathrm{A}$, and set $\mathcal{H}_\mathrm{X}=\mathcal{H}_\mathrm{A}$:
\begin{equation}\label{eq:LudersA}
  \Gamma_k(\sigma) \;=\; \sqrt{F_k}\, \sigma \, \sqrt{F_k} \qquad \forall\,\sigma \in \mathcal{L}(\mathcal{H}_\mathrm{A}).
\end{equation}
Then, for all $\rho \in \mathcal{L}(\mathcal{H}_\mathrm{A}\otimes\mathcal{H}_\mathrm{B})$,
\begin{equation}\label{eq:LudersAB}
  (\Gamma_k \otimes \operatorname{id}_\mathrm{B})(\rho)
  \;=\; (\sqrt{F_k}\otimes \mathbb{I}_\mathrm{B})\, \rho \, (\sqrt{F_k}\otimes \mathbb{I}_\mathrm{B})
  \;=\; \sqrt{A_k}\, \rho \, \sqrt{A_k}.
\end{equation}

\medskip
\noindent\textbf{Construction of $\mathcal{E}^{(k)}$.}

Let $P_k$ be the projector onto $\operatorname{Ran}(A_k)$ and let $A_k^{- 1/2}$ be the generalized inverse of $\sqrt{A_k}$, that is, for the spectral decomposition $\sqrt{A_k}=\sum_j \lambda_j \ketbra{\phi_j}$ with $\lambda_j\ge 0$, $A_k^{-1/2}:=\sum_{j:\,\lambda_j>0} \lambda_j^{-1} \ketbra{\phi_j}$.
Define
\begin{equation}\label{eq:Mki}
  M_{k,i} \;:=\; K_{k,i} \, A_k^{- 1/2}.
\end{equation}
Then
\begin{equation}\label{eq:Mnorm}
  \sum_i M_{k,i}^\dagger M_{k,i}
  \;=\; A_k^{- 1/2}\!\qty(\sum_i K_{k,i}^\dagger K_{k,i}\!)A_k^{- 1/2}
  \;=\; A_k^{- 1/2} A_k A_k^{- 1/2}
  \;=\; P_k.
\end{equation}
Choose any unit vector $\ket{\psi}\in\mathcal{H}_\mathrm{C}$ and an orthonormal basis $\{\ket{e_{k,\ell}}\}_\ell$ of $\ker A_k$. Define additional Kraus operators
\begin{equation}\label{eq:Nkl}
  N_{k,\ell} \;:=\; \ketbra{\psi}{e_{k,\ell}}.
\end{equation}
Then
\begin{align}
  \sum_i M_{k,i}^\dagger M_{k,i} + \sum_\ell N_{k,\ell}^\dagger N_{k,\ell}
  \;=\; P_k + (\mathbb{I}_{\mathrm{A, B}}-P_k) \;=\; \mathbb{I}_{\mathrm{A, B}}.
\end{align}
so the map $\mathcal{E}^{(k)}:\mathcal{L}(\mathcal{H}_\mathrm{A}\otimes\mathcal{H}_\mathrm{B})\to\mathcal{L}(\mathcal{H}_\mathrm{C})$ defined by the set of Kraus operators $\{M_{k,i}\}_i \cup \{N_{k,\ell}\}_\ell$ is a CPTP map on $\mathcal{H}_\mathrm{A}\otimes\mathcal{H}_\mathrm{B}$.

\medskip
\noindent\textbf{Verification of the composition identity.}

Using \Cref{eq:LudersAB} and $A_k^{- 1/2}\sqrt{A_k}=P_k$, while $N_{k,\ell} P_k=0$, we get
\begin{align*}
  \mathcal{E}^{(k)}\!\bigl( (\Gamma_k \otimes \operatorname{id}_\mathrm{B})(\rho) \bigr)
  &= \sum_i M_{k,i}\, \sqrt{A_k}\, \rho \, \sqrt{A_k}\, M_{k,i}^\dagger \\
  &= \sum_i K_{k,i}\, A_k^{- 1/2}\sqrt{A_k}\, \rho \, \sqrt{A_k} A_k^{- 1/2} K_{k,i}^\dagger \\
  &= \sum_i K_{k,i}\, \rho \, K_{k,i}^\dagger \\
  &= \Lambda_k(\rho),
\end{align*}
for all $\rho$.
Since the construction holds for each $k\in\Outk$, the desired decomposition follows.
\end{proof}

\subsection{Proof of \Cref{lem:ons_decomp_mindim}} \label{app:pf_lem_ons_decomp}

\begin{proof}
Let $\mathcal{H}_\text{A} \cong (\mathbb{C}^{2})^{\otimes n_\text{A}}, \mathcal{H}_\text{B} \cong (\mathbb{C}^{2})^{\otimes n_\text{B}}, \mathcal{H}_\text{C} \cong (\mathbb{C}^{2})^{\otimes n_\text{C}}$ for some $n_\text{A}, n_\text{B}, n_\text{C} \in \mathbb{Z}_{\geq 0}$ and $\Lambda := \{\Lambda_k : \mathcal{L}(\mathcal{H}_\text{A} \otimes \mathcal{H}_\text{B}) \to \mathcal{L}(\mathcal{H}_\text{C})\}_{k \in \Outk}$ be a quantum instrument where each $\Lambda_k$ has Kraus rank~1 for every $k \in \Outk$.

\medskip
\noindent\textbf{(b) $\Rightarrow$ (a).}
This is a special case of \Cref{thm:ons_decomp}: taking the trace on both sides of the equality in (b) immediately yields the outcome no-signaling condition.

\medskip
\noindent\textbf{(a) $\Rightarrow$ (b).}
Assume $\Lambda$ satisfies $\mathrm{B}\nrightarrow\mathrm{cl}$. Then, by \Cref{rem:ons_povm}, there exists a POVM $\{F_k\}_{k\in\Outk}$ on $\mathcal{H}_\mathrm{A}$ such that the POVM associated with $\Lambda$ factorizes as
\begin{equation}\label{eq:Ak-effect}
\AssocPOVM{\Lambda}_k \;=\; F_k\otimes \mathbb{I}_\mathrm{B}\qquad \forall\,k\in\Outk.
\end{equation}

Because each $\Lambda_k$ has Kraus rank~$1$, \Cref{lem:output_dim_indecomp} applied to the POVM in \Cref{eq:Ak-effect} implies
\begin{equation}\label{eq:dim-ineq}
\dim\mathcal{H}_\mathrm{C} \;\ge\; \max_{k}\operatorname{rank}\!\big(\AssocPOVM{\Lambda}_k\big)
\;=\;\max_{k}\operatorname{rank}(F_k)\cdot \dim\mathcal{H}_\mathrm{B}.
\end{equation}
In particular, $\dim\mathcal{H}_\mathrm{C}\ge \dim\mathcal{H}_\mathrm{B}$. Since all spaces are qubit systems, we have
$\dim\mathcal{H}_\mathrm{C}/\dim\mathcal{H}_\mathrm{B}=2^{\,n_\mathrm{C}-n_\mathrm{B}}\in\mathbb{N}$.
Set $\mathcal{H}_\mathrm{X}$ so that
\begin{equation}\label{eq:defHX}
\dim\mathcal{H}_\mathrm{X} \;=\; \frac{\dim\mathcal{H}_\mathrm{C}}{\dim\mathcal{H}_\mathrm{B}}
\;=\; 2^{\,n_\mathrm{C}-n_\mathrm{B}}.
\end{equation}
By \Cref{eq:dim-ineq}, this choice ensures $\dim\mathcal{H}_\mathrm{X}\ge \max_k\operatorname{rank}(F_k)$.

Now apply \Cref{lem:output_dim_indecomp} to the POVM $\{F_k\}_k$:
we obtain an instrument $\Gamma=\{\Gamma_k:\mathcal L(\mathcal H_{\mathrm A})\to\mathcal L(\mathcal H_{\mathrm X})\}_{k \in \Outk}$ with associated POVM $\{F_k\}_k$ such that each $\Gamma_k$ has Kraus rank~$1$.
Thus there exists a single Kraus operator $L_k:\mathcal{H}_\mathrm{A}\to\mathcal{H}_\mathrm{X}$ with
\begin{equation}\label{eq:GammaKraus}
\Gamma_k(\sigma)=L_k\,\sigma\,L_k^\dagger,
\qquad
L_k^\dagger L_k = F_k
\qquad \forall\,k\in\Outk.
\end{equation}

Because each $\Lambda_k$ has Kraus rank~$1$, there exists a single Kraus operator
$K_k:\mathcal{H}_\mathrm{A}\otimes\mathcal{H}_\mathrm{B}\to\mathcal{H}_\mathrm{C}$ such that
\begin{equation}\label{eq:lambdaKraus}
\Lambda_k(\rho)=K_k\,\rho\,K_k^\dagger,
\qquad
K_k^\dagger K_k=\AssocPOVM{\Lambda}_k=F_k\otimes\mathbb{I}_\mathrm{B}.
\end{equation}
Combining \Cref{eq:GammaKraus} and \Cref{eq:lambdaKraus} yields
\begin{align}
K_k^\dagger K_k=(L_k\otimes\mathbb{I}_\mathrm{B})^\dagger(L_k\otimes\mathbb{I}_\mathrm{B}).
\end{align}
Hence, by \Cref{lem:isometry_equivalence_of_kraus_op}, there exists a unitary
$U_k:\mathcal{H}_\mathrm{X}\otimes\mathcal{H}_\mathrm{B}\to\mathcal{H}_\mathrm{C}$ (the dimensions match by \Cref{eq:defHX}.) such that
\begin{align}
K_k = U_k\,(L_k\otimes\mathbb{I}_\mathrm{B}).
\end{align}
Finally,
\begin{align}
\Lambda_k(\rho)
=K_k\,\rho\,K_k^\dagger
=U_k\,(L_k\otimes\mathbb{I}_\mathrm{B})\,\rho\,(L_k^\dagger\otimes\mathbb{I}_\mathrm{B})\,U_k^\dagger
=U_k\,(\Gamma_k\otimes\operatorname{id}_\mathrm{B})(\rho)\,U_k^\dagger,
\end{align}
which is the desired form in (b).
\end{proof}

\subsection{Proof of \Cref{lem:povm_space}} \label{app:pf_povm_space}

\begin{proof}
  Let $E = \qty{E_k}_{k \in  K}$ be a POVM on $\mathcal{H} \cong (\mathbb{C}^2)^{\otimes (m-1)}$. If necessary, enlarge the outcome set by adding $E_k = 0$ so that $|K| = 2^{T}$ for some $T \in \mathbb{N}$, and let $k=k_{(1)} \concat k_{(2)} \concat \cdots \concat k_{(T)}\in\{0,1\}^{T}$ be its binary expression.

  In what follows, we show that the POVM $E$ can be written as
  \begin{align}
    \Tr[E_k \rho]
  = \Tr[
    \Gamma^{(T)}_{k_{(T)} \mid k_{(< T)}}\circ \cdots\circ
    \Gamma^{(2)}_{k_{(2)} \mid k_{(< 2)}}\circ
    \Gamma^{(1)}_{k_{(1)} \mid \emptyset}
    \qty(\rho)
    ]
    \qquad\forall \rho \in \mathcal{L}(\mathcal{H}),\, k \in  K,  \label{eq:pf_seq_const_povm}
  \end{align}
  where, for each round $t\in \{1,\cdots,T\}$ and previously observed binary string $k_{(< t)} := k_{(1)} \concat \cdots \concat k_{(t-1)} $, the quantum instrument $\{ \Gamma^{(t)}_{j \mid k_{(< t)}}:\mathcal L(\mathcal{H})\to \mathcal L(\mathcal{H}) \}_{j \in \{0,1\}}$ has a Kraus representation:
  \begin{align}
  \Gamma^{(t)}_{j \mid k_{(< t)}}(\rho)
  = K^{(t)}_{j\mid k_{(< t)}}\, \rho\, \qty(K^{(t)}_{j\mid k_{(< t)}})^\dagger,
  \qquad
  K^{(t)}_{j\mid k_{(< t)}} &:= \qty(\bra{j} \otimes \mathbb{I}_{\mathcal{H}})\, U^{(t)}_{k_{(< t)}} \qty(\ket{0}\otimes \mathbb{I}_{\mathcal{H}}),   \label{eq:pf_seq_const_povm_step}
  \end{align}
  for some unitary $U^{(t)}_{k_{(<t)}}$ on $\mathbb C^2\otimes\mathcal H$. See \Cref{fig:pf_povm_m_space_app} for an illustration of \Cref{eq:pf_seq_const_povm,eq:pf_seq_const_povm_step}.

\begin{figure}[H]
\centering
\includegraphics[width=0.7\linewidth]{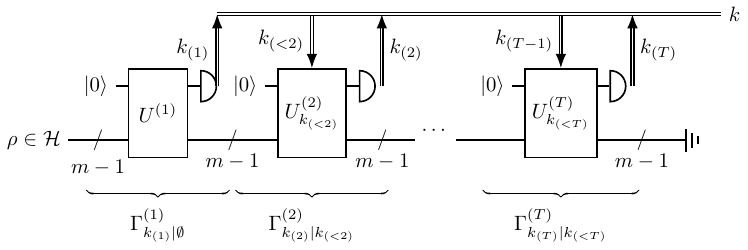}
\caption{Illustration of \Cref{eq:pf_seq_const_povm,eq:pf_seq_const_povm_step}. In each round $t \in \{1, \cdots, T\}$, the quantum instrument $\Gamma^{(t)}_{k_{(t)} \mid k_{(< t)}}$ appends an auxiliary qubit initialized to $\ket{0}$, applies an $m$-qubit unitary $U^{(t)}_{k_{(< t)}}$ depending on the previous measurement outcomes $k_{(< t)}$, and measures the auxiliary qubit in the computational basis.}
\label{fig:pf_povm_m_space_app}
\end{figure}

First, for any binary string $l$ of length $|l|<T$, define
\begin{align}
\mathcal K_l:=\{\,k\in\{0,1\}^{T}: \text{the first $|l|$ bits of $k$ are $l$}\,\},
\qquad
R_l \;:=\; \sum_{k\in \mathcal K_l} E_k \;\;(\ge 0),
\quad
R_{\emptyset}:=\mathbb{I}_\mathcal{H}.
\end{align}
Then $R_{l \concat 0}+R_{l \concat 1}=R_l$ for all $l$.

For round $t \in [T]$ and binary string $l\in\{0,1\}^{t-1}$, define
\begin{align}
K^{(t)}_{j\mid l}:=\sqrt{R_{l \concat j}}\;\sqrt{R_l^{-1}},\qquad
\Gamma^{(t)}_{j\mid l}(\rho):=K^{(t)}_{j\mid l}\,\rho\,(K^{(t)}_{j\mid l})^\dagger,
\end{align}
where $R_l^{-1}$ is the generalized inverse of $R_l$, that is, for the spectral decomposition $R_l=\sum_j \lambda_j \ketbra{\phi_j}$ with $\lambda_j\ge 0$, $R_l^{-1}:=\sum_{j:\,\lambda_j>0} \lambda_j^{-1} \ketbra{\phi_j}$.
Note that the set $\{\Gamma^{(t)}_{j \mid l}\}_{j \in \{0,1\}}$ forms a valid quantum instrument on $\operatorname{Ran}(R_l)$ because
\begin{align}
\sum_{j \in \{0, 1\}} \qty(K^{(t)}_{j\mid l} )^\dagger K^{(t)}_{j\mid l}
= \sum_{j \in \{0, 1\}} \sqrt{R_l^{-1}} R_{l \concat j} \sqrt{R_l^{-1}}
= \sqrt{R_l^{-1}} \qty(R_{l \concat 0}+R_{l \concat 1}) \sqrt{R_l^{-1}}
= \sqrt{R_l^{-1}} R_l \sqrt{R_l^{-1}}
= \Pi_{R_l},
\end{align}
where $\Pi_{R_l}$ is the projector onto $\operatorname{Ran}(R_l)$.
If necessary, one can make $\{\Gamma^{(t)}_{j \mid l}\}_{j \in \{0,1\}}$ a valid quantum instrument on the entire space $\mathcal{H}$ by adding an extra term $\sqrt{\mathbb{I}-\Pi_{R_l}}$ to the Kraus operator of one of the outcomes.  When $\{\Gamma^{(t)}_{j \mid l}\}_{j \in \{0,1\}}$ is applied after $\Gamma^{(t-1)}_{l_{(t-1)} \mid l_{(< t-1)}} \circ \cdots \circ \Gamma^{(1)}_{l_{(1)} \mid \emptyset}$, all states lie in $\operatorname{Ran}(R_l)$, so this additional term vanishes and contributes nothing. For notational simplicity, we therefore omit it in what follows.

In what follows, we will verify that $\{ \Gamma^{(t)}_{j \mid l} \}_{j \in \{0,1\}}$ actually satisfies \Cref{eq:pf_seq_const_povm,eq:pf_seq_const_povm_step}:

\medskip
\noindent\textbf{Verification of \Cref{eq:pf_seq_const_povm}.}
For round $t \in [T]$ and binary string $l \in \{0,1\}^{t}$,
the Kraus operator corresponding to the accumulated quantum instrument $
    \Gamma^{(t)}_{l_{(t)} \mid l_{(< t)}}\circ \cdots\circ
    \Gamma^{(2)}_{l_{(2)} \mid l_{(< 2)}}\circ
    \Gamma^{(1)}_{l_{(1)} \mid \emptyset}$ is
$
K^{\mathrm{acc}(t)}_l
:= K^{(t)}_{l_{(t)}\mid l_{(< t)}}\cdots K^{(2)}_{l_{(2)}\mid l_{(< 2)}}\,K^{(1)}_{l_{(1)}\mid \emptyset}.
$
We claim
\begin{equation}\label{eq:telescoping}
K^{\mathrm{acc}(t)}_l = \sqrt{R_l}
\qquad\text{for all }t\in[T],\ l\in\{0,1\}^t .
\end{equation}
This follows by induction on $t$. For $t=1$,
$K^{\mathrm{acc} (1)}_{j}= \sqrt{R_{j}} \sqrt{R_{\emptyset}^{-1}}=\sqrt{R_j}$ for $j\in\{0,1\}$.
Assuming the statement for $t-1$, and using $\operatorname{Ran}(R_{l \concat j})\subseteq\operatorname{Ran}(R_l)$,
\begin{align}
K^{\mathrm{acc}(t)}_{l \concat j}
=K^{(t)}_{j\mid l}\,K^{\mathrm{acc} (t-1)}_l
= \qty( \sqrt{R_{l \concat j}} \sqrt{R_l^{-1}}) \sqrt{R_l}
= \sqrt{R_{l \concat j}}\,\Pi_{R_l}
= \sqrt{R_{l \concat j}}.
\end{align}
Thus, the claim \Cref{eq:telescoping} holds. In particular, for $t=T$, that is, $k \in \{0,1\}^T$, we obtain
$
K^{\mathrm{acc}(T)}_k \;=\; \sqrt{R_k} \;=\; \sqrt{E_k}. \label{eq:Bleaf}
$
Therefore, for every state $\rho$,
\begin{align}
\Tr[ K^{\mathrm{acc}(T)}_k\,\rho\, \qty(K^{\mathrm{acc}(T)}_k)^\dagger ]
=\Tr[\sqrt{E_k}\,\rho\,\sqrt{E_k}]
=\Tr[E_k\rho],
\end{align}
which establishes \Cref{eq:pf_seq_const_povm}.

\medskip
\noindent\textbf{Verification of \Cref{eq:pf_seq_const_povm_step}:}
Since each round has two outcomes with Kraus operators $\{K^{(t)}_{j\mid l}\}_{j\in\{0,1\}}$, there exists
an isometry $V^{(t)}_l:\mathcal H\to \mathbb C^2\otimes\mathcal H$ satisfying
$V^{(t)}_l\ket{\phi}:=\sum_{j\in\{0,1\}}\ket{j}\otimes K^{(t)}_{j\mid l}\ket{\phi}$.
Extend $V^{(t)}_l$ to a unitary operator $U^{(t)}_l$ on $\mathbb C^2\otimes\mathcal H$. Then
\begin{equation}\label{eq:U-matrix-element}
(\bra{j}\otimes \mathbb{I}_{\mathcal{H}})\,U^{(t)}_l\,(\ket{0}\otimes \mathbb{I}_{\mathcal{H}})
=K^{(t)}_{j\mid l}\qquad (j\in\{0,1\}),
\end{equation}
which is exactly the Kraus representation in \Cref{eq:pf_seq_const_povm_step}.

\medskip
Finally, \Cref{eq:pf_seq_const_povm} can be rewritten as
  \begin{align}
    \Tr[E_k \rho]
    &= \Tr\qty[
      M_{k_{(T)}} U_{k_{(< T)}}
      \cdots
      M_{k_{(1)}} U^{(1)}
      \qty( \ketbra{0} \otimes \rho )
      {U^{(1)}}^\dagger
      { M_{k_{(1)}}}^\dagger
      \cdots
      {U_{k_{(< T)}}}^\dagger
      { M_{k_{(T)}}}^\dagger
    ]  \label{eq:pf_seq_const_povm_verify}
  \end{align}
  where $M_k := \ketbra{0}{k} \otimes \mathbb{I}_{\mathcal{H}}$. \Cref{eq:pf_seq_const_povm_verify} is a repetition of the elementary unitary operation and the elementary computational basis measurement specified in \Cref{eq:def_space_instr_wo_delayed}; therefore, the POVM $E$ is $m$-qubit implementable (without delayed inputs).
This completes the proof.

\end{proof}

\subsection{Proof of \Cref{thm:range_m_space_inst_wodi}} \label{app:pf_range_m_space_inst_wodi}

\begin{proof}
Assume there exist disjoint sets $\Outk_0, \Outk_1$ with $\Outk = \Outk_0 \cup \Outk_1$ and a projective measurement $\{P_b\}_{b \in \{0,1\}}$ on $\mathcal H_{\mathrm{in}}$ with
$\operatorname{rank}(P_b) \leq 2^{m-1}$ for each $b \in \{0,1\}$ such that the associated POVM $\{\AssocPOVM{\Lambda}_k\}_{k\in\Outk}$ of $\Lambda$ satisfies
\begin{align}
\sum_{k\in\Outk_b}\AssocPOVM{\Lambda}_k \;=\; P_b
\qquad \forall \, b\in\{0,1\}.
\end{align}

\noindent\textbf{Reduction to Kraus-rank-1 Instrument.}
Write a Kraus representation $\Lambda_k(\rho)=\sum_{\alpha_k} A_{k,\alpha_k}\rho A_{k,\alpha_k}^\dagger$.
Refine the outcome set to $\widetilde{\Outk}:=\{(k,\alpha_k):k\in\Outk\}$ and define
$\widetilde\Lambda_{(k,\alpha_k)}(\rho):=A_{k,\alpha_k}\rho A_{k,\alpha_k}^\dagger$,
which has Kraus rank~$1$ for every $(k,\alpha_k)$.
Its associated POVM satisfies
\begin{align}
\sum_{(k,\alpha_k):\,k\in\Outk_b}\!\AssocPOVM{\widetilde\Lambda}_{(k,\alpha_k)}
=\sum_{k\in\Outk_b}\sum_{\alpha_k} A_{k,\alpha_k}^\dagger A_{k,\alpha_k}
=\sum_{k\in\Outk_b} \AssocPOVM{\Lambda}_k
= P_b .
\end{align}
If $\widetilde\Lambda$ is $m$-qubit implementable (without delayed inputs), then so is $\Lambda$, since $\Lambda$ is obtained from $\widetilde\Lambda$ by the classical postprocessing $(k,\alpha_k)\mapsto k$.
Hence it suffices to treat the case where each $\Lambda_k$ has Kraus rank~$1$, i.e., $\Lambda_k(\rho)=A_k \rho A_k^\dagger$.

We now prove the claim via the decomposition illustrated in \Cref{fig:pf_half_cut}.
\begin{figure}[H]
\centering
\includegraphics[width=0.9\linewidth]{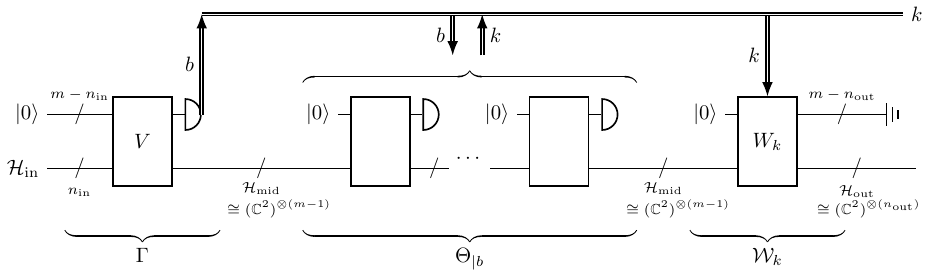}
\caption{Decomposition of the target instrument $\Lambda=\{\Lambda_k\}_{k\in\Outk}$ into three parts: the first instrument $\Gamma=\{\Gamma_b\}_{b\in\{0,1\}}$, the intermediate instrument $\Theta_{\mid b}=\{\Theta_{k\mid b}\}_{k\in\Outk_b}$, and the final channel $\mathcal{W}_{k}$. Each box in the figure represents an $m$-qubit unitary operation.}
\label{fig:pf_half_cut}
\end{figure}

\medskip
\noindent\textbf{First instrument $\Gamma$.}
Define a quantum instrument $\{\Gamma_b:\mathcal{L}(\mathcal{H}_{\mathrm{in}})\to\mathcal{L}(\mathcal{H}_{\mathrm{mid}})\}_{b\in\{0,1\}}$ whose associated POVM is $\{P_b\}_{b\in\{0,1\}}$ and each $\Gamma_b$ has Kraus rank~$1$.
By \Cref{lem:output_dim_indecomp}, we may choose $\dim \mathcal{H}_{\mathrm{mid}} = 2^{m-1} \ge \max_b \operatorname{rank} P_b$.
There exists a single Kraus operator $K_b:\mathcal{H}_{\mathrm{in}}\to\mathcal{H}_{\mathrm{mid}}$ with $\Gamma_b(\rho):=K_b \rho K_b^\dagger$ such that
\begin{align}
  K_b^\dagger K_b = P_b
  \qquad \forall \, b\in\{0,1\}.
\end{align}
We can define an isometry $\widetilde{V} : \mathcal{H}_{\mathrm{in}} \to \mathbb{C}^2 \otimes \mathcal{H}_{\mathrm{mid}}$ such that $ \widetilde{V}\ket{\psi} = \sum_{b} \ket{b} \otimes K_b \ket{\psi}$ for all $\ket{\psi} \in \mathcal{H}_{\mathrm{in}}$.
By extending $\widetilde{V}$ to an $m$-qubit unitary
$V: (\mathbb{C}^2)^{\otimes (m - n_{\mathrm{in}})} \otimes \mathcal{H}_{\mathrm{in}} \to \mathbb{C}^2 \otimes \mathcal{H}_{\mathrm{mid}}$, we have
\begin{align}
  \Gamma_b(\rho)
  =
  \Tr_{\mathbb{C}^2}
  \qty[
    (\ketbra{b} \otimes \mathbb{I}_\mathrm{mid})
    V
    (\rho \otimes \ketbra{0}^{\otimes (m-n_\mathrm{in})})
    V^\dagger
  ]
  \qquad
  \forall \rho \in \mathcal{L}(\mathcal{H}_{\mathrm{in}}).
\end{align}

\medskip
\noindent\textbf{Intermediate instrument $\Theta_{\mid b}$.}
Fix $b\in\{0,1\}$.
Define $N_{k \mid b} := K_b\, \AssocPOVM{\Lambda}_k\, K_b^\dagger$ for $k\in\Outk_b$, and adjust one element by
$N_{k_0\mid b} \leftarrow N_{k_0\mid b} + (\mathbb{I}_{\mathrm{mid}}-K_b K_b^\dagger)$.
Then $\{ N_{k \mid b} \}_{k \in \Outk_b}$ is a POVM on $\mathcal{H}_{\mathrm{mid}}$ because
\begin{align*}
\sum_{k\in\Outk_b} N_{k \mid b}
&= \sum_{k\in\Outk_b} K_b \AssocPOVM{\Lambda}_k K_b^\dagger + (\mathbb{I}_{\mathrm{mid}}-K_b K_b^\dagger) \\
&= K_b P_b K_b^\dagger + (\mathbb{I}_{\mathrm{mid}}-K_b K_b^\dagger) \\
&= K_b K_b^\dagger + (\mathbb{I}_{\mathrm{mid}}-K_b K_b^\dagger) \\
&= \mathbb{I}_{\mathrm{mid}},
\end{align*}
using $K_b P_b = K_b$ (since $P_b = K_b^\dagger K_b$ is a projector).

If needed, enlarge the outcome sets by adding $N_{k \mid b} =0$ to $\{N_{k \mid b}\}_{k \in \Outk_b}$ for each $b \in \{0,1\}$ so that $|\Outk_0|=|\Outk_1|=2^T$ for some $T \in \mathbb{Z}_{\geq 0}$.
As in the proof of Lemma~\ref{lem:povm_space}, the L\"uders instrument
\begin{align}
\Theta_{k\mid b}(\rho)
= \sqrt{N_{k \mid b}}\;\rho\;\sqrt{N_{k \mid b}},
\end{align}
can be implemented by a repetition of $m$-qubit unitary operations and computational basis measurements followed by the initialization to $\ket{0}$, as illustrated in \Cref{fig:pf_half_cut}.

\medskip
\noindent\textbf{Final channel $\mathcal{W}_{k}$.}
Consider the composed quantum instrument $\{ \Theta_{k \mid b}\circ \Gamma_b \}_{k \in \Outk}$. Here, we can index outcomes only by $k \in \Outk$ since $b$ can be uniquely identified from $k$ as $b$ satisfying $k \in \Outk_b$. Set
\begin{align}
(\Theta_{k\mid b}\circ \Gamma_b)(\rho)
= L_k\,\rho\,L_k^\dagger,
\qquad
L_k:=\sqrt{N_{k \mid b}}\,K_b .
\end{align}
Its associated POVM elements are
\begin{align}
L_k^\dagger L_k
= K_b^\dagger\!\qty(N_{k \mid b})\!K_b
= P_b \AssocPOVM{\Lambda}_k P_b
= \AssocPOVM{\Lambda}_k ,
\end{align}
where the last equality uses $\AssocPOVM{\Lambda}_k P_{b}=\AssocPOVM{\Lambda}_k$ for $k \in \Outk_b$.
Thus $L_k^\dagger L_k=A_k^\dagger A_k$.
By Lemma~\ref{lem:isometry_equivalence_of_kraus_op}, there exists a unitary operator
$\widetilde W_{k}:\operatorname{Ran}(L_k)\to \operatorname{Ran}(A_k)$ with
$A_k = \widetilde W_{k}\,L_k$.

We extend $\widetilde W_k$ to a unitary operator $W_k$ depending on the size of $\mathcal H_{\mathrm{out}}$:

\begin{enumerate}[label=(\roman*)]
\item If $n_{\mathrm{out}}\le m-1$, then $\dim\mathcal H_{\mathrm{out}}\le 2^{m-1}=\dim\mathcal H_{\mathrm{mid}}$.
Extend $\widetilde W_{k}$ to an $(m\!-\!1)$-qubit unitary
$W_{k}:\mathcal H_{\mathrm{mid}}\to (\mathbb C^2)^{\otimes(m-1-n_{\mathrm{out}})} \otimes \mathcal H_{\mathrm{out}}$
so that $ W_{k} \ket{\psi} = \ket{0}^{\otimes(m-1-n_{\mathrm{out}})} \otimes \widetilde W_{k} \ket{\psi}$ for all $\ket{\psi} \in \operatorname{Ran}(L_k)$.
Then
\begin{align}
\Lambda_k(\rho) = A_k \rho A_k^\dagger
= \Tr_{(\mathbb C^2)^{\otimes(m-1-n_{\mathrm{out}})}}\!\big[W_{k}\, L_k \rho L_k^\dagger\, W_{k}^\dagger\big].
\end{align}

\item If $n_{\mathrm{out}}=m-1$, then $\dim\mathcal H_{\mathrm{out}}=\dim\mathcal H_{\mathrm{mid}}$.
Extend $\widetilde W_k$ to an $(m\!-\!1)$-qubit unitary
$W_k:\mathcal H_{\mathrm{mid}}\to \mathcal H_{\mathrm{out}}$ so that $A_k=W_k L_k$, hence
\begin{align}
\Lambda_k(\rho) = A_k \rho A_k^\dagger = W_k L_k \rho L_k^\dagger W_k^\dagger.
\end{align}

\item If $n_{\mathrm{out}}=m$, then $\dim\mathcal H_{\mathrm{out}}=2\,\dim\mathcal H_{\mathrm{mid}}$.
Extend $\widetilde W_{k}$ to an $m$-qubit unitary
$W_{k}: \mathbb C^2 \otimes \mathcal H_{\mathrm{mid}}\to \mathcal H_{\mathrm{out}}$ so that
$A_k = W_{k} (\ket{0}\otimes L_k)$. Then
\begin{align}
\Lambda_k(\rho) = A_k \rho A_k^\dagger =
W_{k}\; \qty( L_k \rho L_k^\dagger \otimes \ketbra{0} )\; W_{k}^\dagger.
\end{align}
\end{enumerate}

Define the quantum channel $\mathcal{W}_{k} : \mathcal{L}(\mathcal{H}_{\mathrm{mid}}) \to \mathcal{L}(\mathcal{H}_{\mathrm{out}})$ by
\begin{align}
\mathcal{W}_k(\rho) := \Tr_{(\mathbb C^2)^{\otimes(m-n_{\mathrm{out}})}}\!\big[W_{k}\, (\ketbra{0} \otimes \rho ) \,W_{k}^\dagger\big],
\end{align}
where in cases (i)–(ii) we regard the $(m\!-\!1)$-qubit unitary $W_k$ as an $m$-qubit unitary by tensoring an identity on one extra qubit (we keep the same symbol for simplicity).
All three cases are then summarized by
$\Lambda_k(\rho) = \mathcal{W}_{k}(L_k \rho L_k^\dagger)$.

Combining the pieces, for all $\rho\in\mathcal L(\mathcal H_{\mathrm{in}})$ we have
\begin{align}
\Lambda_k(\rho) \;=\; \big(\mathcal{W}_k \circ \Theta_{k\mid b} \circ \Gamma_b\big)(\rho).
\end{align}
As indicated in \Cref{fig:pf_half_cut}, this composition uses only the elementary unitary operations and the elementary computational basis measurements specified in \Cref{def:space_instr_wo_delayed}. Hence $\Lambda$ is $m$-qubit implementable (without delayed inputs).
\end{proof}

\subsection{Proof of \Cref{thm:suff_qubit_reduction}} \label{app:pf_thm_suff_qubit_reduction}

\begin{proof}

Let $\Lambda := \{\Lambda_k\}_{k \in \Outk}$ with
$\Lambda_k(\rho) := \Tr_\mathrm{R}\!\bigl[(\ketbra{k}_\mathrm{R} \otimes \mathbb{I}_\mathrm{out} )\, U \rho \, U^\dagger\bigr]$,
where $U : \mathcal{H}_\mathrm{in} \to \mathcal{H}_\mathrm{R} \otimes \mathcal{H}_\mathrm{out}$ is unitary and
$\mathcal{H}_\mathrm{in} \cong (\mathbb{C}^2)^{\otimes n_\mathrm{in}},\;
  \mathcal{H}_\mathrm{out} \cong (\mathbb{C}^2)^{\otimes n_\mathrm{out}}$.
The associated POVM of $\Lambda$, given by $\AssocPOVM{\Lambda} = \{U^\dagger (\ketbra{k}_\mathrm{R} \otimes \mathbb{I}_\mathrm{out} ) U\}_{k \in \Outk}$, is a projective measurement.
Suppose projective measurements $E^{(t)} := \{ E^{(t)}_{k_t} \}_{k_t \in \Outk_t}$ for $t=1,\ldots,T$ and an ordered subset of input qubits
$\{\mathrm{A}_1,\ldots,\mathrm{A}_T\}$ satisfies the hypotheses required in the theorem statement. For convenience, set $m := n_\mathrm{in} - T$.

We prove that $\Lambda$ is $m$-qubit implementable by induction on $T \in \mathbb{N}$.

\medskip
\noindent\textbf{Base case $T=0$.}
Here $m = n_\mathrm{in}$.
By definition, the instrument $\Lambda$ is implemented by the $n_\mathrm{in}$-qubit unitary $U$, followed by computational basis measurements on $\mathcal{H}_\mathrm{R}$.
In the notation of \Cref{def:m_space_instr_wdi}, these are the elementary unitary operation and the elementary computational basis measurement (with the ancilla prepared in $\ket{0}$ traced out at the end).
Hence $\Lambda$ is $m$-qubit implementable.

\medskip
\noindent\textbf{Induction step.}
Assume the theorem holds for $T-1$ as the induction hypothesis.

First, define a quantum instrument $\Gamma :=\{\Gamma_{k_T}:\mathcal L(\mathcal H_{\mathrm{in}})\to\mathcal L(\mathcal{H}_{\mathrm{Y}})\}_{k_T \in \Outk_T}$ such that the associated POVM coincides with $E^{(T)}$ and each $\Gamma_{k_T}$ has Kraus rank~1. By Lemma~\ref{lem:output_dim_indecomp}, we may take $\dim\mathcal H_{\mathrm{Y}}= \max_{k_T} \rank E^{(T)} = 2^{m}$.
By \Cref{rem:ons_povm}, $\Gamma$ satisfies the same outcome no-signaling condition $\mathrm{A}_T \nrightarrow \mathrm{cl}$ as $E^{(T)}$.

\begin{figure}[H]
  \centering
  \includegraphics[width=0.7\linewidth]{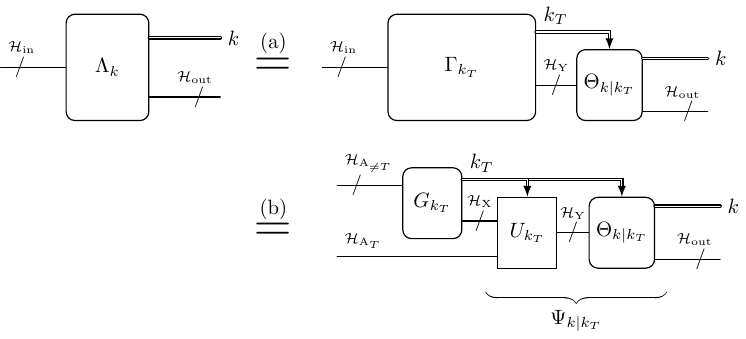}
  \caption{Decomposition of $\{\Lambda_k\}_{k}$ used in the induction step.
  The first equality follows from the composability of $\{ \Lambda_k\}_{k}$ from $\{ \Gamma_{k_T}\}_{k_T}$, as explained in part (a). Here, each $\Theta_{k\mid k_T}$ has Kraus rank 1.
  The second equality follows from the outcome no-signaling condition $\mathrm A_T\nrightarrow\mathrm{cl}$ for $\{ \Gamma_{k_T}\}_{k_T}$, as explained in part (b). Here, each $G_{k_T}$ has Kraus rank 1, $W_{k_T}$ is unitary, and $\dim\mathcal{H}_{\mathrm{X}}=2^{\,n_{\mathrm{in}}-T-1}=2^{\,m-1}$.}
  \label{fig:induction_step_decomp_pf}
\end{figure}

\medskip
\noindent\textbf{(a) Composability for $\Lambda$.}
Because $\AssocPOVM{\Gamma}$ is composable from $\AssocPOVM{\Lambda}$, \Cref{lem:pp_cond_indecomp_instr} implies that $\Lambda$ is composable from $\Gamma$.
Hence, there exists a quantum instrument $\Theta_{\mid k_T}=\{\Theta_{k\mid k_T}: \mathcal{L}(\mathcal{H}_{\mathrm{Y}})\to\mathcal{L}(\mathcal{H}_{\mathrm{out}})\}_{k \in \Outk}$ with each $\Theta_{k\mid k_T}$ having Kraus rank~1, such that
\begin{align}
\Lambda_k=\sum_{k_T}\Theta_{k\mid k_T}\circ \Gamma_{k_T},
\end{align}
as illustrated in \Cref{fig:induction_step_decomp_pf}.

\medskip
\noindent\textbf{(b) Outcome no-signaling condition for $\Gamma$.}
Applying Lemma~\ref{lem:ons_decomp_mindim} to the condition $\mathrm A_T\nrightarrow\mathrm{cl}$ for $\Gamma$, there exists a quantum instrument
$G=\{G_{k_T}:\mathcal L(\mathcal H_{A_{\neq T}})\to\mathcal L(\mathcal{H}_{\mathrm{X}})\}_{k_T}$ with each $G_{k_T}$ of Kraus rank~1,
and a unitary operator $W_{k_T}:\mathcal H_{\mathrm{X}}\otimes\mathcal H_{A_T}\to\mathcal H_{\mathrm{Y}}$ for each $k_T$ such that
\begin{align}
\Gamma_{k_T}(\rho)=W_{k_T}\bigl(G_{k_T}\otimes\mathrm{id}_{A_T}\bigr)(\rho)\,{W_{k_T}}^\dagger.
\label{eq:pf_ons_decomp}
\end{align}
Here, $A_{\neq T}$ represents the set of the input qubits other than $\mathrm A_T$.
Matching input/output dimensions of $W_{k_T}$ gives $\dim \mathcal{H}_{\mathrm{X}} = 2^{\,m-1}$.
Define a quantum instrument $\Psi_{\mid k_T} := \{\Psi_{k\mid k_T}\}_{k \in \Outk}$ by $\Psi_{k\mid k_T}(\sigma):=\Theta_{k\mid k_T}\!\bigl(W_{k_T}\;\sigma\;{W_{k_T}}^\dagger\bigr)$.
Then
\begin{align}
\Lambda_k=\sum_{k_T}\Psi_{k\mid k_T}\circ\qty(G_{k_T}\otimes\mathrm{id}_{A_T}),
\label{eq:lambda_decomp_induction_pf}
\end{align}
as illustrated in \Cref{fig:induction_step_decomp_pf}.

\medskip
\noindent\textbf{(c) $G$ is $m$-qubit implementable (with delayed inputs).}
Calculating the associated POVMs of both sides of \Cref{eq:pf_ons_decomp}, we have
\begin{align}
E^{(T)}_{k_T} =\AssocPOVM{G}_{k_T}\otimes\mathbb I_{A_T}.
\label{eq:pf_assocGt_factor}
\end{align}
Then, $\AssocPOVM{G}=\{ \AssocPOVM{G}_{k_T} \}_{k_T}$ is a projective measurement with $\rank \AssocPOVM{G}_{k_T}=2^{m-1}$ for all $k_T$. By \Cref{rem:equiv_def_target_instr}, there exists a unitary operator $U':\mathcal H_{A_{\neq T}}\to \mathcal{H}_\mathrm{R'} \otimes \mathcal{H}_{\mathrm{X}}$ with $\mathcal{H}_\mathrm{R'} \cong (\mathbb{C}^2)^{\otimes T}$ such that
\begin{align}
G_{k_T} (\rho ) = \Tr_\mathrm{R'} \qty[(\mathbb{I}_{\mathrm{X}} \otimes \ketbra{k_T}_\mathrm{R'}) U' \rho \, {U'}^\dagger]
\qquad
\forall \rho \in \mathcal{L}(\mathcal{H}_{A_{\neq T}}).
\label{eq:pf_G_unitary_rep}
\end{align}

By the definition of the composability of POVMs and \Cref{eq:pf_assocGt_factor}, every $E^{(t)}$ for $t<T$, being composable from $E^{(T)}$, also factorizes as
\begin{align}
E^{(t)}_{k_t} = F^{(t)}_{k_t} \otimes \mathbb I_{A_T},
\qquad \forall \, k_t \in \Outk_t,
\end{align}
where $\{F^{(t)}_{k_t}\}_{k_t \in \Outk_t}$ is a POVM satisfying $\operatorname{rank}F^{(t)}_{k_t}=2^{(n_{\mathrm{in}}-1)-t}$.
Moreover, $E^{(1)} \lcircarrow E^{(2)} \lcircarrow \cdots \lcircarrow E^{(T-1)} \lcircarrow E^{(T)}$ implies
\begin{align}
F^{(1)} \lcircarrow F^{(2)} \lcircarrow \cdots \lcircarrow F^{(T-1)} \lcircarrow \AssocPOVM{G},
\end{align}
and since $\mathrm A_t\nrightarrow\mathrm{cl}$ for $E^{(t)}$ with $\mathrm A_T\neq\mathrm A_t$, we also have $\mathrm A_t\nrightarrow\mathrm{cl}$ for $F^{(t)}$.
Applying the induction hypotheses (for $T-1$) to $G$, which has $(n_\mathrm{in}-1)$-qubit input system, shows that $G$ is $(n_{\mathrm{in}}-1)-(T-1)=m$-qubit implementable (with delayed inputs).

\medskip
\noindent\textbf{(d) $\Psi_{\mid k_T}$ is $m$-qubit implementable (without delayed inputs).}
From the composability $\AssocPOVM{\Gamma} \lcircarrow \AssocPOVM{\Lambda}$, there exists a column-stochastic matrix $\nu := ( \nu_{k_T, k} )_{k_T, k}$ such that
\begin{align}
\AssocPOVM{\Gamma}_{k_T} =\sum_{k}\nu_{k_T,k} \AssocPOVM{\Lambda}_k.
\end{align}
Because both $\AssocPOVM{\Gamma}$ and $\AssocPOVM{\Lambda}$ are projective measurements, the entries of $\nu$ must be $0/1$: indeed, ${(\AssocPOVM{\Gamma}_{k_T})}^2=\AssocPOVM{\Gamma}_{k_T}$ implies
$(\sum_k \nu_{k_T,k}\AssocPOVM{\Lambda}_k )^2=\sum_k \nu_{k_T,k} \AssocPOVM{\Lambda}_k$, and orthogonality of ${\AssocPOVM{\Lambda}_k}$ yields $(\nu_{k_T,k})^2=\nu_{k_T,k}$ for all $k_T,k$.
Define $B_{k_T}:=\{k:\nu_{k_T,k}=1\}$; then
\begin{align}
  \AssocPOVM{\Gamma}_{k_T} \;=\; \sum_{k\in B_{k_T}} \AssocPOVM{\Lambda}_k .
\end{align}
Since $\nu$ is column-stochastic, each $k$ belongs to exactly one set $B_{k_T}$.

Taking associated POVMs in \Cref{eq:lambda_decomp_induction_pf} and using \Cref{eq:pf_G_unitary_rep} gives
\begin{align}
\AssocPOVM{\Lambda}_k
=  U'^\dagger \sum_{k_T} \qty(\AssocPOVM{\Psi_{\mid k_T}}_k \otimes \ketbra{k_T}_\mathrm{R'}) U'
\qquad
\forall \, k \in \Outk.
\label{eq:pf_decomp_assocPOVM_Lambda}
\end{align}

Fix $\widetilde{k_T}$. Summing \eqref{eq:pf_decomp_assocPOVM_Lambda} over $k\in B_{\widetilde{k_T}}$ must reproduce $\AssocPOVM{\Gamma}_{\widetilde{k_T}}$, which is given by $U'^\dagger\bigl(\mathbb{I}_{\mathrm{X},\mathrm{A}_T}\otimes \ketbra{\widetilde{k_T}}_{\mathrm{R}'}\bigr)U'$ by \Cref{eq:pf_G_unitary_rep}. Hence, for each $k_T$,
\[
  \sum_{k\in B_{k_T}} \AssocPOVM{\Psi_{\mid k_T}}_k \;=\; \mathbb{I}_{\mathrm{X},\mathrm{A}_T}
  \quad\text{and}\quad
  \AssocPOVM{\Psi_{\mid k_T}}_k \;=\; 0 \ \ \text{for } k\notin B_{k_T}.
\]
Because each $k$ lies in exactly one $B_{k_T}$, the sum over $k_T$ in \eqref{eq:pf_decomp_assocPOVM_Lambda} has exactly one nonzero term for any fixed $k$, yielding
\begin{align}
  \AssocPOVM{\Lambda}_k
  \;=\; {U'}^\dagger \bigl(\AssocPOVM{\Psi_{\mid k_T}}_k \otimes \ketbra{k_T}_{\mathrm{R}'}\bigr) U',
  \qquad \forall\, k\in\Outk ,
  \label{eq:pf_decomp_assocPOVM_Lambda2}
\end{align}
where $k_T$ is the unique index with $k\in B_{k_T}$. Since $\AssocPOVM{\Lambda}$ is a projective measurement, so is $\AssocPOVM{\Psi_{\mid k_T}}$, and thus
\begin{align}
  \rank \AssocPOVM{\Psi_{\mid k_T}}_k \;=\;
  \begin{cases}
    0 & (k\notin B_{k_T}),\\[2pt]
    2^{n_{\mathrm{out}}} & (k\in B_{k_T}).
  \end{cases}
  \label{eq:pf_rank_Sk_cond}
\end{align}

For each fixed $k_T$, choose any disjoint sets $B_0, B_1$ that satisfy $B_{k_T} = B_0 \cup B_1$ and $|B_0|=|B_1|=|B_{k_T}|/2 = 2^{m-n_{\mathrm{out}}-1}$, and define $\Outk_0 := B_0$ and $\Outk_1 := \Outk \setminus B_0$. Then, we have  $\Outk=\Outk_0\cup \Outk_1$ and
\begin{align}
\sum_{k\in \Outk_j}\AssocPOVM{\Psi_{\mid k_T}}_k
= P_j
\qquad \forall \, j \in \{0,1\},
\end{align}
where $\{P_j\}_{j \in \{0,1\}}$ is a projective measurement on $\mathcal{H}_{\mathrm{X}} \otimes \mathcal{H}_{A_T}$ with $\rank P_j = 2^{m-1}$.
By \Cref{thm:range_m_space_inst_wodi}, $\{ \Psi_{k \mid k_T} \}_{k}$ is $m$-qubit implementable (without delayed inputs).

\medskip
Combining the above,
\begin{align}
  \Lambda_k=\sum_{k_T}\Psi_{k\mid k_T}\circ\bigl(G_{k_T}\otimes\mathrm{id}_{A_T}\bigr),
\end{align}
where $G$ is $m$-qubit implementable (with delayed inputs) and $\Psi_{\mid k_T}$ is $m$-qubit implementable (without delayed inputs).
Substituting the representation \Cref{eq:stairs_wdi_1} for $G$ shows that $\Lambda$ again admits the form \Cref{eq:stairs_wdi_1}.
Hence, by \Cref{rem:equiv_def_target_instr}, $\Lambda$ is $m$-qubit implementable (with delayed inputs).

\end{proof}

\subsection{Proof of \Cref{thm:necc_qubit_reduction}} \label{app:pf_thm_necc_qubit_reduction}

\begin{proof}

We first present a property of classical processing that holds independently of the theorem:
\noindent\textbf{Lemma (pushing classical processing to the final round).}

We show below that, without loss of generality, any $m$-qubit implementable instrument (with delayed inputs) can be written so that the elementary classical-processing operation appears only once, at the final round of the sequence of elementary operations.

Let $\{\rho_k \in \mathcal{L}(\mathcal{H}_\mathrm{m} \otimes \mathcal{H}_\mathrm{in})\}_{k \in \Outk}$ be a set of unnormalized states.
Given a function $f:\Outk \to \Outk'$, the elementary classical-processing operation maps $\{\rho_k\}_{k\in\Outk}\;\mapsto\; \{ \sum_{k \in f^{-1}(k')} \rho_k \}_{k' \in \Outk'}$.

\begin{enumerate}[label=(\alph*)]
  \item If an elementary unitary operation is applied thereafter, we obtain
  \begin{align}
    \qty{ \sum_{k \in f^{-1}(k')} \rho_k }_{k' \in \Outk'}
    \mapsto
    \qty{ \sum_{k \in f^{-1}(k')} U_{k'} \rho_k U_{k'}^\dagger }_{k' \in \Outk'},
  \end{align}
  where $U_{k'}$ is a unitary on $\mathcal{H}_\mathrm{m}\otimes\mathcal{H}_\mathrm{in}$ conditioned on $k'$.
  The same result is obtained by first applying the unitary conditioned on $k$, with $V_k:=U_{f(k)}$, and then applying the same classical processing:
  \begin{align}
    \{\rho_k\}_{k\in\Outk}
    \mapsto
    \{V_k\rho_k V_k^\dagger\}_{k\in\Outk}
    \mapsto
    \qty{ \sum_{k \in f^{-1}(k')} V_k \rho_k V_k^\dagger }_{k' \in \Outk'}
    =
    \qty{ \sum_{k \in f^{-1}(k')} U_{k'} \rho_k U_{k'}^\dagger }_{k' \in \Outk'}.
  \end{align}

  \item If an elementary computational basis measurement or an elementary input-loading operation is applied thereafter, we obtain
  \begin{align}
    \qty{ \sum_{k \in f^{-1}(k')} \rho_k }_{k' \in \Outk'}
    \mapsto
    \qty{ \sum_{k \in f^{-1}(k')} N_x \rho_k N_x^\dagger }_{k'x \in \Outk' \times \{0,1\}^{|S|}},
  \end{align}
  where, for $S := S(k') \subseteq [m]$, we take $N_x := (\ketbra{x}_S \otimes \mathbb{I})$ for the elementary computational basis measurement and $N_x := (\bra{x}_S \otimes \mathbb{I})$ for the elementary input-loading operation, with $x \in \{0,1\}^{|S|}$ in both cases.
  The same result is obtained by first applying the measurement/loading on $S := S(f(k))$, and then the classical processing with $\widetilde{f}:\Outk\times\{0,1\}^{|S|} \to \Outk'\times\{0,1\}^{|S|}$ defined by $\widetilde{f}(kx)=f(k)x$:
  \begin{align}
    \{\rho_k\}_{k\in\Outk}
    \mapsto
    \{N_x \rho_k N_x^\dagger\}_{kx}
    \mapsto
    \qty{ \sum_{kx \in \widetilde{f}^{-1}(k'x')} N_x \rho_k N_x^\dagger }_{k'x'}
    =
    \qty{ \sum_{k \in f^{-1}(k')} N_x \rho_k N_x^\dagger }_{k'x}.
  \end{align}

  \item If another elementary classical-processing operation with $g:\Outk'\to\Outk''$ follows, then
  \begin{align}
    \qty{ \sum_{k \in f^{-1}(k')} \rho_k }_{k' \in \Outk'}
    \mapsto
    \qty{ \sum_{k' \in g^{-1}(k'')} \sum_{k \in f^{-1}(k')} \rho_k }_{k'' \in \Outk''}
    =
    \qty{ \sum_{k \in (g\circ f)^{-1}(k'')} \rho_k }_{k'' \in \Outk''}.
  \end{align}
  Therefore, the two classical-processing steps can be merged into a single one with $g\circ f$.
\end{enumerate}
By iterating (a)–(c), all elementary classical-processing operations can be merged and postponed to the final round.

\bigskip
\noindent\textbf{Proof of the theorem.}
Let $\Lambda := \{\Lambda_k\}_{k \in \Outk}$ with
$\Lambda_k(\rho) := \Tr_\mathrm{R}\!\bigl[(\ketbra{k}_\mathrm{R} \otimes \mathbb{I}_\mathrm{out} )\, U \rho \, U^\dagger\bigr]$,
where $U : \mathcal{H}_\mathrm{in} \to \mathcal{H}_\mathrm{R} \otimes \mathcal{H}_\mathrm{out}$ is unitary and
$\mathcal{H}_\mathrm{in} \cong (\mathbb{C}^2)^{\otimes n_\mathrm{in}}$, $\mathcal{H}_\mathrm{out} \cong (\mathbb{C}^2)^{\otimes n_\mathrm{out}}$.

Assume $\Lambda$ is $(n_\mathrm{in}-T)$-qubit implementable (with delayed inputs). By \Cref{rem:stairs_wdi} and the above “push-to-final-round” argument, $\Lambda$ admits the form
\begin{multline}
  \Lambda_k (\rho_\mathrm{in}) =
  \sum_{k_{n_\mathrm{in}+1} \in f^{-1}(k)}
  \Tr_{(\mathbb{C}^2)^{\otimes (n_\mathrm{in} -T-n_\mathrm{out})}}
  \Bigl[
    \Bigl(
      \sum_{k_0, \ldots, k_{n_\mathrm{in}}}
      \widetilde\Gamma^{(n_\mathrm{in}+1)}_{k_{n_\mathrm{in}+1}\mid k_{n_\mathrm{in}}}
      \circ \cdots \circ
      \widetilde\Gamma^{(1)}_{k_1\mid k_0}
    \Bigr)
    \bigl(\ketbra{0}^{\otimes (n_\mathrm{in} -T)} \otimes \rho_{\mathrm{in}} \bigr)
  \Bigr]\\
  \forall\, \rho_{\mathrm{in}} \in \mathcal{L}(\mathcal H_{\mathrm{in}}),
  \label{eq:lambda_decomp_necc_pf}
\end{multline}
with
\begin{align}
  \widetilde{\Gamma}^{(t)}_{k_t \mid k_{t-1}}
  \;:=\;
  \Gamma^{(t)}_{k_t \mid k_{t-1}} \otimes \operatorname{id}_{\mathrm{A}_{t}, \ldots, \mathrm{A}_{n_\mathrm{in}}},
\end{align}
where each instrument $\Gamma^{(t)}_{\mid k_{t-1}} := \{\Gamma^{(t)}_{k_t \mid k_{t-1}}\}_{k_t \in \Outk_t}$ is $(n_\mathrm{in} - T)$-qubit implementable (without delayed inputs). See \Cref{fig:pf_stairs_wdi}. Unlike \Cref{rem:stairs_wdi}, there is now a single elementary classical-processing operation applied only at the end, and no $\Gamma^{(t)}_{\mid k_{t-1}}$ contains any classical-processing step.

\begin{figure}[H]
\centering
\includegraphics[width=0.99\linewidth]{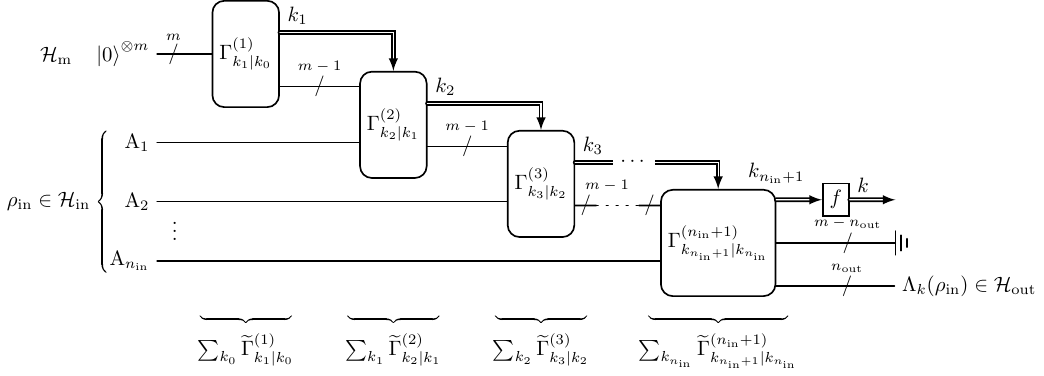}
\caption{An implementation of $\{\Lambda_k\}_k$ as an $(n_\mathrm{in}-T)$-qubit implementable instrument (with delayed inputs). Each instrument $\{\Gamma^{(t)}_{k_t \mid k_{t-1}}\}_{k_t}$ is $(n_\mathrm{in}-T)$-qubit implementable (without delayed inputs) and contains no elementary classical-processing operation. The final classical-processing $f$ is applied only once, at the end. In the figure, we may write $m := n_\mathrm{in} - T$ for brevity.}
\label{fig:pf_stairs_wdi}
\end{figure}

For each $t\in [n_\mathrm{in}]$, define the accumulated instrument $\Xi^{(t)} := \{\Xi^{(t)}_{k_t}\}_{k_t \in \Outk_t}$ by
\begin{align}
  \Xi^{(t)}_{k_t}
  \;:=\;
  \sum_{k_0, \ldots, k_{t-1}}
  \widetilde{\Gamma}^{(t)}_{k_t \mid k_{t-1}} \circ \cdots \circ \widetilde{\Gamma}^{(1)}_{k_1 \mid k_0}
  \qquad \forall \, k_t \in \Outk_t.
  \label{eq:def_Xi_pf}
\end{align}
Since no $\Gamma^{(t)}_{\mid k_{t-1}}$ contains classical processing, each $\Xi^{(t)}_{k_t}$ has Kraus rank~$1$. Indeed, by the definitions of the elementary unitary, computational basis measurement, and input-loading operations (see \Cref{eq:def_elementary_u_wdi,eq:def_elementary_m_wdi,eq:def_elementary_l_wdi}), any composition of them preserves the Kraus-rank-$1$ property.

Moreover, the output system of $\Xi^{(t)}$ has dimension $2^{\,2 n_\mathrm{in} - T - t}$; hence, by the Kraus-rank-$1$ property and \Cref{lem:output_dim_indecomp}, for each $t \in [n_\mathrm{in}]$,
\begin{align}
  \rank \AssocPOVM{\Xi^{(t)}}_{k_t} \;\le\; 2^{\,2 n_\mathrm{in} - T - t}
  \qquad \forall\, k_t \in \Outk_t.
\end{align}

From \Cref{eq:lambda_decomp_necc_pf} and \Cref{eq:def_Xi_pf} we obtain the composability conditions
\begin{equation}\label{eq:stair-pp-chain}
  (\Xi^{(1)} \rcircarrow \cdots )\; \Xi^{(n_\mathrm{in}-T+1)} \rcircarrow \cdots \rcircarrow \Xi^{(n_\mathrm{in})} \rcircarrow \Lambda,
\end{equation}
which, by \Cref{lem:pp_cond_indecomp_instr}, yields the corresponding composability conditions for the associated POVMs:
\begin{align}
  \bigl(\AssocPOVM{\Xi^{(1)}} \lcircarrow \cdots \bigr)
  \AssocPOVM{\Xi^{(n_\mathrm{in}-T+1)}}
  \lcircarrow \cdots \lcircarrow
  \AssocPOVM{\Xi^{(n_\mathrm{in})}}
  \lcircarrow \AssocPOVM{\Lambda}.
\end{align}

For each $t\in [n_\mathrm{in}]$, the instruments $\widetilde{\Gamma}^{(1)}_{\mid k_0},\ldots,\widetilde{\Gamma}^{(t)}_{\mid k_{t-1}}$ act trivially on $\mathrm{A}_t$, and hence so does $\Xi^{(t)}$. Thus, by \Cref{thm:ons_decomp}, $\Xi^{(t)}$ satisfies the outcome no-signaling condition $\mathrm{A}_t \nrightarrow \mathrm{cl}$; by \Cref{rem:ons_povm}, so does $\AssocPOVM{\Xi^{(t)}}$.

The associated POVMs $\AssocPOVM{\Xi^{(t)}}$ for $t=n_\mathrm{in}-T+1,\ldots,n_\mathrm{in}$ satisfy all the required conditions in the theorem statement. Indeed, set $E^{(s)} := \AssocPOVM{\Xi^{(n_\mathrm{in}-T+s)}}$ and $\mathrm{A}_s := \mathrm{A}_{n_\mathrm{in}-T+s}$ for $s \in [T]$. Then:
\begin{itemize}
  \item The composability conditions hold: $E^{(1)} \lcircarrow E^{(2)} \lcircarrow \cdots \lcircarrow E^{(T)} \lcircarrow \AssocPOVM{\Lambda}$.
  \item Each $E^{(s)}$ satisfies the outcome no-signaling condition $\mathrm{A}_s \nrightarrow \mathrm{cl}$.
  \item Each $E^{(s)} := \{ E^{(s)}_{k_s} \}_{k_s \in \Outk_s}$ satisfies $\rank E^{(s)}_{k_s} \le 2^{\,n_\mathrm{in} - s}$ for all $k_s \in \Outk_s$.
\end{itemize}
This completes the proof.
\end{proof}

\subsection{Proof of \Cref{thm:no_gap_ent_distill}} \label{app:pf_thm_no_gap_ent_distill}

\begin{proof}
Let $\mathcal{C}$ be an $[[n,k]]$ stabilizer code with stabilizer generators $\{g_1, \ldots, g_{n-k}\}$ and let $U_{\mathrm{enc}}: \mathcal{H}_\mathrm{in} \to \mathcal{H}_\mathrm{R} \otimes \mathcal{H}_\mathrm{out}$ be the encoding unitary of $\mathcal{C}$ satisfying
\begin{align}
U_{\mathrm{enc}} \big(  Z_i \otimes \mathbb{I}_\mathrm{out} \big) U_{\mathrm{enc}}^\dagger = g_i,
\label{eq:encoding_unitary_pf}
\end{align}
for all $i \in \{1, \cdots, n-k\}$, where $\mathcal{H}_\mathrm{in} \cong (\mathbb{C}^2)^{\otimes n}$, $\mathcal{H}_\mathrm{out} \cong (\mathbb{C}^2)^{\otimes k}$, and $\mathcal{H}_\mathrm{R} \cong (\mathbb{C}^2)^{\otimes (n-k)}$.
Let $\Lambda^{\mathrm{dist}} := \{\Lambda^{\mathrm{dist}}_s : \mathcal{L}(\mathcal{H}_\mathrm{in}) \to \mathcal{L}(\mathcal{H}_\mathrm{out})\}_{s \in \mathbb{F}_2^{n-k}}$ be the quantum instrument defined by
\begin{align}
  \Lambda_s^{\mathrm{dist}}(\rho) := \Tr_\mathrm{R} \qty[( \ketbra{s}_\mathrm{R} \otimes \mathbb{I}_\mathrm{out} )\, U_{\mathrm{enc}}^\dagger \rho \, U_{\mathrm{enc}}]
  \qquad \forall \rho \in \mathcal{L}(\mathcal{H}_\mathrm{in}), \quad s \in \mathbb{F}_2^{n-k}.
\end{align}
The associated POVM of $\Lambda^{\mathrm{dist}}$ is a projective measurement where each element has rank $2^k$.
Below, write $\Lambda := \Lambda^{\mathrm{dist}}$ and $U:=U_{\mathrm{enc}}^\dagger$ for brevity.

\bigskip
\noindent\textbf{Expansion of associated POVM $\AssocPOVM{\Lambda}$ in terms of the stabilizer generators.}

Using \Cref{eq:encoding_unitary_pf} and the identity
\begin{align}
  \ketbra{s}_{\mathrm R}=\frac{1}{2^{n-k}}\sum_{r\in\mathbb F_2^{n-k}}(-1)^{s\cdot r}\, Z^r,\qquad
  Z^r:=\bigotimes_{i=1}^{n-k}Z_i^{\,r_{(i)}},
\end{align}
we obtain
\begin{equation}\label{eq:Ps-Fourier_pf}
  \AssocPOVM{\Lambda}_s \;=\; \frac{1}{2^{n-k}}\sum_{r\in\mathbb F_2^{n-k}}(-1)^{s\cdot r}\, g^r,
  \qquad g^r:=\prod_{i=1}^{n-k} g_i^{\,r_{(i)}},
\end{equation}
where $r_{(i)}$ is the $i$-th entry of $r\in\mathbb F_2^{n-k}$ and $s\cdot r:=\sum_{i=1}^{n-k}s_{(i)}r_{(i)}$ is the standard inner product over $\mathbb F_2$.

\bigskip
\noindent\textbf{Hypotheses from the necessary conditions.}

Fix $T \in \{1, 2, \cdots, n\}$. We now assume the necessary condition stated in \Cref{thm:necc_qubit_reduction} holds for $\Lambda$, namely, there exist
POVMs $E^{(t)} := \{E^{(t)}_{s_t}\}_{s_t \in \Outs_t}$ for $t = 1, 2, \cdots, T$ such that
\begin{itemize}
  \item The composability conditions: $E^{(1)} \lcircarrow E^{(2)} \lcircarrow \cdots \lcircarrow E^{(T)} \lcircarrow \AssocPOVM{\Lambda} $.
  \item Each $E^{(t)}$ satisfies the outcome no-signaling condition $\mathrm{A}_t \nrightarrow \mathrm{cl}$, where $\{\mathrm{A}_1, \mathrm{A}_2, \ldots, \mathrm{A}_T\}$ is an ordered subset of the input qubits.
  \item Each $E^{(t)}$ satisfies $\rank E^{(t)}_{s_t} \leq 2^{n - t}$ for all $s_t \in \Outs_t$.
\end{itemize}
From the composability conditions, there exists a column-stochastic matrix $\nu^{(t)} := (\nu^{(t)}_{s_t, s})_{s_t \in \Outs_t, \, s \in \mathbb{F}_2^{n-k}}$ for each $t \in [T]$ such that
\begin{align}
  E^{(t)}_{s_t} = \sum_{s} \nu^{(t)}_{s_t, s}\, \AssocPOVM{\Lambda}_s.
  \label{eq:def_column_pf_ent}
\end{align}
Since $\{\AssocPOVM{\Lambda}_s\}_s$ are pairwise orthogonal projectors with rank $2^k$, the rank bound on $E^{(t)}$ is equivalent to
\begin{align}
  \text{each row of $\nu^{(t)}$ has at most $2^{\,n-k-t}$ nonzero entries.}
  \label{eq:nonzero_count_pf}
\end{align}

\medskip
\noindent\textbf{Equivalent form of the outcome no-signaling conditions.}

For each $t\in[T]$ we also have $\mathrm{A}_{t},\mathrm{A}_{t+1},\cdots,\mathrm{A}_T \nrightarrow \mathrm{cl}$ for $E^{(t)}$, since $E^{(t)}$ is related to each $E^{(\tau)}$ for $\tau \in \{t,\ldots,T\}$ by a column-stochastic matrix and hence $\mathrm{A}_{\tau} \nrightarrow \mathrm{cl}$ for $E^{(\tau)}$ implies the same condition for $E^{(t)}$.

By definition, the condition $\mathrm{A}_t,\ldots,\mathrm{A}_T \nrightarrow \mathrm{cl}$ for $E^{(t)}$ is equivalent to saying that $E^{(t)}$ acts trivially on $\mathrm{A}_t,\ldots,\mathrm{A}_T$.
Expanding in the Pauli basis, for each $\tau \in \{t, \ldots, T\}$ and $\Pi \in \{X, Y, Z\}$ we have
\begin{align}
  \Tr_{{\mathrm{A}_\tau}} \!\big[ \Pi_{\mathrm{A}_\tau} \,E^{(t)}_{s_t} \big] = 0
  \qquad \forall \ s_t \in \Outs_t.
\end{align}
Substituting \Cref{eq:Ps-Fourier_pf,eq:def_column_pf_ent} yields
\begin{align}
  \sum_{r \in \mathbb{F}_2^{n-k}} \qty(\sum_{s} \nu^{(t)}_{s_t, s} (-1)^{s \cdot r} )\,
  \Tr_{{\mathrm{A}_\tau}} \!\big[ \Pi_{\mathrm{A}_\tau} g^r \big] = 0
  \qquad \forall \, s_t \in \Outs_t. \label{eq:ons-expansion}
\end{align}
For each $\tau$, define $x_\tau,z_\tau\in\mathbb F_2^{n-k}$ so that the $i$-th entry of $x_\tau$ (resp. $z_\tau$) is $1$ iff the generator $g_i$ contains Pauli-$X$ (resp. Pauli-$Z$) on qubit $\mathrm{A}_\tau$. Equivalently, these are the column vectors for $\mathrm{A}_\tau$ in the check matrix of $\mathcal{C}$. Define
\begin{align}
  R_{X_\tau} &:= \{ r \in \mathbb{F}_2^{n-k} : r \cdot x_\tau = 1,\ r \cdot z_\tau = 0 \},\\
  R_{Y_\tau} &:= \{ r \in \mathbb{F}_2^{n-k} : r \cdot x_\tau = 1,\ r \cdot z_\tau = 1 \},\\
  R_{Z_\tau} &:= \{ r \in \mathbb{F}_2^{n-k} : r \cdot x_\tau = 0,\ r \cdot z_\tau = 1 \}.
\end{align}
The sum in \Cref{eq:ons-expansion} can be restricted to $r \in R_{\Pi_\tau}$, since $\Tr_{{\mathrm{A}_\tau}}[ \Pi_{\mathrm{A}_\tau} g^r ] = 0$ whenever $r \notin R_{\Pi_\tau}$. Orthogonality of different Pauli operators then gives, for each $\tau \in \{t, \ldots, T\}, \Pi \in \{X, Y, Z\}$, and $s_t \in \Outs_t$,
\begin{align}
  \sum_{s} \nu^{(t)}_{s_t, s} (-1)^{s \cdot r} = 0 \quad \forall\, r \in R_{\Pi_\tau}.
\end{align}
Define the subspace
\begin{equation}\label{eq:L-subspace}
  \Outl_{[t,T]}\;:=\;\bigl\{r\in\mathbb F_2^{n-k}:\ r\!\cdot\!x_\tau=0\ \text{and}\ r\!\cdot\!z_\tau=0\ \ \forall\,\tau\in\{t,\dots,T\}\bigr\}.
\end{equation}
Its orthogonal complement is given by
\begin{align}
  \Outl_{[t,T]}^\perp = \operatorname{span} \{x_\tau, z_\tau : \tau \in \{t, \ldots, T\}\}. \label{eq:L-perp}
\end{align}
Consequently, $\mathrm{A}_{t},\ldots,\mathrm{A}_T \nrightarrow \mathrm{cl}$ for $E^{(t)}$ is equivalent to
\begin{align}
  \sum_{s} \nu^{(t)}_{s_t, s} (-1)^{s \cdot r} = 0 \quad \forall\, r \notin \Outl_{[t,T]},\ \forall \, s_t \in \Outs_t.
\end{align}
By \Cref{lem:fourier_support}, this is equivalent to the coset-constancy condition: for each $s_t \in \Outs_t$,
\begin{align}
  \nu^{(t)}_{s_t,(s + \ell)} = \nu^{(t)}_{s_t,s} \quad \forall\, s \in \mathbb{F}_2^{n-k},\ \forall\, \ell \in \Outl_{[t,T]}^\perp.
  \label{eq:coset_constancy_pf}
\end{align}
In words, each row of $\nu^{(t)}$ has the same entries in every coset of $\Outl_{[t,T]}^\perp$. Here and throughout, for a subspace $V$ of $\mathbb F_2^{n-k}$, a coset of $V$ is a subset of $\mathbb{F}_2^{n-k}$ of the form $s + V := \{ s + v : v \in V \}$ for some $s \in \mathbb{F}_2^{n-k}$.

Let $d_t:=\dim \Outl_{[t,T]}^\perp$. Then $d_t\le n-k-t$ must hold; otherwise \Cref{eq:nonzero_count_pf} and \Cref{eq:coset_constancy_pf} cannot be simultaneously satisfied.

\bigskip
\noindent\textbf{Construction of projective measurements required in the sufficient conditions.}

Let $u_1, u_2, \cdots, u_{n-k} \in \mathbb{F}_2^{n-k}$ be the vectors obtained by scanning the ordered set $\{ x_T, z_T, x_{T-1}, z_{T-1}, \ldots, x_1, z_1 \}$ and removing any vector that lies in the span of the previously selected ones.
For each $t\in[T]$, the definition of $\Outl_{[t,T]}^\perp$ (\Cref{eq:L-perp}) implies $\operatorname{span} \{u_1, \ldots, u_{d_t}\} = \Outl_{[t,T]}^\perp$.
Define $J_t := \operatorname{span}\{u_1,\ldots,u_{n-k-t}\}$.
The number of cosets of $J_t$ is given by $2^{n-k}/2^{n-k-t} = 2^{t}$, and hence we can label the cosets of $J_t$ by elements of $\mathbb{F}_2^{t}$.
Also, define a projective measurement $P^{(t)} := \{P^{(t)}_{s'_t}\}_{s'_t \in \mathbb{F}_2^{t}}$ by
\begin{align}
  P^{(t)}_{s'_t} := \sum_{ s \in \text{$s'_t$-th coset of } J_t } \AssocPOVM{\Lambda}_s
  \qquad \forall s'_t \in \mathbb{F}_2^{t}.
\end{align}
Equivalently, the column-stochastic matrix $\mu^{(t)} := (\mu^{(t)}_{s'_t, s})_{s'_t, s}$ for $P^{(t)}$ (i.e., $P^{(t)}_{s'_t} = \sum_s \mu^{(t)}_{s'_t, s} \AssocPOVM{\Lambda}_s$) is
\begin{align}
  \mu^{(t)}_{s'_t, s} =
  \begin{cases}
    1, & \text{if $s$ is in the $s'_t$-th coset of $J_t$,}\\
    0, & \text{otherwise.}
  \end{cases}
\end{align}

The projective measurements $P^{(1)},\ldots,P^{(T)}$ satisfy the conditions required by \Cref{thm:suff_qubit_reduction}, as follows.
Since $\AssocPOVM{\Lambda}$ is a projective measurement, $\operatorname{rank} P^{(t)}_{s'_t} = |J_t| \cdot \operatorname{rank} \AssocPOVM{\Lambda}_s = 2^{n-k-t} \cdot 2^k = 2^{n-t}$ for all $s'_t \in \Outs_t$.
By definition of $J_t$, every coset of $J_t$ decomposes into two disjoint cosets of $J_{t+1}$, and hence $P^{(t)} \lcircarrow P^{(t+1)}$ for each $t\in[T-1]$.
More generally, every coset of $J_t$ decomposes into disjoint cosets of $J_{t'}$ for any $t' \geq t$, so each row of $\mu^{(t)}$ is constant on cosets of $J_{t'}$ for any $t' \ge t$.
Since $d_t \leq n-k-t$, there exists $t' \geq t$ with $d_t = n-k-t'$, and hence
\begin{align}
  J_{t'} = \operatorname{span}\{u_1,\ldots,u_{n-k-t'}\}
  = \operatorname{span}\{u_1,\ldots,u_{d_t}\}
  = \Outl_{[t,T]}^\perp.
\end{align}
Therefore each row of $\mu^{(t)}$ is constant on cosets of $\Outl_{[t,T]}^\perp$, and by \Cref{eq:coset_constancy_pf} the POVM $P^{(t)}$ satisfies the outcome no-signaling condition $\mathrm{A}_t \nrightarrow \mathrm{cl}$. This completes the proof.

\end{proof}

\end{document}